\documentclass[10pt]{article}
\usepackage{minitoc}
\usepackage[toc,page,header]{appendix}

\usepackage{amsfonts}
\usepackage{setspace} 
\usepackage{bm}
\usepackage{graphicx}
\usepackage{array}
\usepackage{color}
\usepackage{multirow}
\usepackage{mathrsfs}
\usepackage{tabularx,tabulary}
\usepackage{booktabs}
\usepackage{float}
\usepackage{bbold}
\usepackage{verbatim}
\usepackage{dsfont}
\usepackage{xcolor,colortbl}
\usepackage[sort]{natbib}
\usepackage[T1]{fontenc}
\usepackage{eurosym} 
\usepackage{titling}
\usepackage{lmodern}
\usepackage[normalem]{ulem}
\usepackage[fontsize=11.2bp]{fontsize}
\usepackage{bbm}
\usepackage{makecell}
\usepackage{hhline}
\usepackage[labelfont=bf]{caption}
\usepackage{hyperref} 
\DeclareGraphicsExtensions{.jpg,.eps,.pdf} 
\usepackage{enumerate}
\usepackage{amsmath,amssymb,amsfonts,amsthm} 
\newcolumntype{Y}[1]{>{\raggedright\arraybackslash}p{#1}}
\newcolumntype{Z}[1]{>{\centering\arraybackslash}p{#1}}

\newtheorem{theorem}{Theorem}

\newtheorem{corollary}{Corollary}

\newtheorem{definition}[theorem]{Definition}

\newtheorem{lemma}{Lemma}

\newtheorem{proposition}{Proposition}
\newtheorem{remark}[theorem]{Remark}

\theoremstyle{remark}

\newtheorem{assump}{Assumption}[section]

\newtheorem{innercustomgeneric}{\customgenericname}
\providecommand{\customgenericname}{}
\newcommand{\newcustomtheorem}[2]{%
  \newenvironment{#1}[1]
  {%
   \renewcommand\customgenericname{#2}%
   \renewcommand\theinnercustomgeneric{##1}%
   \innercustomgeneric
  }
  {\endinnercustomgeneric}
}

\newcustomtheorem{customthm}{Theorem}
\newcustomtheorem{customlemma}{Lemma}
\newcustomtheorem{customassump}{Assumption}
\newcustomtheorem{customcorr}{Corollary}

\definecolor{bleu}{cmyk}{1,0.4,0,0}
\definecolor{bleu2}{cmyk}{0.75,0.55,0,0}
\definecolor{orange}{cmyk}{0,0.56,0.88,0}
\newcommand{\N}{\mathbb{N}}
\newcommand{\E}{\mathbb{E}}

\newcommand{\R}{\mathbb{R}}

\newcommand{\C}{\mathbb{C}}

\newcolumntype{a}{>{\columncolor{Gray}}c}
\newcommand{\diag}{\textnormal{diag}}
\newcommand{\tr}{\textnormal{tr}}

\definecolor{Gray}{gray}{0.95}
\definecolor{White}{gray}{1.00}
\definecolor{lightgray}{gray}{0.9}
\newcommand*\dd{\mathop{}\!\mathrm{d}}

\usepackage{xcolor}
\hypersetup{
citebordercolor=white,
filebordercolor=red,
linkbordercolor=blue
}
\hypersetup{
colorlinks,
citecolor=blue,
linkcolor=red,
urlcolor=blue}

\definecolor{green}{rgb}{0.2,0.7,0.3}
\usepackage{palatino}
\fontfamily{ppl}

\usepackage{xcolor}
\hypersetup{
citebordercolor=white,
filebordercolor=green,
linkbordercolor=white
}
\hypersetup{
colorlinks,
citecolor=green,
linkcolor=bleu2,
urlcolor=bleu2
}

\providecommand{\keywords}[1]{\textbf{\textit{Keywords---}} #1}

\DeclareUnicodeCharacter{2212}{-}

\makeatletter
\def\@fnsymbol#1{\ensuremath{\ifcase#1\or *\or \text{a} \or \text{b}\or
   \text{c}\or \dagger\or \|\or **\or \dagger\dagger
   \or \ddagger\ddagger \else\@ctrerr\fi}}
\makeatother

\title{Overparametrized models with posterior drift\thanks{The authors would like to thank the participants of the 18th Financial Risks International Forum, the 7th QFFE event, the 29th Forecasting Financial Markets Conference, the 19th BiGSEM Doctoral Workshop, the Scientific Beta Alumni Global Summit, the 2025 AFFI conference, and the SAF lab seminar for their helpful comments.}}

\author{Guillaume Coqueret\thanks{EMLYON Business School, 144 avenue Jean Jaurès, 69007 Lyon, FRANCE. E-mail: coqueret@em-lyon.com} \and Martial Laguerre \thanks{EMLYON Business School, 144 avenue Jean Jaurès, 69007 Lyon, FRANCE. E-mail: laguerre@em-lyon.com}\thanksgap{0.5ex} \thanks{Université Claude Bernard Lyon 1, 43 Boulevard du 11 Novembre 1918, 69100 Villeurbanne, France.}\thanksgap{0.5ex} \thanks{Corresponding author.}}

\usepackage{setspace}
\begin{document}
\maketitle

\begin{abstract}

This paper investigates the impact of posterior drift on out-of-sample forecasting accuracy in overparametrized machine learning models. We document the loss in performance when the loadings of the data generating process change between the training and testing samples. This matters crucially in settings in which regime changes are likely to occur, as can be the case in financial markets. Applied to equity premium forecasting, our results underline the sensitivity of a market timing strategy to sub-periods and to the bandwidth parameters that control the complexity of the model. For the average investor, we find that focusing on holding periods of 15 years can generate very heterogeneous returns, especially for small bandwidths. Large bandwidths yield much more consistent outcomes, but are far less appealing from a risk-adjusted return standpoint. All in all, our findings tend to recommend cautiousness when resorting to large linear models for stock market predictions.  
\end{abstract}

\keywords{Overparameterization, Equity Premium, Benign Overfiting, Random Matrix Theory, Market Timing.}


\clearpage

\doparttoc 
\faketableofcontents 


\section{Introduction}

Recently, the literature in machine learning (ML) has investigated the notion of ``\textit{double descent}'', whereby linear overparametrized models (i.e., with more parameters than observations) can have surprising out-of-sample benefits.\footnote{This is sometimes also referred to as ``\textit{benign overfitting}'' (\cite{bartlett2020benign}), or ``\textit{grokking}'' (\cite{power2022grokking}, \cite{varma2023explaining}), though these notions do not necessarily perfectly overlap.} As is customary in most machine learning contributions, including on double descent, a key assumption to derive analytical results is the invariance in distributions between the training and testing phases: the data generating process (DGP) is assumed to remain the same once the model has been calibrated. Unfortunately, this may not be the case in practice, especially in financial markets which can be subject to sudden regime changes. If we write $P_{y,X}=P_{y | X}P_X$ for the joint law of $X$ and $y$, we see that there are potentially two drivers for the change of distribution in the DGP: the conditional one, $P_{y | X}$, and the unconditional one, $P_X$. When the law of $X$ changes, we refer to \textit{covariate shift}, whereas when the link between $X$ and $y$ evolves, we will talk of \textit{posterior drift} (sometimes also referred to as \textit{concept drift}, see, e.g. \cite{gama2014survey} and \cite{lu2018learning}). Interestingly, the notion of posterior drift has also witnessed a surge in interest lately, in particular from the ML community (see, e.g., \cite{maity2024linear}, \cite{hu2025transfer} and \cite{wang2025conformal}).

The goal of the present paper is to investigate how posterior drift can be detrimental to out-of-sample forecasting accuracy in the context of overparametrized models.\footnote{The topic of covariate shift is already partly covered in Section 5 of \cite{hastie2022surprises} (via what they call the misspecified case), so that the room for novel contributions on the matter is limited.} Our practical motivation originates in the seminal paper by \cite{kelly2023virtue} which argues that the performance of aggregate market timing with linear models increases with the number of parameters. Since complexity is the ratio between the number of parameters and the sample size, the authors conclude that complexity is virtuous for equity premium prediction. However, several contributions have since then challenged this point of view, all from different angles. For instance, \cite{berk2023comment} underlines the lack of financial grounding (equilibrium consistency), whereas \cite{cartea2025limited} argue that the effect of noise is underestimated in \cite{kelly2023virtue}. Indeed, as the number of factors (i.e., predictors) increases, the amount of noise may also increase, thereby impairing the models' accuracy. \cite{nagel2025seemingly} demonstrates that the random Fourier feature (RFF) trick used to artificially increase the number of predictors in fact collapses to a low-complexity kernel ridge regression and the strategy boils down to volatility-timed trend-following. Resorting to RFFs is also problematic because they often require in-sample rescaling that violates the shift-invariance property required in theoretical results (\cite{fallahgoul2025high}). 
Finally, \cite{buncic2025simplified} and \cite{guo2025simplicity} point to another important shortcoming: adding a constant in the set of predictors reverses the pattern and performance then decreases with complexity. As a consequence, it turns out that Sharpe ratios obtained with low dimensional predictors are much higher than those reported in \cite{kelly2023virtue} for large levels of complexity. 

In focusing on the perils of regime changes, we shed light on another explanation of the ambiguous role of overparametrization in the forecasting efficiency of linear models. Indeed, we contend that the efficacy of complex models in equity premium prediction can be severely jeopardized by changing economic environments. In practice, the relationships that are inferred during the estimation phase may change due to unpredictable shocks, and, in this case, the out-of-sample precision of predictions can be substantially attenuated. Another representation of such phenomena can be made through the lens of the signal-to-noise ratio. In the presence of posterior drift, the information from predictors wanes, and, when signals are too weak, recent results indicate that ridgeless estimators perform worse than models that ignore the data completely (see Theorem 2 in \cite{shen2024can} and Corollary 1 in \cite{fallahgoul2025high} for instance).

Our contributions are twofold. First, in Section \ref{sec:theo}, we extend the misspecified isotropic results of \cite{hastie2022surprises} to the case of non-i.i.d. data and derive the expected return of a market-timing strategy when \textit{posterior drift} is taken into account. Our theoretical predictions span a large range of cases, as summarized in Table \ref{tab:summary} below. In each of the three panels, the first rows pertain to the main contributions kept in the text, while more peripheral propositions are postponed to the Appendix. In Section \ref{sec:numerical}, we corroborate our theoretical findings with Monte-Carlo simulations.

\begin{table}[ht]
\centering
\smallskip
\renewcommand{\arraystretch}{1.35}
\footnotesize
\begin{tabular}{@{} Z{2cm} Y{2.2cm} Z{2.5cm} Z{2.9cm} Z{1.9cm} Z{2.9cm} @{}}
\toprule
\textbf{Proposition} & \textbf{Quantity} & \textbf{Model} & \textbf{Feature correlation} & \textbf{Regularization} & \textbf{Location in paper} \\
\midrule
\multicolumn{6}{@{}l}{\textbf{Panel A. Expected return of the strategy}} \\
\midrule
2 & $\mathbb{E}[r^{(s)}_{t+1}|X]$ & Misspecified & General & Ridge & Body, Eq. \eqref{eq:oosexpretnonasymp}\\
4 & $\mathbb{E}[r^{(s)}_{t+1}|X]$ & Misspecified & IID & Ridge & Body, Eq. \eqref{eq:prop3} \\ \midrule
7 & $\mathbb{E}[r^{(s)}_{t+1}|X]$ & Well-specified & General & Ridge & App. \ref{subsub:wellspecifiedridge}, Eq. \eqref{eq:E_prop7} \\
9 & $\mathbb{E}[r^{(s)}_{t+1}|X]$ & Well-specified & IID & Ridge & App. \ref{subsub:wellspecifiedridge}, Eq. \eqref{eq:E_prop9} \\
11 & $\mathbb{E}[r^{(s)}_{t+1}|X]$ & Well-specified & General & $z\!\to\!0$ & App. \ref{subsub:wellspecifiedminnorm}, Eq. \eqref{eq:E_prop11} \\
12 & $\mathbb{E}[r^{(s)}_{t+1}|X]$ & Well-specified & IID & $z\!\to\!0$ & App. \ref{subsub:wellspecifiedminnorm}, Eq. \eqref{eq:E_prop12}\\
13 & $\mathbb{E}[r^{(s,w)}|X]$& Well-specified & General & Both & App. \ref{subsub:misspecifiedgeneral}, Eq. \eqref{eq:E_prop13} \\
14 & $\mathbb{E}[r^{(s)}_{t+1}|X]$ & Misspecified & General & Both & App. \ref{subsub:misspecifiedgeneral}, Eq. \eqref{eq:asyretgeneral} \\
\midrule
\multicolumn{6}{@{}l}{\textbf{Panel B. Second moment / Variance of the strategy}} \\
\midrule
3 & $\mathbb{V}[r^{(s)}_{t+1}|X]$ & Misspecified & General & Both & Body, Eq. \eqref{eq:var_prop3} \\
5 & $\mathbb{V}[r^{(s)}_{t+1}|X]$ & Misspecified & IID & Ridge & Body, Eq. \eqref{eq:var_prop5} \\ \midrule
15 & $\mathbb{E}[(r^{(s)}_{t+1})^2|X]$ & Misspecified & General & Both & App. \ref{subsub:misspecifiedgeneral}, \\
16 & $\mathbb{E}[(r^{(s)}_{t+1})^2|X]$ & Misspecified & IID & Ridge & App. \ref{subsub:misspecifiedisotropic} \\
\midrule
\multicolumn{6}{@{}l}{\textbf{Panel C. Sharpe ratio}} \\
\midrule
Eq.~(11) & $SR(z;c\phi)$ & Misspecified & IID & Ridge & Body, Eq.\eqref{eq:sharpe_ratio_iid} \\
\bottomrule
\end{tabular}
\caption{\textbf{Summary of the main theoretical results in the paper}. \small All limits are taken in the proportional regime $n, p, q \to \infty$ with $p/n \to c\phi$ and $(p+q)/n \to c$. \label{tab:summary}}
\end{table}

Second, in Section \ref{sec:evidence}, we empirically confirm the intuitions. We first provide in Section \ref{sec:betas} new evidence of time-variation in links between predictors and aggregate returns (betas). We then replicate in Section \ref{sec:sensi} the study of \cite{kelly2023virtue} with some subtleties. We test their procedure over several sub-periods of 15 years (the horizon for a representative investor) and across a range of bandwidth parameters, which we define below. 

Our empirical results indicate strong discrepancies in both dimensions (periods and bandwidths). Plainly speaking, this means that even with the bandwidth chosen in \cite{kelly2023virtue}, the performance can be either well above (e.g., +7\% monthly in 2005-2019) or well below (+0.5\% in 1975-1989) than the one that we report for the full sample (+4\%). It is possible to reduce this uncertainty by increasing the bandwidth value, but in this case, the average return decreases invariably towards zero, making complexity much less appealing. 

In sum, while it is likely that sophistication \textit{can} bring value in forecasting models, our findings suggest that analysts who rely on overparametrized models should pay particular attention to the stability and sensitivity of performance in their backtests. \\

\textbf{Notations.} $n,p \in \N_{>0}$ and $q \in \N_{\geq 0}$ are large dimensional parameters, possibly tending to infinity. $n$ is the number of observations and $p+q$ the number of independent variables in linear models. We use $C$ and $D$ (resp. $\tau$) for arbitrary large (resp. small) positive constants. Let $\langle \cdot, \cdot \rangle$ stand for the scalar product, i.e. for any vectors $u$, $v \in \R^p$, $\langle u, v \rangle = u'v$. Denote by $\| \cdot \| = \| \cdot \|_2$ the Euclidean norm of vectors. For any matrix $A \in \R^{p \times p}$ and any vector $v \in \R^p$, we denote by $\| v \|_A = \sqrt{v'Av}$, a reweighted version of the Euclidean norm. Similarly, for any matrix $A \in \R^{p \times p}$ and any vector $v,u \in \R^p$, define $\langle u, v \rangle_A = u'Av$. We use $\| \cdot \|_{\text{op}}$ to denote the Euclidean $2$-norm of a matrix, namely, for any matrix $A \in \R^{n \times p}$, $\| A \|_{\text{op}}$ is the largest singular value of $A$. 

\section{Theoretical grounding}
\label{sec:theo}

\subsection{Setup: misspecification \textit{and} posterior drift}
\label{subsec:baseline}

This paper pertains to large models which, as in most of the recent literature, are \textit{linear}, see, e.g., \cite{bartlett2020benign}, and \cite{kelly2023virtue}. \cite{hastie2022surprises} do cover some non-linear cases but these cases require very lengthy proofs which is why the present paper focuses on extending linear combinations of features. Henceforth, our modelling framework follows that of \cite{hastie2022surprises} closely. We assume a data generating process (DGP) of the form
\begin{align}
\label{eq:dgpmis}
\begin{split}
    ((x_i,w_i),e_i) &\sim P_{x,w} \times  P_e,  \quad i=1,\dots,n \\
    y_i &= x'_i\beta+ w'_i\theta + e_i,  \quad i=1,\dots,n, 
\end{split}
\end{align}
where the $n$ random draws are independent. $P_{x,w}$ is a distribution on $\R^{p+q}$ with $\E[x_i]=0$ and $\E[w_i]=0$ and 
\begin{align*}
    \mathbb{C}\text{ov}[(x_i,w_i)]=\Sigma=
    \begin{bmatrix}
        \Sigma_x & \Sigma_{xw} \\
        \Sigma'_{xw} & \Sigma_w
    \end{bmatrix}. \label{eq:covmis}
\end{align*}
$P_e$ is a distribution on $\R$ with zero mean and variance $\sigma^2$. In our empirical application, as in \cite{kelly2023virtue}, $y$ will be the equity premium and $x$ a set of macroeconomic predictors that are expected to have some predictive power over the premium.

The agent has access to the $x_i$ but not to the $w_i$, and hence only sees an incomplete picture, this is why her model is labeled as \textit{misspecified}. This corresponds to the assumptions in Section 5 of \cite{hastie2022surprises} and, to some extent, to Section IV of \cite{kelly2023virtue}. We write
$(y_i, x_i)\in \R \times \R^p  $ for one observation in the training data, and $(y,X)$ for the corresponding aggregated vector (in $\R^n$) and matrix (in $\R^{n\times p}$). 

As in \cite{hastie2022surprises}, we define out-of-sample prediction risk of the well specified model, i.e. the model in which all the covariates are observed, as
\begin{equation}
\label{eq:R_X}
    R_X(\hat{\beta},\beta) := \E[(x_0'(\hat{\beta}-\beta))^2|X] = B_X(\hat{\beta}, \beta) + V_X(\hat{\beta}, \beta),
\end{equation}
where $x_0$ is a random draw of $P_x$ that is independent from the training data. In the equation, the \textbf{bias} is $B_X(\hat{\beta}, \beta) =\|\E[\hat{\beta}|X]-\beta \|_\Sigma^2$, and the \textbf{variance} is $V_X(\hat{\beta}, \beta)=\text{Tr}[\mathbb{C}\text{ov}(\hat{\beta}|X)\Sigma]$.

The main departure from existing results is that we now discriminate between the $(\beta',\theta')'$ vector that is used for the two data generating processes in Equation \eqref{eq:dgpmis}: one vector for the training data ($(\beta_{\text{is}}',\theta_{\text{is}}')'$) and another one for the testing data ($(\beta_{\text{oos}}',\theta_{\text{oos}}')'$). In doing so, we include \textbf{posterior drift} in the model: we postulate a change in the link between $y$ and $(x',w')'$. Empirical support for this assumption will be provided in Section \ref{sec:betas} below.

The estimator $\hat{\beta}$ is evaluated with the in-sample (is) or out-of-sample (oos) generated data: 
\begin{equation}
    \hat{\beta}_{u} := \hat{\beta}_{u}(z_n)= (X'X+ \alpha I)^{-1}X'(X\beta_{u}+W \theta_u+e)=n^{-1}(\hat{\Sigma}+z_nI)^{-1}X'(X\beta_{u}+W \theta_u+e),
    \label{eq:estimator}
\end{equation}
where $\alpha > 0$ and $u \in \{\text{is}, \text{oos} \}$.
Crucially, the above estimator relies on shrinkage with intensity $z_n = \alpha / n$ towards the identity matrix. In the remainder of the paper, we omit the dependence of the regularization parameter on the sample size $n$ and set $z := z_n$. If $n<p$ and $z \to 0$, the pseudo-inverse is taken for the inverse matrix. We have also set $\hat{\Sigma}=X'X/n$. Below, we make explicit the terms for the quadratic error \eqref{eq:R_X} that is made when using the estimator $\hat{\beta}_{\text{is}}$ (from the data generated with $\beta_{\text{is}}$) when compared to the changed (or realized) $\beta_{\text{oos}}$. We are interested in this subsection in results for ridgeless estimators (when $z \to 0^+$).

\begin{proposition}
    
\label{lem:2}
    Under model \eqref{eq:dgpmis} and with $z \to 0$, the prediction risk of the misspecified model is 
\begin{align*}
    R_X^m(\hat{\beta}_{\text{is}},\beta_{\text{oos}}, \theta_{\text{oos}}) & = \E[(x_0'(\hat{\beta}_{\text{is}}-\beta_{\text{oos}})-w_0'\theta_{\text{oos}})^2|X] \\
    & = R_X(\hat{\beta}_{\text{is}},\beta_{\text{oos}}) + M(\theta_{\text{oos}}),
\end{align*}
where $M(\theta_{\text{oos}}) =\E[(w_0'\theta_{\text{oos}}-\E[w_0'\theta_{\text{oos}}|x_0])^2]$ is the \textit{misspecification bias}, i.e., the signal embedded in the unobserved features (following the notation of \cite{hastie2022surprises}). Moreover, in the isotropic case $\Sigma = I$,
\begin{align}
    \underset{n,p \rightarrow \infty}{\lim}R_X^m(\hat{\beta}_{\text{is}},\beta_{\text{oos}},\theta_{oos}) &= \left\{ \begin{array}{c l}
    (\sigma^2 + \| \theta_{is} \|_2^2) \frac{c}{1-c} + \|{\beta}_{\text{oos}}-{\beta}_{\text{is}}\|_2^2 + \| \theta_{oos} \|_2^2 & \quad p/n \rightarrow c < 1 \\
    (\sigma^2 + \| \theta_{is} \|_2^2)(c-1)^{-1} +c^{-1}\|\beta_{\text{is}}-\beta_{\text{oos}} \|^2_2 &\quad p/n \rightarrow c > 1  \\+ (1-c^{-1}) \|\beta_{\text{oos}} \|_2^2 + \| \theta_{oos} \|_2^2 .
    \end{array} \right. 
\end{align}
\end{proposition}
The proof of the proposition is postponed to Appendix \ref{app:proof_l2}. The limiting expressions generalize Theorem 4 in \cite{hastie2022surprises} by adding the posterior drift component. Plainly, both expressions for the risk are increasing in the \textbf{misspecification risk} $\| \theta_{oos} \|_2^2$ due to the unobserved features \textit{and} in the \textbf{drift risk} $\|\beta_{\text{is}}-\beta_{\text{oos}} \|_2^2$ that comes from the change in loadings.

\subsection{Portfolio performance under posterior drift}
\label{subsec:performance}

Proposition \ref{lem:2} quantifies the impact of the presence of posterior drift on the risk of out-of-sample prediction of linear (ridge) regressions in the over-parameterized setting. We now focus on the impact on the financial performance of a market timing strategy. As in \cite{kelly2023virtue}, we analyze the returns of the strategy which consists of adjusting the position in an asset according to its estimated expected returns. Formally, the strategy's return can be written as 

\begin{equation}
\label{eq:strat}
    r^{(s)}_{t+1}(z) = \hat{\pi}_t(z) r_{t+1},
\end{equation}
where $\hat{\pi}_t(z) = \sum_{k=1}^p \hat{\beta}_{\text{is}}^{(k)} (z)x_t^{(k)}$ is the adjusted position in the asset and $r_{t+1}$ is its realized monthly return (the superscript $^{(s)}$ in the notation stands for `\textit{strategy}'). The estimated loadings in this case can be shrunk as in the general case in Equation \eqref{eq:estimator}.
For brevity, we directly consider the misspecified framework described in \eqref{eq:dgpmis} and without loss of generality, we assume $\sigma^2 = 1$. We provide additional results within the well-specified framework in Appendix \ref{subsub:wellspecifiedridge} and \ref{subsub:wellspecifiedminnorm}.

The following assumption establishes the linear model from which the covariates are considered to be drawn.

\begin{assump}
\label{ass:iddistcov}
    The covariates have the form $x_i = \Sigma_x^{1/2} z_i$ where $z_i = (z_{i1}, z_{i2}, \dots, z_{ip})$ is a vector of latent features having independent and identically distributed entries. Moreover, for all $j \in \{1,\dots,p\}$, $\E[z_{ij}] = 0$, $\E[z_{ij}^2] = 1$ and $\E[|z_{ij}|^{4+a}] \leq C < \infty$ for $a > 0$.
\end{assump} 

Denote by $r^{(s,w)}_{t+1}(z)$ the return of the strategy if the model is well-specified, i.e. there are no unobserved features, $w_i$, in the DGP \eqref{eq:dgpmis}. Heed the superscript "\textit{w}" which stands for "\textit{well-specified}". The out-of-sample expected return of the timing strategy under posterior drift $r^{(s)}_{t+1}$ can then be decomposed as follows. 

\begin{proposition}
\label{prop:expretmis}
Let $z > 0$, $\xi = (e_{1}, e_{2}, \dots , e_{n})'$ with $(e_i)_i$ as in \eqref{eq:dgpmis} and Assumption (\ref{ass:iddistcov}) hold. Then,

    \begin{align}
    \label{eq:oosexpretnonasymp}
        \mathbb{E}\left[ r^{(s)}_{t+1}(z) \big| X  \right]  = \mathbb{E}\left[ r^{(s,w)}_{t+1}(z) \big| X  \right] + \mathcal{M}(z)
    \end{align}

    where 

    \begin{align*}
        \mathbb{E}\left[ r^{(s,w)}_{t+1}(z) \big| X  \right] &= \beta_{\text{oos}}' \Sigma_x (zI + \hat{\Sigma}_x)^{-1} \big( \hat{\Sigma}_x \beta_{\text{is}} + \frac{1}{n} X' \xi \big) \\
        \mathcal{M}(z) & = \big( \beta_{\text{oos}}' \Sigma_x + \theta_{\text{oos}}' \Sigma_{xw}' \big) (zI + \hat{\Sigma}_x)^{-1} \E \Big[\hat{\Sigma}_{xw}\Big| X \Big] \theta_{\text{is}} \\
        & \quad + \theta_{\text{oos}}' \Sigma_{xw}' (zI + \hat{\Sigma}_x)^{-1} \big( \hat{\Sigma}_x \beta_{\text{is}} +  \frac{1}{n} X' \xi \big) .
    \end{align*}

\end{proposition}

The proof of the proposition is postponed to Section \ref{subsub:misspecifiedgeneral} in the Appendix and simulations for verification are provided in Section \ref{sec:numerical} in particular cases. The proposition shows that the out-of-sample expected return of the timing strategy under posterior drift crucially depends on the interactions between the observed and the unobserved covariates. These complex relations are encapsulated in the quantity $\mathcal{M}(z)$. 

\subsubsection{A general model: correlated features}

In this section, we consider a general type of covariance structure for the covariates. Namely, we relate the unobserved features to the $x_i$s through the following assumption. 


\begin{assump}
\label{ass:unobscov}
    The unobserved covariates have the form $w_i = \Sigma_w^{1/2} P z_i$ where $P \in \R^{q \times p}$ is deterministic. 
\end{assump}

Next, we define the following dimensional ratios with their posited limiting behavior.  

\begin{assump}
\label{ass:complexity}
    For some $\tau > 0$, $\tau \leq c \phi \leq \tau^{-1}$, with $c = \lim_{p,q,n \to \infty}(p+q)/n$ and $\phi = \lim_{p,q \to \infty}p/(p+q) \in (0,1]$. 
\end{assump}

This assumption is standard in statistical theory and characterizes the proportional regime in which our study takes place. The constant $c$ stands for the complexity of the true model while $\phi$ is the misspecification ratio. The product $c\phi$ is the complexity of the misspecified model. 

Our most general result on the out-of-sample expected return of the timing strategy is, in all fairness, impossible to interpret. The formula is extremely intricate when unobserved covariates are introduced in the model, that is, when $\mathcal{M}(z)$ is not null. In this case, few stylized facts stand out, so to avoid over-burdening the main text, we have decided to postpone it to the Appendix - see Proposition \ref{prop:expretgeneral}. Fortunately, the expression becomes fairly compact and instructive when the observed and unobserved features are uncorrelated. This is discussed below of Proposition \ref{prop:expretdriftisomis} in Section \ref{subsubsec:iid}.


In addition to raw performance (returns), the \textit{risks} of financial strategies are of utmost interest to practitioners. Thus, in the following proposition, we establish the (squared) volatility of such a timing strategy that we define as the second moment of the strategy's return. Before stating the proposition, note that, under Assumptions \ref{ass:iddistcov}-\ref{ass:unobscov}, the DGP in \eqref{eq:dgpmis} can be rewritten as

\begin{align*}
\begin{split}
    (z_i,e_i) &\sim P_{z} \times  P_e,  \quad i=1,\dots,n \\
    y_i &= z_i'\omega_u + e_i,  \quad i=1,\dots,n, 
\end{split}
\end{align*}

where $\omega_u := \Sigma_x^{1/2} \beta_u + P' \Sigma_w^{1/2} \theta_u$ with $u \in \{\text{is}, \text{oos}\}$. Hence, the original model can be cast into an equivalent sequence where the features are i.i.d. and the regression coefficients are linear transformations of the initial ones. Using the notations of the above DGP, we are ready to establish the limiting variance of the strategy's return.   

\begin{proposition}
\label{prop:volstratret}
    Let $z>0$ and Assumptions (\ref{ass:iddistcov})-(\ref{ass:complexity}) and (\ref{ass:opbound})-(\ref{ass:esvd_diag_omega_oos}) hold and denote $m_4 = \E[z_i^4]$. Then, 
    \begin{align}
        \mathbb{V}\left[ r^{(s)}_{t+1}(z) \Big| X  \right] &\xrightarrow[\substack{\mathstrut n,p,q \rightarrow \infty \\ p/n \rightarrow c\phi }]{\mathds{P}} \quad (1+\|\omega_{oos}\|^2) \mathcal{L}(z) + \big( \mathcal{E}^{(d)}(z;c\phi,\iota_d)\big)^2 + (m_4-3)\mathcal{K}(z), \label{eq:var_prop3}
    \end{align}
    where
    \begin{itemize}
        \item $\mathcal{E}^{(d)}(z;c\phi,\iota_d)$ is the asymptotic average return of the strategy and is given in Eq. \eqref{eq:asyretgeneral};
        \item $\mathcal{L}(z)$ is the (limiting) leverage of the strategy, that is $\E[|\hat{\pi}_t(z)|^2 |X]  \to \mathcal{L}(z)$ in probability when $n,p,q \to \infty$ such that $p/n \to c\phi$, and is defined in the Appendix in Eq. \eqref{eq:leverage_gen_ridge} for ridge regularization and in Eq. \eqref{eq:leverage_ridgeless} when the penalty level is vanishing. 
        \item $\|\omega_{oos}\|^2$ is the amount of signal contained in the out-of-sample data.
        \item $\mathcal{K}(z)$ is a weighted version of the leverage that accounts for posterior drift. 
        Its detailed expression is provided in Eq. \eqref{eq:K1_K2_gen_ridge}, while $\lim_{z \to 0^+} \mathcal{K}(z)$ is given in Eq. \eqref{eq:K1_K2_ridgeless}.
    \end{itemize}
\end{proposition}

The proof of Proposition \ref{prop:volstratret} is also located in Appendix \ref{subsub:misspecifiedgeneral}. The proposition shows that the volatility is linearly increasing in the leverage induced by the strategy. The impact of leverage increases when the signal contained in the out-of-sample equity premium, i.e. $\|\omega_{\text{oos}}\|^2$, is strong. For notational convenience, in the remainder of the article, we omit the dependence of the leverage on the penalty level $z$, that is $\mathcal{L} \equiv \mathcal{L}(z)$. The variance is also increasing in the magnitude of the strategy's expected return, i.e. $[ \mathcal{E}^{(d)}(z;c\phi,\iota_d)]^2$. Finally, the last term pertains to posterior drift but is hard to interpret. It can nonetheless be made more explicit in a simpler i.i.d. framework (see Section \ref{subsubsec:iid} below).

From Propositions \ref{prop:volstratret} (and  \ref{prop:expretgeneral}), it is of course possible to compute the (limiting) risk-adjusted return (Sharpe ratio) of the timing strategy. These results, because of their abstraction, carry little insight, unfortunately. In the next section, we analyze the simpler case when the covariates have i.i.d. coordinates, and this for two reasons. First, it improves interpretability by producing results that are more intuitively understandable. Second, and more importantly, in our empirical exercises, we resort to random Fourier features (RFFs). One implication of this transformation of data is that the predictors behave as if i.i.d.,\footnote{More precisely, instead of employing the raw features $x_i$ as the independent variables, we work with $\tilde{x}_i = \varphi(Wx_i)$ where $\varphi$ is a nonlinear function and $W$ is a matrix with i.i.d. random entries. As is demonstrated in \cite{hastie2022surprises}, under some mild conditions on $\varphi$, $\tilde{X}$ behaves as if it has i.i.d. entries.} which further justifies our focusing on this special case.



\subsubsection{The case of i.i.d. features}
\label{subsubsec:iid}

In what follows, we study the simpler case of i.i.d. features with unit variance, which entails $\Sigma=I$. Our next proposition gives the limiting out-of-sample expected return of the timing strategy ($\mathcal{E}^{(d)}(z;c\phi)$) under both posterior drift and i.i.d. features. Moreover, for the sake of comparison, we also introduce an important benchmark, namely the limiting zero-drift expected return from Proposition 3 in \cite{kelly2023virtue}:

\begin{equation}
    \mathcal{E}(z;c\phi) = \underset{\substack{\mathstrut n,p \rightarrow \infty \\ p/n \rightarrow c\phi }}{\lim} \ \mathbb{E}\left[ r^{(s)}_{t+1}(z) \big| X, \beta_{\text{is}}=\beta_{\text{oos}}\right]. 
    \label{eq:zd_return}
\end{equation}

\begin{proposition}
\label{prop:expretdriftisomis}
    Let $z>0$ and assume $(x_i,w_i)'$ has independent and identically distributed entries with unit variance and finite moments of order $4+a$, for some $a > 0$. Then, 
    \begin{equation}
        \mathbb{E}\left[ r^{(s)}_{t+1}(z) \big| X  \right] \xrightarrow[\substack{\mathstrut n,p \rightarrow \infty \\ p/n \rightarrow c\phi }]{\mathds{P}} \quad f(z;c\phi) \langle \beta_{\text{is}}, \beta_{\text{oos}} \rangle =: \mathcal{E}^{(d)}(z;c\phi), 
        \label{eq:prop3}
    \end{equation}
    where $f(z;c\phi)$ and $f(c\phi):=\lim_{z \to 0} f(z;c\phi)$ are defined in Equations \eqref{eq:f_isotropic} and \eqref{eq:isotropicridgeless} in the Appendix.\footnote{Propositions \ref{prop:nonasyboundsdrift} and \ref{prop:nonasyboundsdriftridgeless} in the Appendix provide non-asymptotic bounds (under stronger conditions) for the approximation of the expected return. However, these bounds are additive in the errors and, therefore, not relative to the scale of the return. Hence, given the typical order of magnitude of financial returns, it would make sense to establish multiplicative bounds similar to those developed by \cite{cheng2022dimension}, \cite{misiakiewicz2024non} or \cite{defilippis2024dimension} for the approximation of the test error of different types of ridge regression. We leave this avenue open for further research.}
In particular, if $\| \beta_{is} \|^2 = \| \beta_{oos} \|^2$, then
    \begin{equation*}
        \mathbb{E}\left[ r^{(s)}_{t+1}(z) \big| X  \right] \xrightarrow[\substack{\mathstrut n,p \rightarrow \infty \\ p/n \rightarrow c\phi }]{\mathds{P}} \quad \mathcal{E}^{(d)}(z;c\phi) = \mathcal{E}(z;c\phi) - \frac{1}{2} f(z;c\phi) \| \beta_{\text{is}} - \beta_{\text{oos}} \|^2,
    \end{equation*}
Furthermore, $z \mapsto f(z;c\phi) \langle \beta_{\text{is}}, \beta_{\text{oos}} \rangle$ is monotone decreasing (resp. increasing) in $z$ when \\ $\langle \beta_{\text{is}}, \beta_{\text{oos}} \rangle \geq 0$ (resp. $\langle \beta_{\text{is}}, \beta_{\text{oos}} \rangle \leq 0$).    
\end{proposition}

This proposition, the proof of which is located in Section \ref{subsub:misspecifiedisotropic}, calls for several remarks. First, the above expression generalizes the results of \cite{kelly2023virtue}. Indeed, the above formulation allows to recover Figure 5 in their theoretical section when setting $c=10$ and $\beta = \beta_{\text{is}} = \beta_{\text{oos}}$ with $\beta_i = \beta_j$ for any $i,j \in \{1, \dots, p \}$ and such that $\| \beta \|^2 = 0.2$ when $q=0$. The pattern exhibited in \cite{kelly2023virtue} is a special case when the scalar product $\langle \beta_{\text{is}}, \beta_{\text{oos}} \rangle $ is fixed to a particular value.

The shape in the r.h.s. of \eqref{eq:prop3} is a multiple of the scalar product $\langle \beta_{\text{is}}, \beta_{\text{oos}} \rangle $. Simply put, the expected performance is proportional to the alignment between the in-sample loadings and the out-of-sample ones. Given the fact that the mean of the estimated loadings is very close to zero (because of the Random Fourier transform), the scalar product can also be interpreted as the sample covariance between $\beta_{\text{is}}$ and $\beta_{\text{oos}}$. In particular, if the two vectors are orthogonal (with null correlation), then the average return is simply zero. This makes sense, as it is in this case impossible to extract any meaningful information from the training set.

The proportional form in \eqref{eq:prop3} is in fact very general and can also be obtained with a DGP that includes latent features, as in \cite{hastie2020surprises}. The corresponding results and proofs are postponed to Appendix \ref{subsub:misspecifiedlatent}.

\begin{remark}
The limiting expression in (\ref{eq:prop3}) can be adapted to any specific form of $\beta_{\text{oos}}$. For instance, let $\beta_{\text{oos}}$ be a linear transformation of $\beta_{\text{is}}$, i.e. $\beta_{\text{oos}} = W \beta_{\text{is}}$ with $W \in \R^{p \times p}$ being a fixed (deterministic) matrix. If we denote by $(\xi_k)_{1 \leq k \leq p}$ the eigenvalues of $W$ and $(v_i)_{1 \leq k \leq p}$ their associated eigenvectors, then the (asymptotic) expected return of the strategy converges to
\begin{align*}
    \mathcal{E}^{(d)}(z;c\phi) = h(z;c\phi,\beta_{\text{is}},W) \|\beta_{\text{is}}\|^2,
\end{align*}
where $h(z;c\phi,\beta_{\text{is}},W) := f(z;c\phi) \sum_{k=1}^p \xi_k \text{cosim}(\beta_{\text{is}}, v_k)^2$ and $ \text{cosim}(\beta_{\text{is}}, v_k) := \langle \beta_{is}, v_k \rangle / \|\beta_{\text{is}}\|$ is the cosine similarity between $\beta_{\text{is}}$ and $v_k$, which lies between $-1$ and $1$. In the Appendix, we also show that $f \in [0,1]$, thus it is straightforward that $\mathcal{E}^{(d)}(z;c\phi) \geq 0$ if $W$ is positive semidefinite. Otherwise, if $W$ is not positive semidefinite, the picture is more nuanced and the expected return can be negative.
\end{remark}

More importantly, when the covariates are isotropic, the misspecification term $\mathcal{M}(z)$ in \eqref{eq:oosexpretnonasymp} vanishes, and the out-of-sample expected return behaves as if the model was well-specified with a true complexity equal to $c\phi = p/n$. This result leads us to draw the following remark which establishes the conditions under which the presence of posterior drift deteriorates the expected returns of the market-timing strategy. 

\begin{remark}
\label{rem:condmis}
    Let $z>0$ and assume that $(x_i,w_i)'$ has independent and identically distributed entries with unit variance and finite moments of order $4+a$, for some $a > 0$. Then,

    \begin{equation*}
        \mathcal{E}^{(d)}(z;c\phi) < \mathcal{E}(z;c\phi)
    \end{equation*}

    if and only if 

    \begin{equation}
    \label{eq:condprod}
        \langle (\beta_{\text{oos}} - \beta_{\text{is}}), \beta_{\text{is}} \rangle < 0,
    \end{equation}

    or alternatively, if and only if,

    \begin{equation*}
        \| \beta_{\text{oos}} \|^2 - \| \beta_{\text{is}} \|^2 < \| \beta_{\text{is}} - \beta_{\text{oos}} \|^2.
    \end{equation*}
\end{remark}

Consequently, the loss of financial performance depends on the scalar product between the vector of \textit{in-sample} regression coefficients, $\beta_{\text{is}}$, and a \textit{residual} vector, $\beta_{\text{oos}} - \beta_{\text{is}}$, arising from the change in DGP. In particular, when the amount of signal contained in the first $p$ features is preserved or is less important \textit{out-of-sample} than \textit{in-sample}, i.e. $\| \beta_{\text{oos}} \|^2 \leq \| \beta_{\text{is}} \|^2$, Remark \ref{rem:condmis} implies that the presence of posterior drift not only increases the out-of-sample prediction risk (as documented in our first proposition), but, more interestingly from a financial standpoint, it also systematically deteriorates the returns of market timing strategies based on these predictions.

For the sake of completeness, we also provide the second (demeaned) moment of the strategy's return in the following proposition. 

\begin{proposition}
\label{prop:volstratret_iid}
    Let $z>0$ and Assumptions (\ref{ass:iddistcov})-(\ref{ass:complexity}) and (\ref{ass:opbound})-(\ref{ass:esvd_diag_omega_oos}) hold. Denote $m_4 = \E[z_i^4]$. Then, 
    \begin{align}
        \mathbb{V}\left[ r^{(s)}_{t+1}(z) \Big| X  \right] &\xrightarrow[\substack{\mathstrut n,p,q \rightarrow \infty \\ p/n \rightarrow c\phi }]{\mathds{P}}  (1+S_{oos}) \mathcal{L} + ( \hspace{0.4em} \underbrace{f(z;c\phi) \langle \beta_{oos}, \beta_{is} \rangle}_{=\mathcal{E}^{(d)}(z;c\phi)} \hspace{0.4em} )^2 + (m_4-3) f(z;c\phi) \|\beta_{is} \circ \beta_{oos}\|^2  \label{eq:var_prop5}
    \end{align}

    where $S_{oos} = \|\beta_{oos}\|^2+\|\theta_{oos}\|^2$ is the amount of signal contained in the out-of-sample data and $\mathcal{L}$ is the (limiting) leverage of the strategy, that is, $\E[|\hat{\pi}_t(z)|^2 |X] \to \mathcal{L}$ in probability when $n,p,q \to \infty$ such that $p/n \to c\phi$. The latter quantity is defined in the Appendix in Eq. \eqref{eq:leverage_ridge_iid} for ridge regularization and in Eq. \eqref{eq:leverage_ridgeless_iid} when the penalty level is vanishing. Recall that $f$ is defined in Eq. \eqref{eq:f_isotropic} for the ridge scenario and in Eq. \eqref{eq:isotropicridgeless} for the ridgeless case.
\end{proposition}


Here again, the volatility of the strategy is increasing in the leverage generated by the strategy. Interestingly, it is also increasing in the magnitude of the strategy's expected return $\mathcal{E}^{(d)}(z;c\phi)$. Said differently, in the i.i.d. scenario, the risk of the strategy is decreasing in the level of posterior drift, as embodied by the term $\langle \beta_{oos}, \beta_{is} \rangle^2$. Finally, the effect of the third term is harder to expound. For fixed parameters vectors $\beta_{is}$ and $\beta_{oos}$, the returns are all the more volatile when the kurtosis of the latent variables $z$ increases. On the other hand, when the kurtosis level $m_4$ is fixed to a level superior to $2$, the variance of the returns is increasing in another proxy for posterior drift: $\|\beta_{is} \circ \beta_{oos}\|^2$.\footnote{Indeed, note that this term can formulated as 
\begin{align*}
    \|\beta_{is} \circ \beta_{oos}\|^2 = \sum_i (\beta_{is})_i^2 (\beta_{oos})_i^2 = \langle \beta_{oos}, \beta_{is} \rangle^2 - \sum_{i\neq j} (\beta_{is})_i^2 (\beta_{oos})_j^2.
\end{align*}
In plain words, it is a quantification of the (squared) posterior drift, to which the interactions between the \textit{in-sample} and \textit{out-of-sample} coefficients corresponding to different directions are subtracted. In fact, because of the sphericity of the observed features covariance matrix ($\Sigma_x \propto I_p$), the scalar product of each eigenvector of $\Sigma_x$ with the parameter vectors depends \underline{only} on the coefficient of the considered direction. Recently, the effects of different geometrical settings on the learning task has thoroughly been explored (see e.g., \cite{wu2020optimal}, \cite{hastie2020surprises}, \cite{richards2021asymptotics}).} Quite surprisingly though, \textit{thin-tailedness} of the latent variables distribution ($m_4 < 3$) generates a surprising effect that decreases the volatility of the strategy when $\|\beta_{is} \circ \beta_{oos}\|^2$ increases.

Finally, by the continuous mapping theorem, we obtain an analytic formula for the Sharpe ratio of the market timing strategy in the large $n,p,q$ limit. That is, 
\begin{equation}
\begin{split}
    SR(z;c\phi) :&= \frac{f(z;c\phi) \langle \beta_{\text{is}}, \beta_{\text{oos}} \rangle}{\sqrt{(1+S_{oos}) \mathcal{L} + (f(z;c\phi) \langle \beta_{oos}, \beta_{is} \rangle)^2 + (m_4-3) f(z;c\phi) \|\beta_{is} \circ \beta_{oos}\|^2  }} \\
    &= \text{sign}(\langle \beta_{\text{is}}, \beta_{\text{oos}} \rangle)\Bigg[ 1 + (1+S_{oos})\frac{\mathcal{L}}{(\mathcal{E}^{(d)}(z;c\phi))^2} + \frac{m_4-3}{f(z;c\phi)} \cdot \Bigg(\frac{\|\beta_{is} \circ \beta_{oos}\|}{\langle \beta_{oos}, \beta_{is} \rangle} \Bigg)^2\Bigg]^{-1/2}.
\end{split}
\label{eq:sharpe_ratio_iid}
\end{equation}

A crucial quantity for the risk-adjusted return is the ratio between the asymptotic leverage and the squared expected return of the strategy. In particular, the higher the leverage relatively to the magnitude of the expected return, the weaker the Sharpe ratio of the strategy. Intuitively, the strategy will be all the less profitable on a risk-adjusted basis when the size of the position induced by the strategy is important compared to the potential gains.

The analytical formula for the Sharpe ratio allows us to delve into an essential matter: to what extent does the posterior drift affect the Sharpe ratio of the strategy along the complexity path? As evidenced in Figure \ref{fig:sr_signal}, the answer depends on the amount of signal encapsulated in the data. Indeed, when the \textit{in-sample} signal is strong (right graph), the Sharpe ratio is only mildly impacted by a changing DGP and the level of Sharpe ratio is reasonable (approximately that of the S\&P500 over long horizons).

\begin{figure}[!h]
\centering
    \includegraphics[width=16cm]{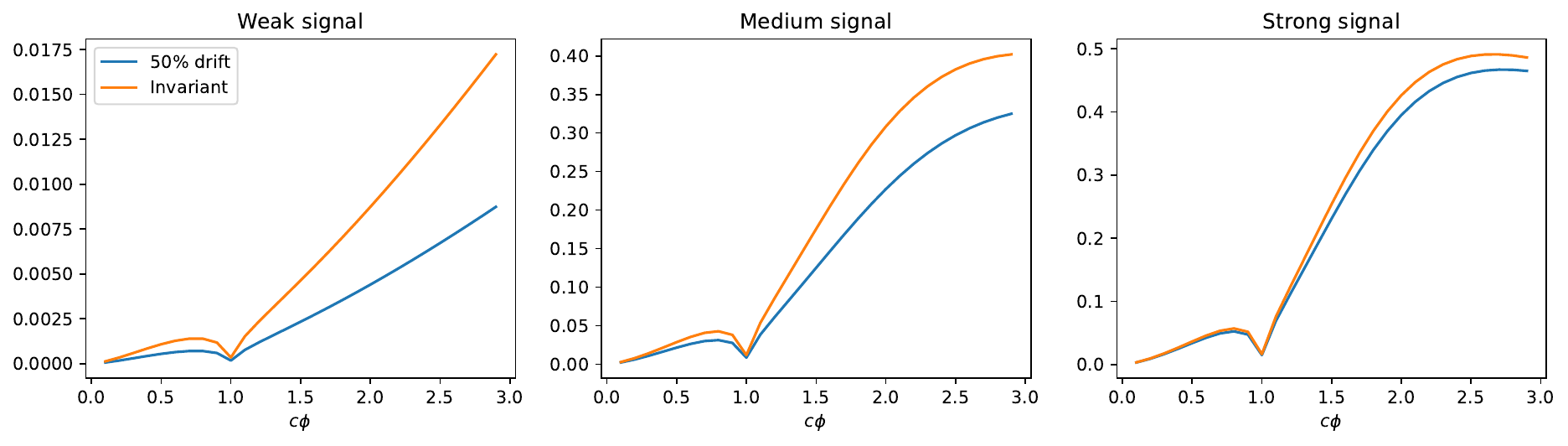}
\caption{\textbf{Sharpe ratio under different signal levels.} \small The model's true complexity is $c=3$ ($p+q = 3000$, $n=1000$). The regularization parameter of the ridge estimator is set to $z=10^{-5}$. We plot the theoretical Sharpe ratio of the strategy as per Eq. \eqref{eq:sharpe_ratio_iid} (solid line). From the leftmost to the rightmost graph, the \textit{in-sample} signal level ($\|\beta_{is}\|^2 + \|\theta_{is}\|^2$) amounts to $0.2$, $2$ and $5$, respectively.  \label{fig:sr_signal}}
\end{figure}

However, the gap increases and the Sharpe ratios shrink as the signal weakens (middle and left graphs). In a financial context, this observation is worrisome. Indeed, the weakness of the signal contained in financial returns is a stylized fact (see, e.g., \cite{shen2024can}). For instance, empirical $R^2$ values for (monthly) return forecasts based on macroeconomic predictors are most often below $5 \%$ (\cite{welch2008comprehensive}, \cite{ferreira2011forecasting}, \cite{nagel2025seemingly}). In addition, the loss in performance due to posterior drift is all the more pronounced in the overparametrized regime. We however note that the double ascent phenomenon still holds in the presence of posterior drift in the data.

In the next section, we verify numerically the conformity of the theoretical expressions for the moments of the market-timing strategy. 

\section{Simulations}
\label{sec:numerical}

The proofs of Propositions \ref{prop:expretmis} and \ref{prop:expretdriftisomis} are somewhat cumbersome. This section is dedicated to the verification of these theoretical results via Monte-Carlo simulations. To this end, we generate data as per the DGP of Equation \eqref{eq:dgpmis}. Let $p+q=300$, we draw $n=100$ samples of covariates from a $300$-dimensional standard multivariate normal distribution, i.e. $(x_i,w_i) \sim N(0,\Sigma)$ with $\Sigma = I_{300}$. \footnote{Extended numerical experiments are carried out under general structure for features' correlation (see Sections \ref{subsec:numanalysis_correlated_equidist} and \ref{subsec:numanalysis_correlated_conc} in the Appendix). All in all, simulated and theoretical moments are in line with stronger deviations observed when the parameter vector is concentrated in the top coefficients.} The error term is a white Gaussian noise, viz. $e_i \sim N(0,1)$. 

The \textit{in-sample} and \textit{out-of-sample} true parameters vectors for the observed data, $\beta_{\text{is}}$ and $\beta_{\text{oos}}$ respectively, are equally distributed on the eigenvectors of $\Sigma$, namely, on the canonical basis vectors. $\beta_{is}$ has unit signal, that is $\| \beta_{\text{is}} \|^2 = 1$, while we consider a sequence of $(\beta_{\text{oos},k})_k$ with decreasing signal. In particular, $\beta_{\text{is}} = p^{-1/2}(1,1,\dots,1,1)'$ and $\beta_{\text{oos},k} = k p^{-1/2}(1,1,\dots,1,1)'$ with $k \in \{ 0.2, 0.4, 0.6, 0.8, 1\}$, which yields $\| \beta_{\text{oos},k} \|^2 = k^2 \leq 1$ and $\langle \beta_{\text{is}}, \beta_{\text{oos}} \rangle = k$.

We analogously generate the true parameters vectors for the unobserved data, $\theta_{\text{is}}$ and $\theta_{\text{oos}}$, setting only a different number of coefficients, viz. $q$ instead of $p$. Thus, the signal of the model \textit{in-sample} is always equal to $2$ while it is set to $2k^2$ \textit{out-of-sample}, for $k \in \{ 0.2, 0.4, 0.6, 0.8, 1\}$. Finally, we simulate the performance of two market-timing strategies based on ridge estimators with shrinkage intensity $z \in \{0.01,0.1\}$.  

\begin{figure}[!h]
\centering
    \includegraphics[width=16cm]{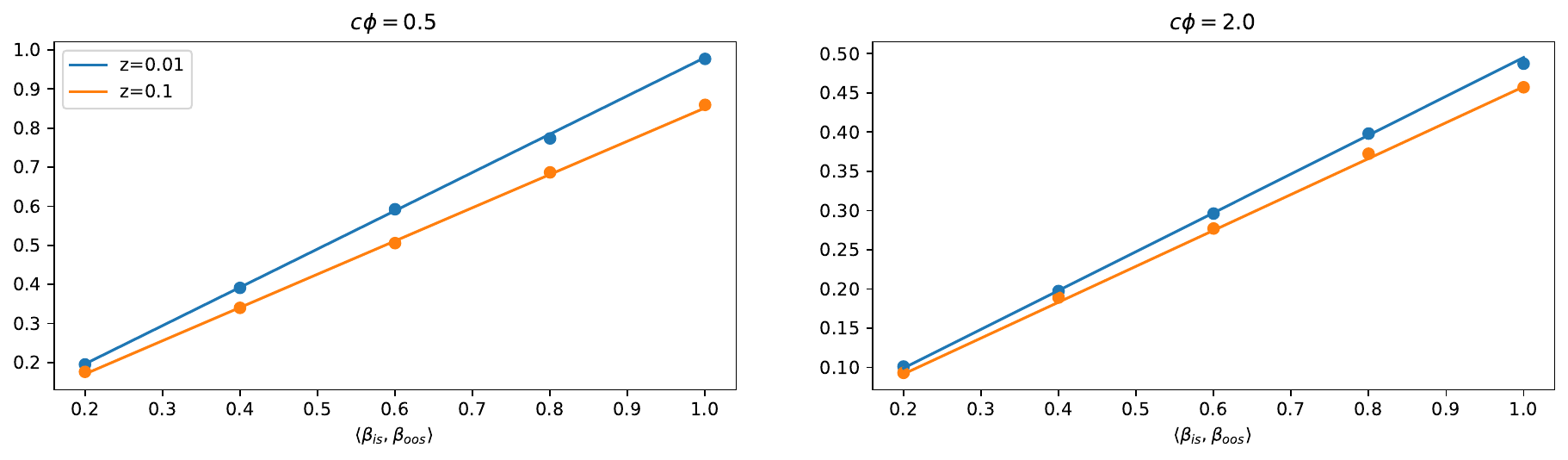}
\caption{\textbf{Simulated and theoretical returns of the strategy under mild regularization.} \small The model's true complexity is $c=3$ ($p+q = 300$, $n=100$). $z$ is the regularization parameter of the ridge estimator. We plot the theoretical returns of the strategy as per Proposition \ref{prop:expretdriftisomis} (solid line). Simulated returns (dots) are generated according to the process described in Section \ref{sec:numerical} and are averaged over 100,000 random draws. \label{fig:retmildprop}}
\end{figure}

As expected from Proposition \ref{prop:expretdriftisomis}, the strategy's expected return is linear in the scalar product of $\beta_{\text{is}}$ and $\beta_{\text{oos}}$. This is verified in Figure \ref{fig:retmildprop}. This pattern arises whether the model is misspecified and underparameterized (Fig. \ref{fig:retmildprop}, left), misspecified and overparameterized (Fig. \ref{fig:retmildprop}, center) or well-specified (Fig. \ref{fig:retmildprop}, right). Note that the rightmost point on each plot corresponds to the expected return in the no-drift scenario. 

We defer the comparison of simulated and theoretical expected moments of the market-timing strategy under strong regularization ($z \in \{10,100\}$) to Appendix \ref{subsec:simretstrongequi}. Once again, these experiments confirm the claims stated in Propositions \ref{prop:expretdriftisomis} and \ref{prop:volstratret_iid}. 

\begin{figure}[!h]
\centering
    \includegraphics[width=16cm]{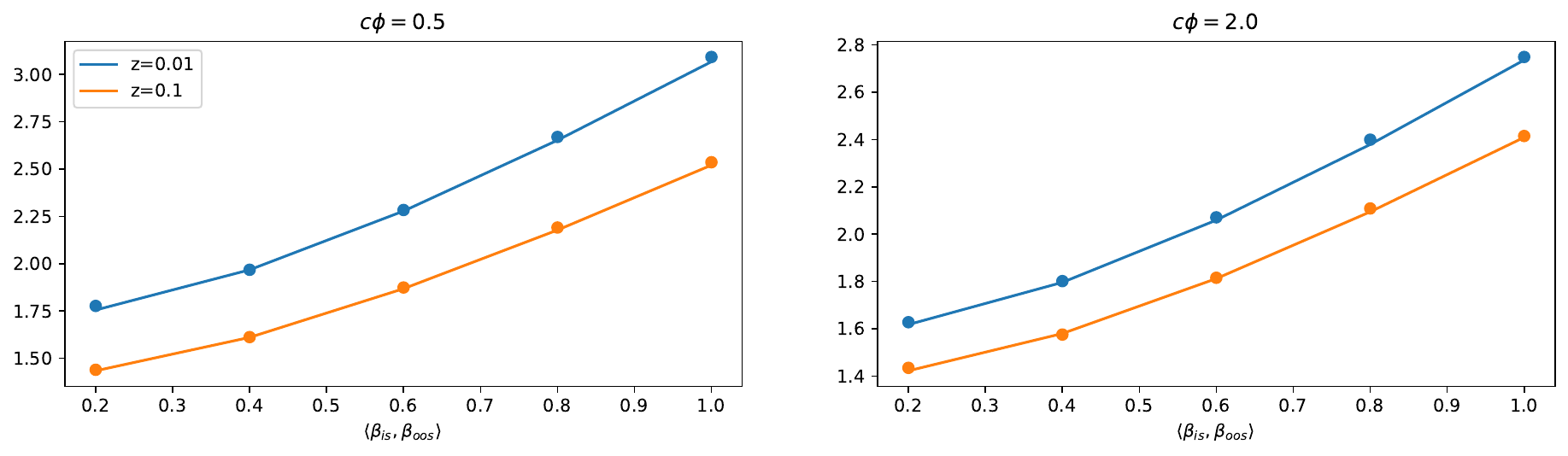}
\caption{\textbf{Simulated and theoretical volatility of the portfolio's return under mild regularization.} \small The model's true complexity is $c=3$ ($p+q = 300$, $n=100$). $z$ is the regularization parameter of the ridge estimator. We plot the theoretical volatility of the strategy's return as per Proposition \ref{prop:volstratret_iid} (solid line). Simulated returns (dots) are generated according to the process described in Section \ref{sec:numerical} and the (sample) standard deviation is computed over 100,000 random draws. \label{fig:volmildprop}}
\end{figure}

Figure \ref{fig:volmildprop} implies that the volatility of the portfolio's return is increasing and slightly convex in the scalar product of the in-sample and out-of-sample target coefficients vector. Furthermore, returns are less volatile when the size of the position is guided by a more regularized ridge regression (i.e., a higher $z$). Indeed, a higher penalty level reduces the amplitude of forecasts and thus of positions, thereby dampening return dispersion.

\begin{figure}[!h]
\centering
    \includegraphics[width=16cm]{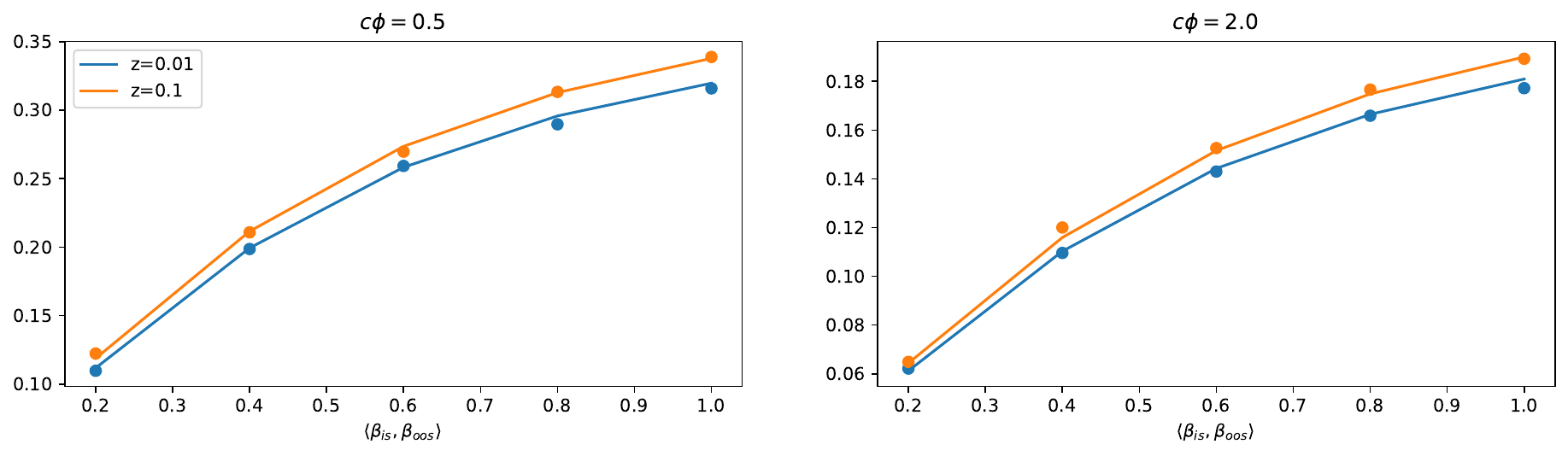}
\caption{\textbf{Simulated and theoretical risk-adjusted returns under mild regularization.} \small The model's true complexity is $c=3$ ($p+q = 300$, $n=100$). $z$ is the regularization parameter of the ridge estimator. We plot the theoretical Sharpe ratio of the strategy as per Eq. \eqref{eq:sharpe_ratio_iid} (solid line). Simulated returns (dots) are generated according to the process described in Section \ref{sec:numerical} and results are averaged over 100,000 random draws. \label{fig:srmildprop}}
\end{figure}

The simulated results confirm our theoretical predictions regarding the Sharpe ratio of the strategy. Higher risk-adjusted returns are obtained for stronger levels of penalization. This is contrary to the findings of \cite{kelly2023virtue} who find that min-norm interpolation yields a market-timing strategy that surpasses all the others in terms of Sharpe ratio. 

In Appendix \ref{subsec:simretconcentrated}, we perform the same analysis with true parameter vectors concentrated in some directions of $\Sigma$. As evidenced in Figures \ref{fig:retmildconc} and \ref{fig:retstrongconc}, the average return of the strategy is again closely aligned with the theory formulated in Eq. \eqref{eq:prop3}.

\section{Empirical evidence}
\label{sec:evidence}

\subsection{Data}
\label{sec:data}

We rely on the same data as \cite{kelly2023virtue}, which comes from \cite{goyal2023comprehensive} and is available on \href{https://sites.google.com/view/agoyal145\string#Amit Goyal's website}{Amit Goyal's website}. We use the CRSP value-weighted returns minus the T-bill rate as dependent variable. The predictors consist of the 14 macro-economic indicators used in \cite{kelly2023virtue}, to which we add lagged returns. 

In Section \ref{sec:betas}, we will compute betas from predictive regressions:
\begin{equation}
r_{t+1}= a +\sum_{k=1}^p\beta^{(k)}x_t^{(k)}+e_{t+1},
    \label{eq:predreg}
\end{equation}
where the $x_t^{(k)}$ are all $p$ predictors. We estimate $\hat{\beta}$ every month on rolling windows of 15 years (180 points), which is our reference horizon. This generates time-series of loadings $\hat{\beta}_t$ that characterize the predictive link between the predictors and the returns from the CRSP portfolio $r_{t+1}$. These series are plotted in Figure \ref{fig:betasTS} below and therein we have included lagged returns up to order 6. 

\begin{figure}[!h]
\centering
    \hspace*{-1cm}\includegraphics[width=16cm]{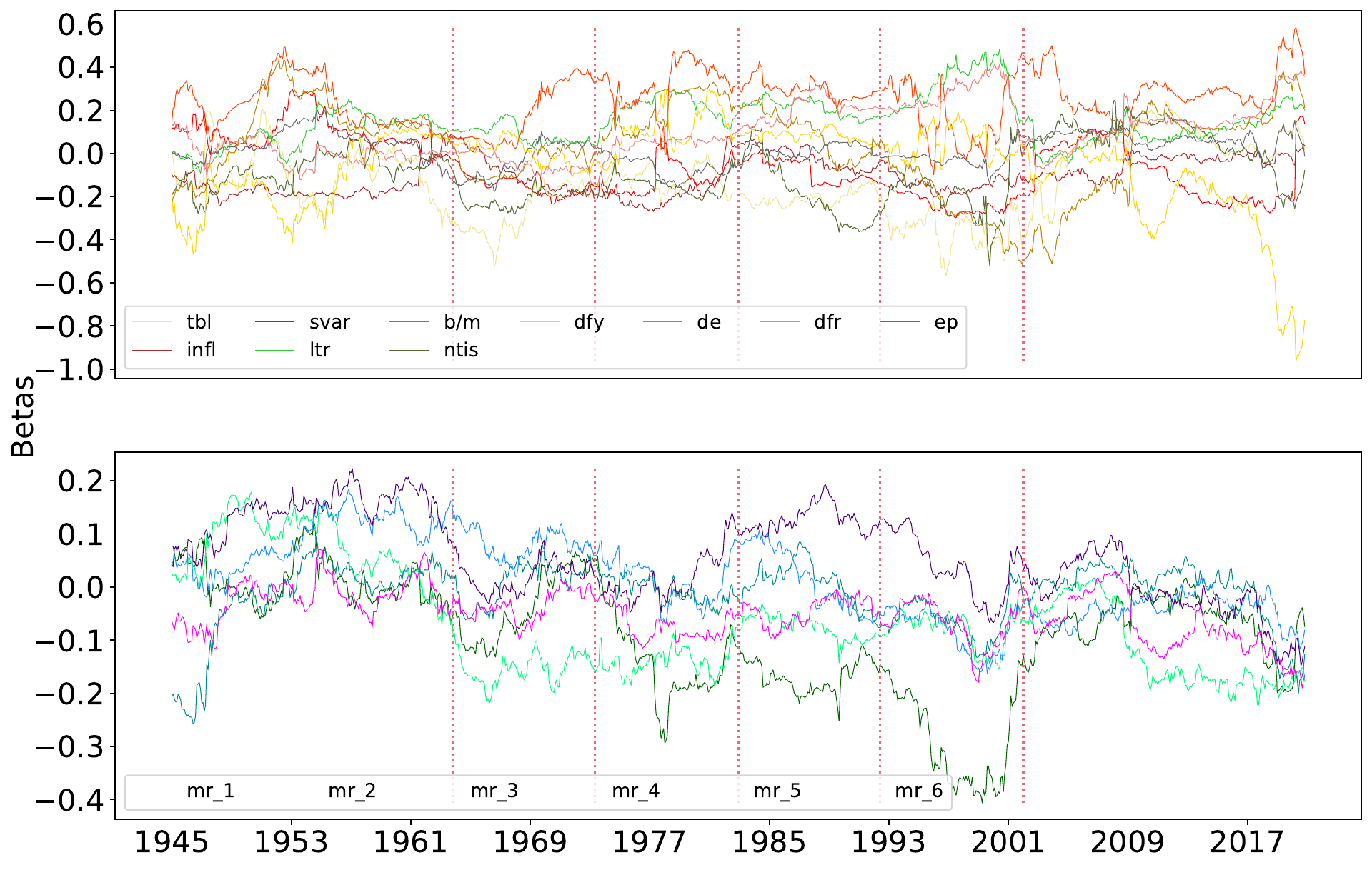}
\caption{\textbf{Time series of the betas.} \small Each month, we plot the coefficients of the linear time series regressions of the returns onto each predictor (the $\hat{\beta}^{(k)}$ in Equation \eqref{eq:predreg}). Estimations are performed on a fifteen-year rolling window. mr\_i denotes the $i$-month lagged marked return. Returns and predictors are standardized following the protocol of \cite{kelly2023virtue} (Section V). Change point estimates (see Section \ref{sec:betas}) are represented by red (vertical) dashed lines. \label{fig:betasTS}}
\end{figure}

\subsection{Time-varying betas}
\label{sec:betas}

There is considerable evidence supporting the time-varying nature of the coefficients in asset pricing models. Such effects have been documented on the performance of \textbf{predictive regressions} (\cite{dangl2012predictive}, \cite{coulombe2023time}, \cite{farmer2023pockets}), in \textbf{factor models} (\cite{ghysels1998stable}, \cite{bollerslev2023optimal}) and for \textbf{risk premia estimation} in particular (\cite{gagliardini2016time}, \cite{bakalli2023penalized}). With regard to return prediction with nonlinear models, we also refer to \cite{capponi2025nonstationarity} who show that model performance strongly fluctuates through time.

Empirically, this stylized fact can be revealed by change point detection. Such techniques seek to uncover sharp changes in the distribution of ordered observations (time series). They can take two forms, parametric and non-parametric. Parametric change point detection constrains the model to learn a particular functional form, i.e., observations between the change points are assumed to be drawn from a specific family of distribution (see \cite{Frick2014MultCPD}, \cite{roy2016cpdMarkov}, \cite{Enikeeva2019hdCpd}). This assumption is relaxed in the case of non-parametric change point detection and we refer to \cite{aminikhanghahi2017survey} for a complete survey on change point detection applied to time series. Recent contributions evaluate the dissimilarity of distributions using kernel distances (\cite{garreau2018consistent}, \cite{arlot2019kernel}) or rank (\cite{lung2015homogeneity}).

We choose to rely on the multivariate nonparametric multiple change point detection method recently proposed by \cite{londschien2023RfDetection} and that builds upon the use of classifiers. Notably, this will allow us to study betas in a joint manner: the entire time series structure of the coefficients is used to detect change points. Specifically, our analysis of time-variation is carried out with random forests via the \textit{changeforest} algorithm\footnote{The repository is available at \href{https://github.com/mlondschien/changeforest.git\string\#https://github.com/mlondschien/changeforest.git}{https://github.com/mlondschien/changeforest.git}} whose technicalities we briefly outline below.

The idea behind the algorithm is to determine recursive binary splits at the points maximizing the nonparametric classifier log-likelihood ratio. 
This is shown in Figure \ref{fig:changepoints}.\footnote{We stick to the default hyperparameters of the algorithm with two exceptions. First, we set the maximum tree depth to 2 since, by default, $p^{1/2} = 4$ features are selected for each tree. Second, we arbitrarily choose a minimal relative segment length of 0.125 so the results appear sparser. By setting a smaller minimum relative length, the algorithm detects an abundance of change points.} Therein, we have fed the time-series of the betas (the $\hat{\beta}_t^{(p)}$ from Figure \ref{fig:betasTS}) to \textit{changeforest}.\footnote{We exclude the variables "lty", "dp", "dy" and "tms" from the regressions as they introduced too much multicollinearity, yielding some estimates to be overly positive while others compensating with large negative values.} The top panel of the graph shows that the most relevant split occurs in December 1982. It is the one that maximizes the gain of the classifier when deciding between two regimes (two classes). The second horizontal panel shows the following two splits, in November 1963 and January 2002. Lastly, the bottom panel provides three new splits. Indeed, we see that the first and last periods are not split in two because the corresponding tests were not conclusive. This procedure tests if the split significantly improves the gain of the classifier. When the gain is too marginal, the algorithm stops, as for the pruning of simple decision trees. The final result clearly identifies several break points in the time-series of loadings, which justifies our assumption that changes occur between the training (in-sample) phase and the testing (out-of-sample) stage.

\begin{figure}[!h]
\centering
    \includegraphics[width=14cm]{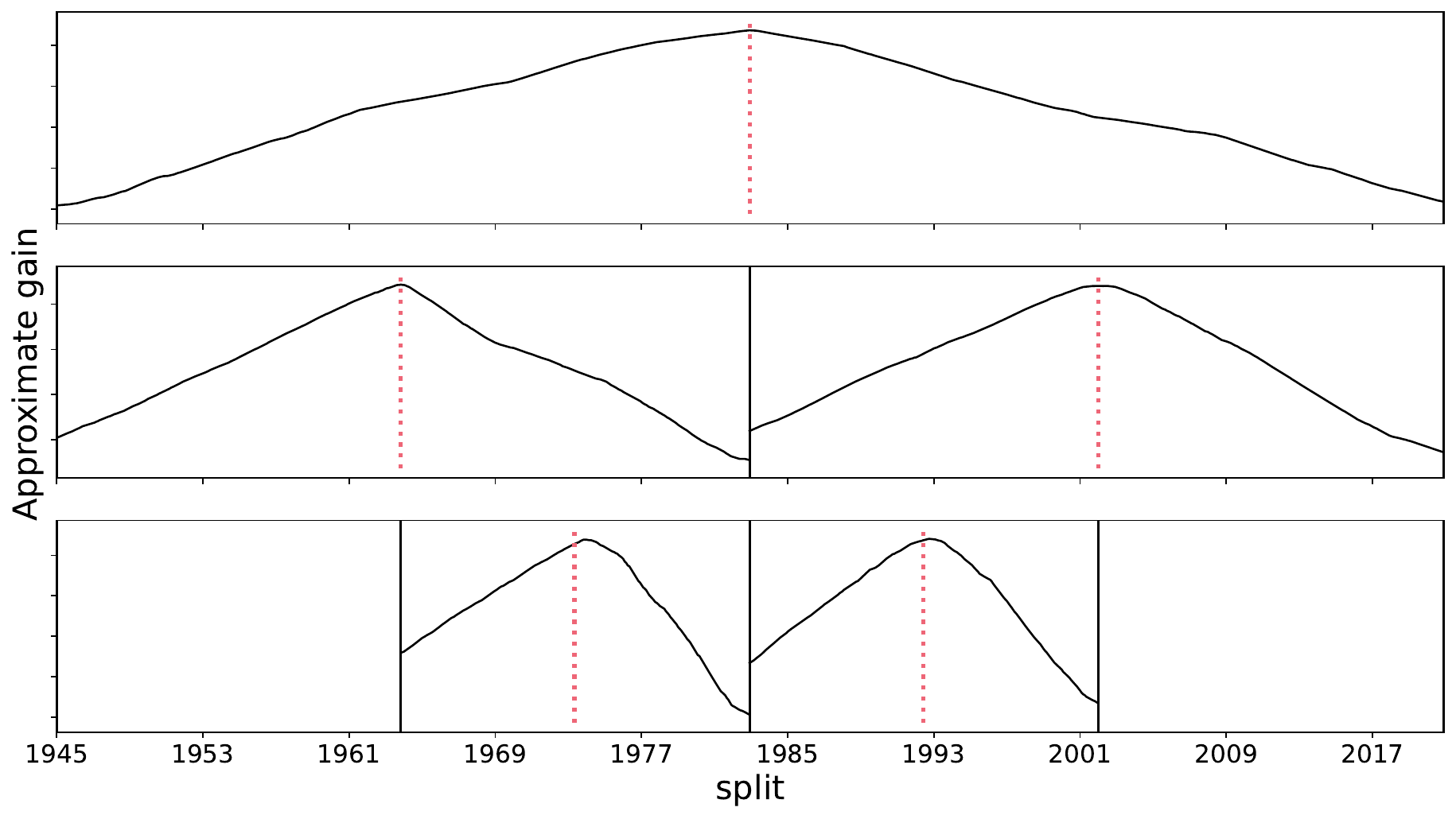}
\caption{\textbf{Change point detection.} \small We plot the approximate gain curves from the last step of the two-step search for the three binary segmentations obtained with the \textit{changeforest} algorithm (\cite{londschien2023RfDetection}). Change point estimates are marked with vertical red dots. \label{fig:changepoints}}
\end{figure}

\subsection{Protocol}
\label{sec:sensi}

We now closely follow the empirical protocol of \cite{kelly2023virtue} and we keep the notations therein. The idea is to model aggregate returns as a linear function of a large number of predictors which are created from a smaller set of macro-economic indicators, namely those analyzed in \cite{welch2008comprehensive}. To expand the predictor space, \cite{kelly2023virtue} resort to Random Fourier Features (RFFs). More precisely, if $G_t$ is the $15\times 1 $ vector of the original predictors, we will use
\begin{equation}
    S_{i,t}=\frac{1}{\sqrt{p}} \left[\sin(\gamma w_i'G_t), \cos(\gamma w_i'G_t) \right], \quad w_i \overset{d}{=}\mathcal{N}(O,I_{15})
    \label{eq:RFF}
\end{equation}
as our independent variables in the predictive regression $r_{t+1}=S_{i,t}'\beta_i+e_{t+1}$. Importantly, $i$ is the index of a random draw of the stochastic weights $w_i$, which are i.i.d. Because the predictors will be random, it is important to cancel the noise by averaging across many independent draws of $w_i$ - usually several hundreds. Moreover, we are interested in the overparametrized regime, hence all our results below pertain to the case $n=12$ and $p=600$, so that the complexity level is equal to $c=50$. In \cite{kelly2023virtue}, the benefits from complexity vanish theoretically for $c>10$ and empirically for $c>50$. This is why we henceforth work with a high and fixed level of complexity $c=50$. 

In Equation \eqref{eq:RFF}, an important parameter is $\gamma$, the bandwidth of the kernel, which determines the level of nonlinearity of the model induced by the RFFs. In \cite{kelly2023virtue}, the authors argue that their approach is insensitive to this choice, at least when $\gamma$ lies between 0.5 and 2. Our findings below paint a more nuanced picture. 

Each month, given a random draw $w_i$ and the features from Equation \eqref{eq:RFF}, we estimate the $\hat{\beta}$ as 
\begin{equation}
    \hat{\beta}^{RF}_i(z)=\left(zI + n^{-1}\sum_{t=1}^n S_{i,t} S_{i,t}'\right)^{-1} \frac{1}{n}\sum_{t=1}^n S_{i,t}R_{t+1}, 
\end{equation}
where $RF$ stands for "\textit{random features}" and $z>0$ is a regularization parameter which ensures that the inverse matrix is well defined when the number of predictor $p$ exceeds the sample size $n$. The weight assigned to the market portfolio is then simply this estimator defined above and the corresponding realized return is evaluated after one month. The procedure is then repeated until the end of the sample and monthly returns are averaged over alternative periods of time.

For comparison purposes, we also report the average returns of the strategy when the raw features, $G_t$, are used to compute the position, instead of the RFFs, $S_{i,t}$. The portfolio weight of this strategy is hence given by 
\begin{equation}
    \hat{\beta}^{L}(z)=\left(zI + n^{-1}\sum_{t=1}^n G_{t} G_{t}'\right)^{-1} \frac{1}{n}\sum_{t=1}^n G_{t}R_{t+1}, 
    \label{eq:linear}
\end{equation}
where the superscript $L$ stands for "\textit{linear}". In Figure \ref{fig:expret}, the corresponding returns are shown with horizontal dashed lines (they do not depend on $\gamma$).  

\subsection{Results: sensitivity to bandwidths and periods}

In Figure \ref{fig:expret}, we plot the average return obtained when using $\hat{\beta}^{RF}_i(z)$ as portfolio weight - averaged across 500 draws of $w_i$. We average returns over periods of 15 years for two reasons. First, because 15 years is already a long horizon for an investor and second because this allows to span several types of economic conditions: both periods of growth and slowdown. The values are reported for two levels of regularization: $z=0.01$ (left) and $z=100$ (right).

\begin{figure}[!h]
\centering
    \includegraphics[width=16cm]{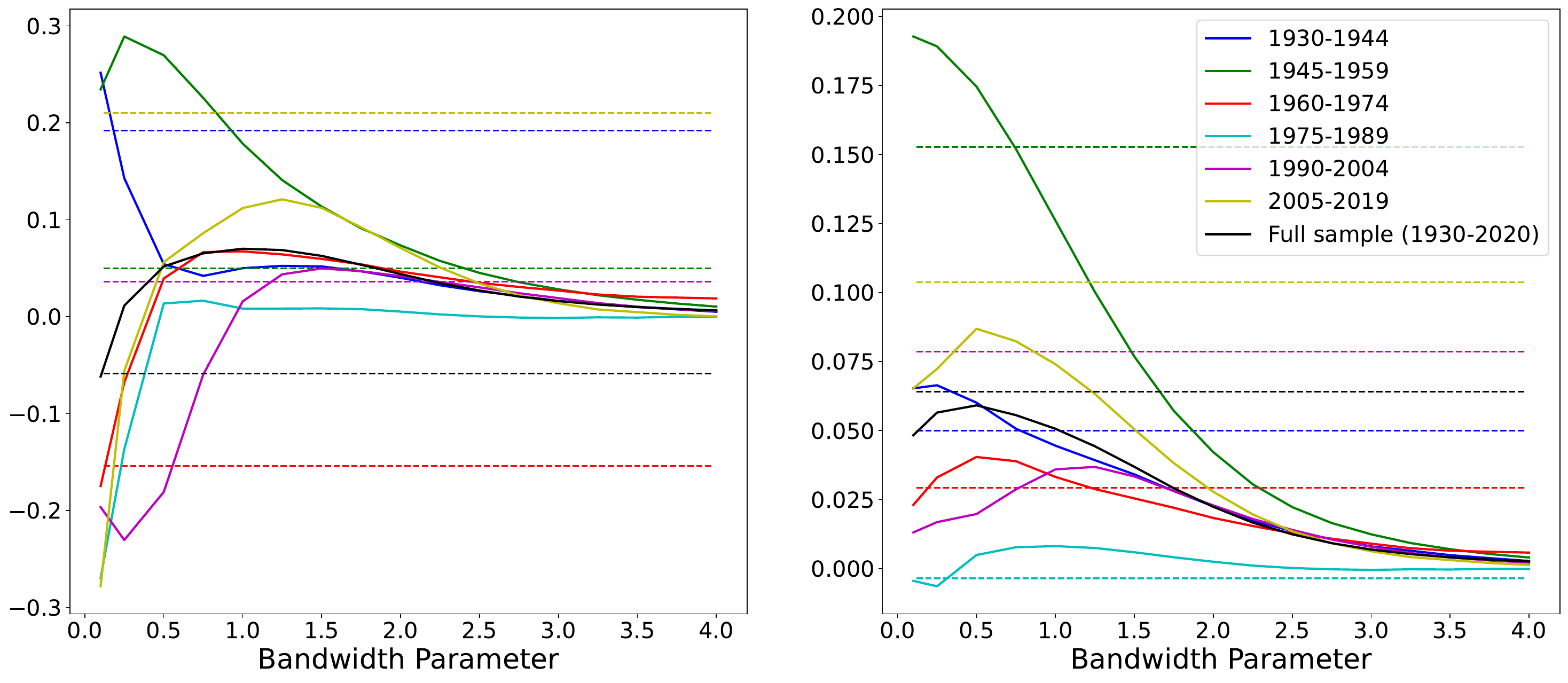}
\caption{\textbf{Returns' sensitivity to periods and bandwidth.} \small We plot average realized returns obtained on six sub-periods and on the full sample (solid lines) for $p=600$ RFFs, as a function of the bandwidth parameter of the kernel from the Random Fourier Features. The (dashed) horizontal lines display the return of the strategy when the raw features are used instead of RFFs. In this case, $p=15$. For the sake of visibility, we do not plot the return of the strategy when using raw features for the 1975-1989 period (-50\%). The left plot shows the case when the penalization parameter is $z=0.01$ (mild regularization), while the right one pertains to $z=100$ (strong regularization). Portfolios are averaged over 500 random draws.\label{fig:expret}}
\end{figure}

For a small regularization ($z=0.01$) and $\gamma = 2$, as in the original paper, the return for the period 2005-2019 is 7.24\% but it shrinks to 0.52\% for 1975-1989 - a reduction by a factor of 14. An analyst in 1974 would probably choose $\gamma=1$ (from the red curve with $z = 0.01$), and hope for a return of 5\%. But the performance in the following 15 years would prove very disappointing (less than 1\%: again, a significant contraction). 

For the sake of completeness, we also depict the sensitivity of the annualized Sharpe ratio in Figure \ref{fig:sharpe}. Over the full sample, we find a maximum value of 0.3, slightly lower than the roughly 0.4 reported by \cite{kelly2023virtue} in their Figure 8 (Panel A). These values are somewhat underwhelming if we put them in perspective with the 0.44 Sharpe ratio of the excess returns on the market index documented in \cite{pav2021sharpe} over the 1926-2020 period - from the \href{https://mba.tuck.dartmouth.edu/pages/faculty/ken.french/data_library.html\string#Ken French data library}{Ken French data library}. For the raw market portfolio, when the risk-free rate is not subtracted, the Sharpe ratio increases to 0.62. This is even before taking trading costs into account: holding the market portfolio entails minimum trading costs, which is not the case of the timing strategy presented in \cite{kelly2023virtue}. 

\begin{figure}[!h]
\centering
    \includegraphics[width=14cm]{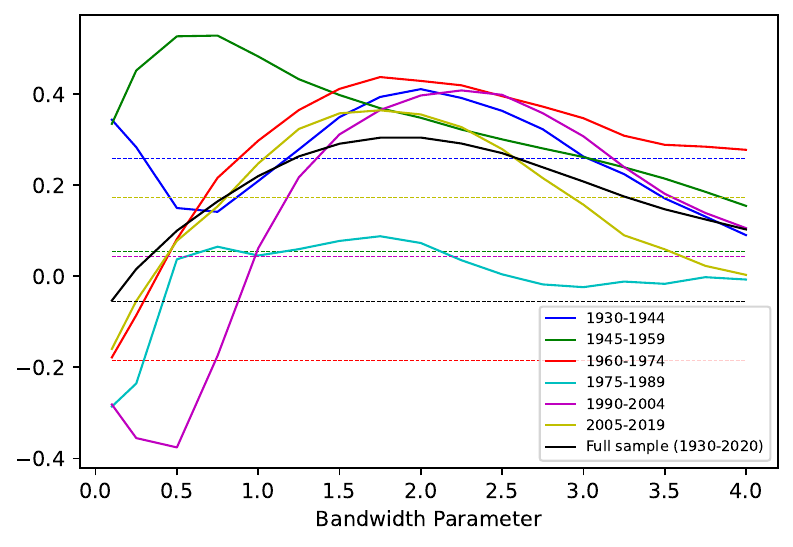}
\caption{\textbf{Sharpe ratio sensitivity to periods and bandwidth.} We plot average annualized Sharpe ratios obtained on six sub-periods and on the full sample for $p=600$ RFFs, as a function of the bandwidth parameter of the kernel from the Random Fourier Features. The (dashed) horizontal lines display the return of the strategy when the raw features are used instead of RFFs. In this case, $n=12$ and $p=15$. For the sake of visibility, we do not plot the return of the strategy when using raw features for the 1975-1989 period (-0.11). The penalization parameter is $z=0.01$ (mild regularization). Portfolios are averaged over 500 random draws. \label{fig:sharpe}}
\end{figure}

In fact, a closer look reveals the unreasonably high levels of volatility that are implied by these Sharpe ratios. A monthly return of 0.03 (reported both in Figure \ref{fig:expret} and in \cite{kelly2023virtue} for a bandwidth of two) translates arithmetically to 0.36 on an annual basis, thus a Sharpe ratio of 0.3 suggests a yearly volatility above 100\%, which is prohibitively high, even for risk-prone investors (e.g., hedge funds).

Lastly, Figure \ref{fig:sharpe} overwhelmingly indicates that the best choice of bandwidth is indeed close to two. Except for one period (subsequent to World War II), all curves reach their maximum when the bandwidth is in the vicinity of two. This is a stable and consistent pattern - which may explain why it is presented as the baseline case in \cite{kelly2023virtue}.

\subsection{Time-varying bandwidths}

In Figure \ref{fig:expret}, we note that the optimal choice of $\gamma$ is time-varying. It lies between $\gamma^*=0$ in the first period (left graph) or in the second period (right graph) and $\gamma^*=1.5$ in 1990-2004 or 2005-2019 (left panel) and in 1990-2004 (right panel). Choosing larger values for $\gamma$ reduces uncertainty, but at the cost of lower returns. 

We now want to account for the fact that investors can adjust $\gamma$ based on past performance. In a new exercise, we pick, for a given period of 15 years, the bandwidth that yielded the best results in the previous period (we call this the ``\textit{feasible}'' strategy), but also the bandwidth that proved the best ex-post, i.e., in \textit{hindsight}. The corresponding returns are gathered in Table \ref{table:stratreturns}.

As expected, when $\gamma$ is chosen in hindsight (rightmost columns), the strategy based on RFFs always outperforms the feasible one. Note that this scheme is unfeasible in practice. However, there are periods during which the RFF-based return is  dominated by the simpler linear strategy. Therefore, this further underlines some limitations of the more sophisticated approach. With enough regularization, we see that the simpler model reaches an unconditional (full sample) return that is higher than the high-complexity strategy. This is in line with the conclusions of \cite{shen2024can}.

\begin{table}[!h]
\centering
\begin{tabular}{ c c c | c c | c c }
\toprule
\multirow{2}{*}{\textbf{Period}} & \multicolumn{2}{c}{\textbf{Linear}} & \multicolumn{2}{c}{\textbf{RFFs (feasible)}} & \multicolumn{2}{c}{\textbf{RFFs (hindsight)}}\\
\cmidrule{2-7}
  & $z=0.01$ & $z=100$ & $z=0.01$ & $z=100$ & $z=0.01$ & $z=100$ \\
\midrule \makecell{1930-1944} & \makecell{19.22\%} & \makecell{5.00 \%} & \makecell{5.08\%} & \makecell{2.82\%} & \makecell{25.16\%} & \makecell{6.64\%} \\
\makecell{1945-1959} & \makecell{4.97\%} & \makecell{15.28\%} & \makecell{23.45\%} & \makecell{18.92\%} & \makecell{28.90\%} & \makecell{23.45\%} \\
\makecell{1960-1974} & \makecell{-15.40\%} & \makecell{2.92\%} & \makecell{-6.78\%} & \makecell{2.31\%} & \makecell{6.73\%} & \makecell{4.04\%} \\
\makecell{1975-1989} & \makecell{-49.32\%} & \makecell{-0.35\%} & \makecell{0.84\%} & \makecell{0.49\%} & \makecell{0.84\%} & \makecell{0.81\%} \\
 \makecell{1990-2004} & \makecell{3.63\%} & \makecell{7.86\%} & \makecell{4.37\%} & \makecell{3.60\%} & \makecell{4.98\%} & \makecell{3.68\%} \\
\makecell{2005-2019} & \makecell{21.01\%} & \makecell{10.38\%} & \makecell{12.11\%} & \makecell{6.33\%} & \makecell{12.11\%} & \makecell{8.69\%}\\ \midrule
 \makecell{Full sample} & \makecell{-5.87\%} & \makecell{6.40\%} & \makecell{6.51\%} & \makecell{5.75\%} & \makecell{13.12\%} & \makecell{7.89\%} \\
\bottomrule \vspace{-6mm}
\end{tabular}
\caption{\textbf{Strategy returns under different learning schemes.} \small We exhibit the (monthly) returns of the market-timing strategy under three learning schemes. The column ``Linear'' refers to the strategy of Equation \eqref{eq:linear} relying on ridge regression directly applied to \citeauthor{goyal2023comprehensive}'s macro-economic indicators ($p=15$). The last two columns display the returns based on Random Features ridge regression. The ``hindsight'' strategy represents the best case scenario, i.e., if the investor is able to choose the bandwidth with knowledge of the future. The ``feasible'' strategy consists in choosing the bandwidth parameter based on the best value from the previous period. Note that for the first subperiod (1930-1944), we took the average return over the different values of $\gamma$. For these last two schemes, we generate $p=600$ random features and average portfolios over 500 random draws. \label{table:stratreturns}}
\end{table}

\section{Counterfactual returns}
\label{sec:cf_ret}

Our most important theoretical result is Proposition \ref{prop:expretdriftisomis}, which quantifies the loss in performance based on the misalignment between the in-sample loadings and the out-of-sample loadings. In this last section, we thus propose to assess the strategy's return in the hypothetical scenario where the data would not suffer from posterior drift. This will allow us to quantify the loss from Proposition \ref{prop:expretdriftisomis}. To this end, we generate counterfactual returns using the time series of estimated loadings $\hat{\beta}_t$ obtained by regressing the RFFs onto the real returns $r_{t+1}$. 
Formally, the counterfactual returns are sampled as follows:
\begin{equation}
r^{(c)}_{t+1}= \sum_{k=1}^p \hat{\beta}^{(k)}_{t} x_t^{(k)}
    \label{eq:predreg_cf}
\end{equation}
where the $x_t^{(p)}$ are all $p$ predictors (i.e., the RFFs, not the original signals). We underline that this assumes no misspecification as well, i.e., $\theta=0$ in Equation \eqref{eq:dgpmis}. We opt for simplicity and adding misspecification would artificially burden the exercise with an additional source of errors. Moreover, this route is already explored in \cite{kelly2023virtue}, whereas we focus here on posterior drift.

Once they are generated, the hypothetical returns have low average returns as well as low volatility over the whole period (90 years). This was expected because the ridgeless algorithm selects the solution with the smallest norm. Therefore, we scale these returns in order to match the first two empirical moments of the historical realized returns, i.e.,
\begin{equation}
\Tilde{r}^{(c)}_{t+1}= \Bigl(\frac{r^{(c)}_{t+1} - \Bar{r}_c}{\sigma_c}\Bigl) \sigma + \Bar{r}
    \label{eq:cf_scaled}
\end{equation}
where $\Bar{r}_c$ (resp. $\Bar{r}$) is the sample mean of the counterfactual (resp. realized) returns and $\sigma_c$ (resp. $\sigma$) is their sample standard deviation.

The average returns are plotted in Figure \ref{fig:cf_ridgeless} and confirm the theoretical results of Section \ref{subsec:performance}. The performance of the strategy is substantially stronger without posterior drift (dashed lines). This difference shown for $z=0.01$ is more pronounced when the bandwidth is small. As the latter increases, all returns slowly shrink to zero. 
The case of strong regularization ($z=100$) corroborates these patterns and is postponed to Figure \ref{fig:cf_strongreg} in the Appendix. Overall, it is clear that when the loadings do not change, the model that is learned allows to generate substantial profits. But when the DGP changes (via the posterior drift), then returns are curtailed.   

\begin{figure}[!h]
\centering
    \includegraphics[width=16cm]{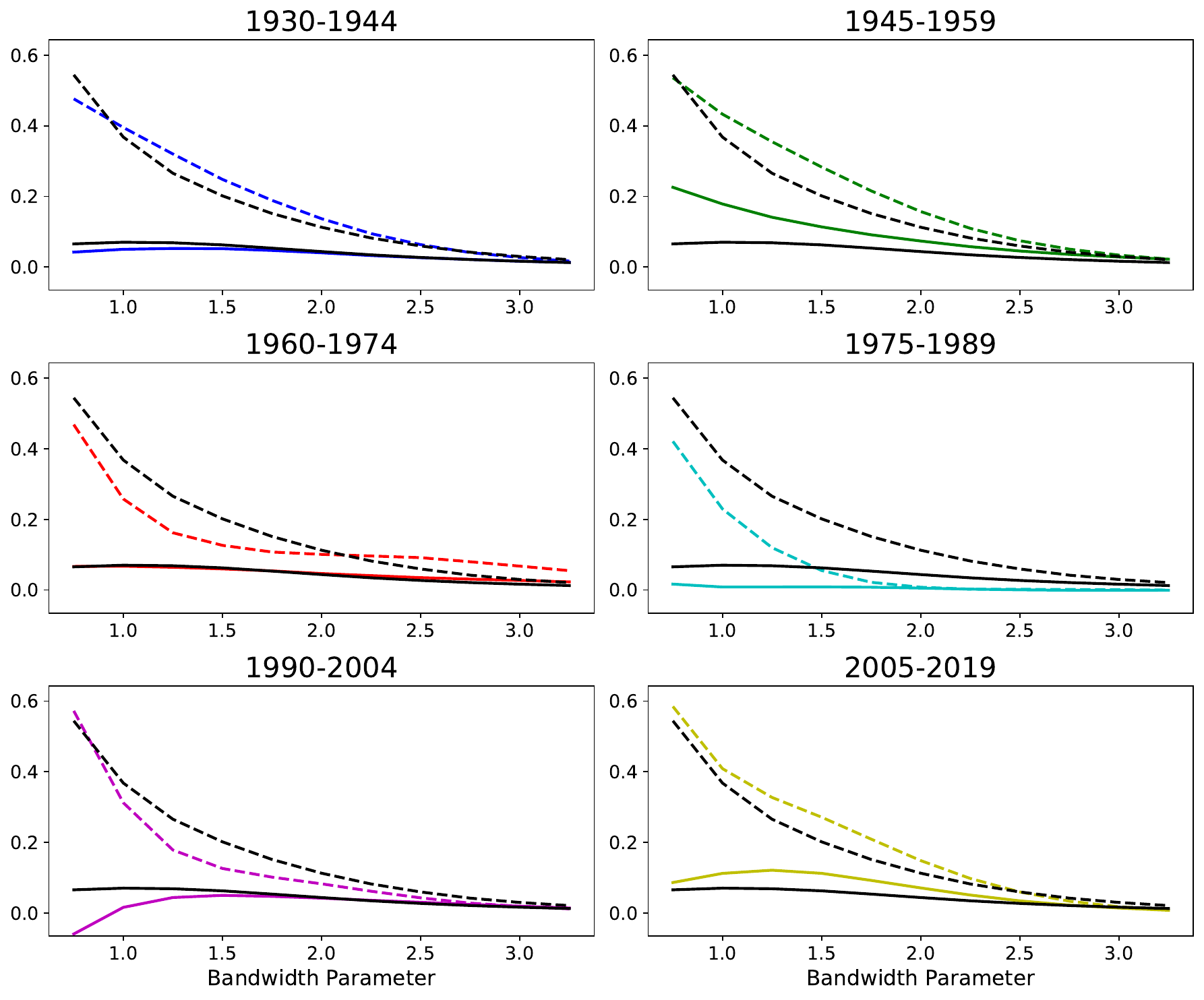}
\caption{\textbf{Strategy returns without posterior drift under mild regularization.} \small We plot the real (with posterior drift, solid lines) and the counterfactual (without posterior drift, dashed lines) returns of the strategy obtained on six sub-periods for $p=600$ RFFs, as a function of the bandwidth parameter of the kernel from the Random Fourier Features. Colored lines represent the strategy returns for the corresponding sub-periods, while \textbf{black lines} show returns over the full sample. The penalization parameter is $z=0.01$ (mild regularization). Returns are averaged over 500 random draws of RFFs.  \label{fig:cf_ridgeless}}
\end{figure}

\clearpage

\section{Conclusion}
\label{sec:conc}

This article seeks to document the sensitivity to posterior drift of the accuracy of overparametrized linear models in equity premium prediction. Our theoretical results exhaustively and invariably point to a natural degradation of performance when the shift between in-sample and out-of-sample betas increases. Empirically, we find that out-of-sample timing can be sometimes promising, but also very unstable. Average (and risk-adjusted) returns strongly depend on market conditions, even on long horizons, and on a critical parameter that is not straightforward to tune. We attribute these limitations to changes in links between premia and macro variables. These changes strongly attenuate the signal contained in predictors. In sum, large models \textit{can} bring value, but not unconditionally, and one should resort to them while keeping these caveats in mind. 

\section*{Replication package}

\noindent The python code to replicate Figure \ref{fig:sr_signal} is available 
\href{https://drive.google.com/file/d/1lf2zgPKyz_WoYlrakVyNsTVu8XPSvzO5/view?usp=sharing}{here}.

\noindent The python code to replicate Figures \ref{fig:retmildprop}, \ref{fig:volmildprop}, \ref{fig:srmildprop} and the ones in sections \ref{sec:app_numericanalysis_iid} and \ref{sec:app_numericanalysis_correlated}  of the Appendix (except Figure \ref{fig:convergence}) is available 
\href{https://drive.google.com/file/d/1A1p7gocZYnrKp7TXctA8P-DXObTeoE45/view?usp=sharing}{here}.

\noindent The python code to replicate Figures \ref{fig:betasTS} and \ref{fig:changepoints} is available
\href{https://drive.google.com/file/d/11X1xOwGs0XWIbMDuCYDn2wZGx1Ke-4ZY/view?usp=sharing}{here}.

\noindent The python code to replicate Figures \ref{fig:expret}, \ref{fig:sharpe}, \ref{fig:cf_ridgeless} and \ref{fig:cf_strongreg} is available 
\href{https://drive.google.com/file/d/1VpnsCpIOqY_xv7z-jikvJvCvM7vOFGup/view?usp=sharing}{here}.

\noindent The python code to replicate Figure \ref{fig:convergence} is available 
\href{https://drive.google.com/file/d/15W-fjSYRjqU5cN7TOm5lY8f2YeYwN8wj/view?usp=sharing}{here}.

\bibliographystyle{chicago}
\bibliography{bib}

\clearpage

\appendix


\appendix

\addcontentsline{toc}{section}{Appendix} 
\part{Appendix} 
\parttoc 

\section{Numerical experiments: iid features}
\label{sec:app_numericanalysis_iid}

\subsection{Simulated moments under strong regularization (equidistributed parameters vector)}
\label{subsec:simretstrongequi}

\begin{figure}[!h]
\centering
    \includegraphics[width=16cm]{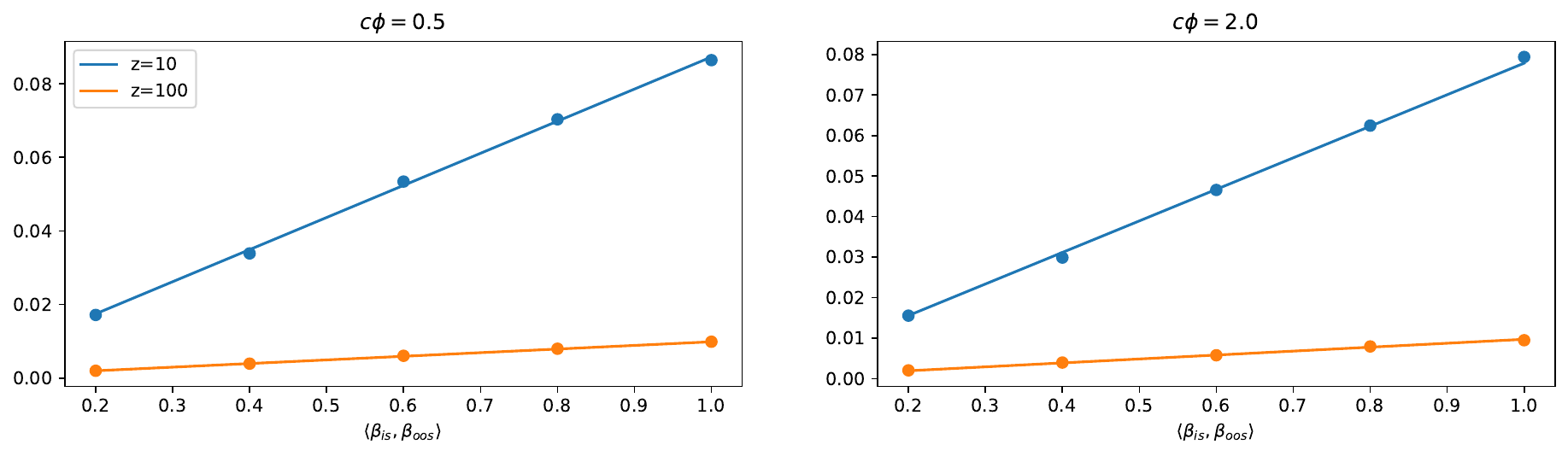}
\caption{\textbf{Simulated and theoretical returns of the strategy under strong regularization.} \small The model's true complexity is $c=3$ ($p+q = 300$, $n=100$). $z$ is the regularization parameter of the ridge estimator. We plot the theoretical returns of the strategy as per Proposition \ref{prop:expretdriftisomis} (solid line). Simulated returns (dots) are generated according to the process described in Section \ref{sec:numerical} and are averaged over 100,000 random draws. \label{fig:retstrongprop}}
\end{figure}

\begin{figure}[!h]
\centering
    \includegraphics[width=16cm]{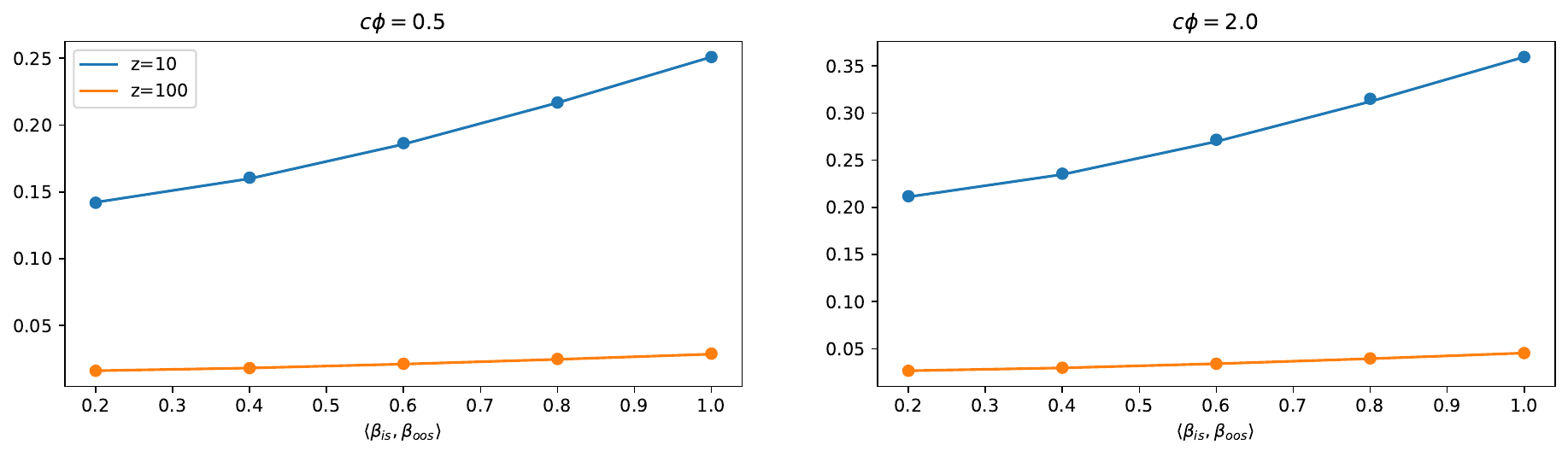}
\caption{\textbf{Simulated and theoretical volatility of the strategy under strong regularization.} \small The model's true complexity is $c=3$ ($p+q = 300$, $n=100$). $z$ is the regularization parameter of the ridge estimator. We plot the theoretical volatility of the strategy as per Proposition \ref{prop:volstratret_iid} (solid line). Simulated returns (dots) are generated according to the process described in Section \ref{sec:numerical} and are averaged over 100,000 random draws. \label{fig:volstrongprop}}
\end{figure}

\begin{figure}[!h]
\centering
    \includegraphics[width=16cm]{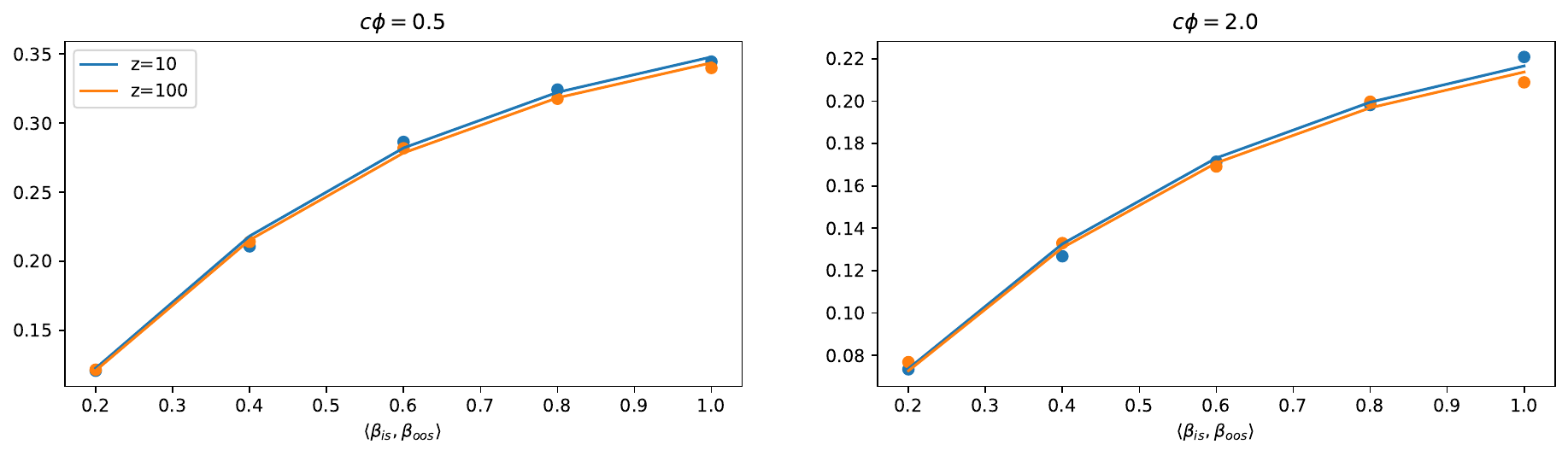}
\caption{\textbf{Simulated and theoretical Sharpe ratio of the strategy under strong regularization.} \small The model's true complexity is $c=3$ ($p+q = 300$, $n=100$). $z$ is the regularization parameter of the ridge estimator. We plot the theoretical Sharpe ratio of the strategy as per Eq. \eqref{eq:sharpe_ratio_iid} (solid line). Simulated returns (dots) are generated according to the process described in Section \ref{sec:numerical} and are averaged over 100,000 random draws. \label{fig:srstrongprop}}
\end{figure}

\clearpage

\subsection{Simulated moments (concentrated parameters vector)}
\label{subsec:simretconcentrated}

We simulate the expected moments of the market timing strategy following the process described in Section \ref{sec:numerical} except for the true parameters vectors related to the observed data, $\beta_{is}$ and $\beta_{oos}$.

In this section, the latter are concentrated on some directions of $\Sigma$. The first half of the coefficients of $\beta_{is}$ are equal to $1$. The remaining components are equal to $0.1$. Furthermore, we normalize $\beta_{is}$ so it has unit signal, i.e. $\| \beta_{is} \|^2=1$.  

We consider a sequence of $(\beta_{oos,k})_k$, $k \in \{ 1,2,3,4,5 \}$, with unit signal. To this end, let the first $(50/k)\%$ of the coefficients of $\beta_{oos,k}$ be equal to $1$ while the remaining components are set to $0.1$. As for $\beta_{is}$, we normalize these vectors so they have unit signal, namely $\| \beta_{oos,k} \|^2 = 1$ for any $k \in \{ 1,2,3,4,5 \}$. 

We analogously generate the true parameters vectors for the unobserved data, $\theta_{is}$ and $\theta_{oos}$, setting only a different number of coefficients, viz. $q$ instead of $p$.

\begin{figure}[!h]
\centering
    \includegraphics[width=16cm]{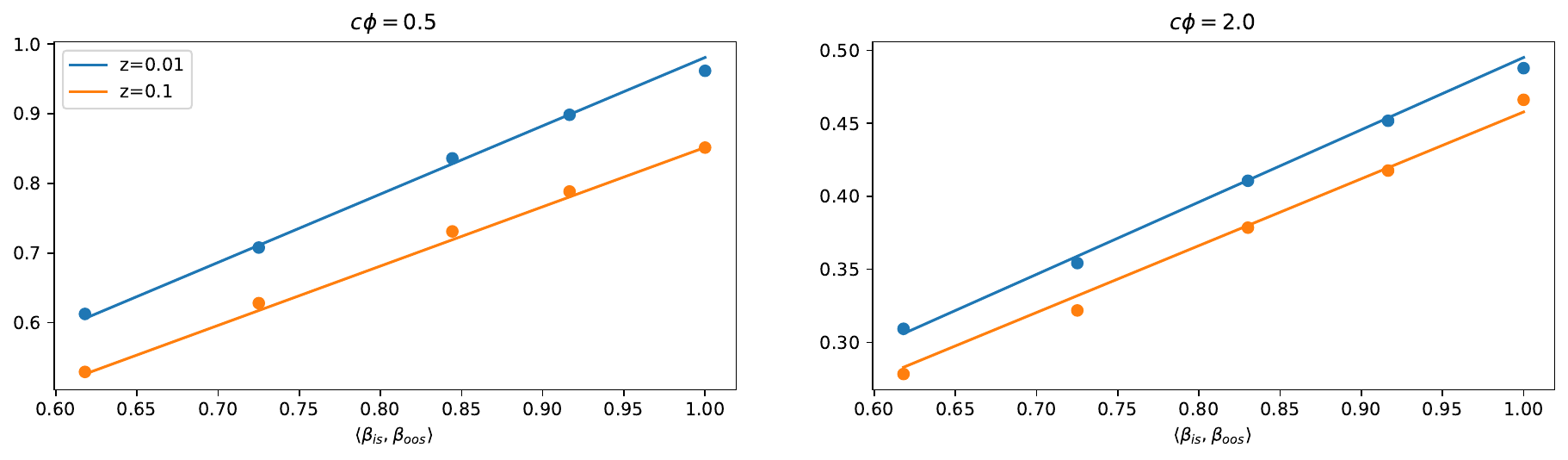} 
\caption{\textbf{Simulated and theoretical returns of the strategy under mild regularization (concentrated (true) parameters vectors).} \small The model's true complexity is $c=3$ ($p+q = 300$, $n=100$). $z$ is the regularization parameter of the ridge estimator. We plot the theoretical returns of the strategy as per Proposition \ref{prop:expretdriftisomis} (solid line). Simulated returns (dots) are generated according to the process described in Section \ref{subsec:simretconcentrated} and are averaged over 100,000 random draws. \label{fig:retmildconc}}
\end{figure}

\begin{figure}[!h]
\centering
    \includegraphics[width=16cm]{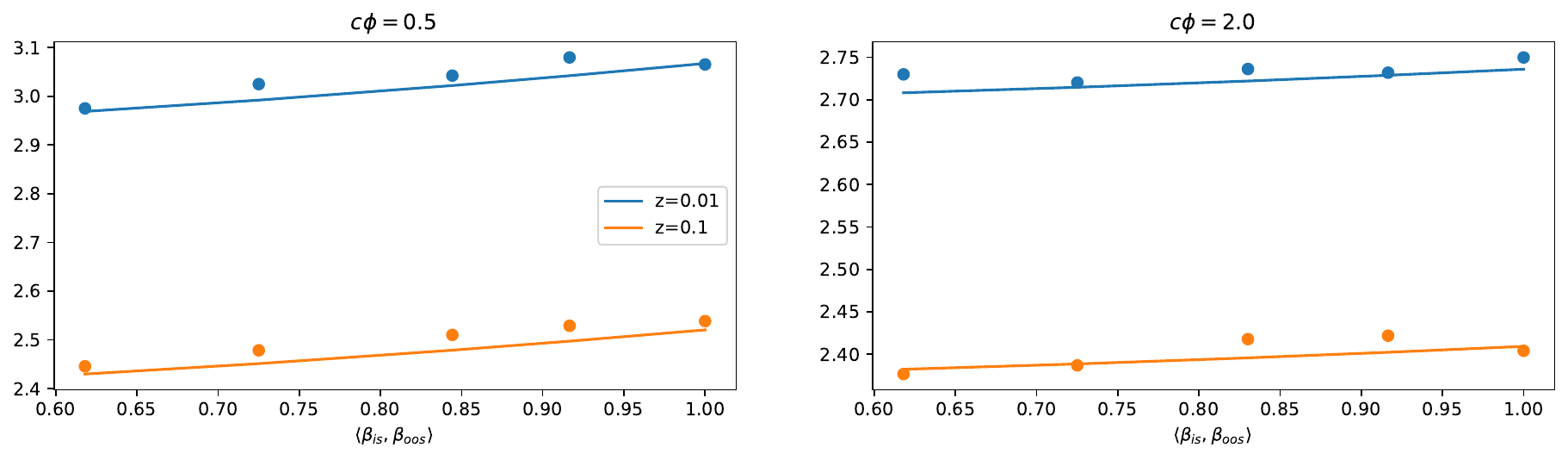}  
\caption{\textbf{Simulated and theoretical volatility of the strategy under mild regularization (concentrated (true) parameters vectors).} \small The model's true complexity is $c=3$ ($p+q = 300$, $n=100$). $z$ is the regularization parameter of the ridge estimator. We plot the theoretical volatility of the strategy as per Proposition \ref{prop:volstratret_iid} (solid line). Simulated returns (dots) are generated according to the process described in Section \ref{subsec:simretconcentrated} and the sample volatility is computed over 100,000 random draws. \label{fig:volmildconc}}
\end{figure}

\begin{figure}[!h]
\centering
    \includegraphics[width=16cm]{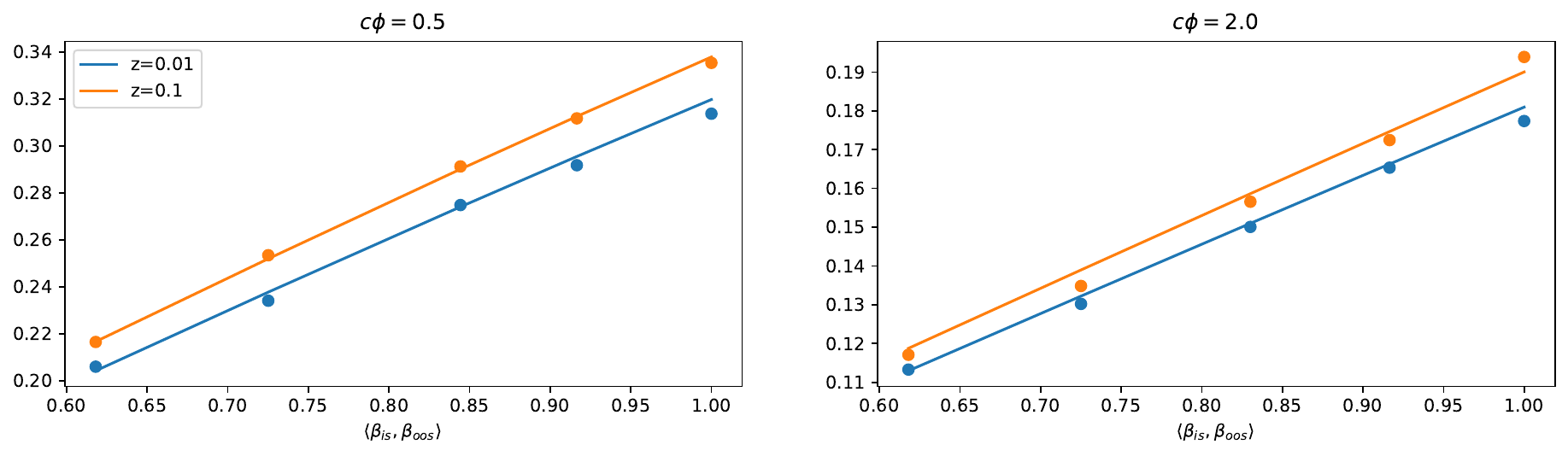}  
\caption{\textbf{Simulated and theoretical Sharpe ratio of the strategy under mild regularization (concentrated (true) parameters vectors).} \small The model's true complexity is $c=3$ ($p+q = 300$, $n=100$). $z$ is the regularization parameter of the ridge estimator. We plot the theoretical Sharpe ratio of the strategy as per Eq. \eqref{eq:sharpe_ratio_iid} (solid line). Simulated returns (dots) are generated according to the process described in Section \ref{sec:numerical} and the average  averaged over 100,000 random draws.  \label{fig:srmildconc}}
\end{figure}

\begin{figure}[!h]
\centering
    \includegraphics[width=16cm]{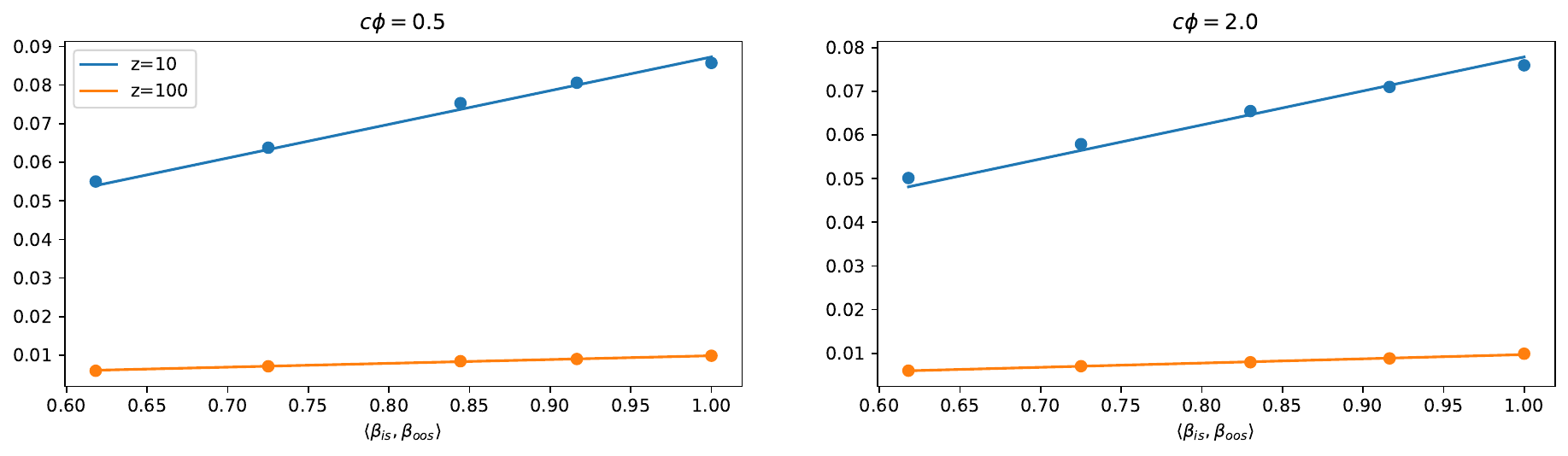}  
\caption{\textbf{Simulated and theoretical returns of the strategy under strong regularization (concentrated (true) parameters vectors).} \small The model's true complexity is $c=3$ ($p+q = 300$, $n=100$). $z$ is the regularization parameter of the ridge estimator. We plot the theoretical returns of the strategy as per Proposition \ref{prop:expretdriftisomis} (solid line). Simulated returns (dots) are generated according to the process described in Section \ref{subsec:simretconcentrated} and are averaged over 100,000 random draws. \label{fig:retstrongconc}}
\end{figure}

\begin{figure}[!h]
\centering
    \includegraphics[width=16cm]{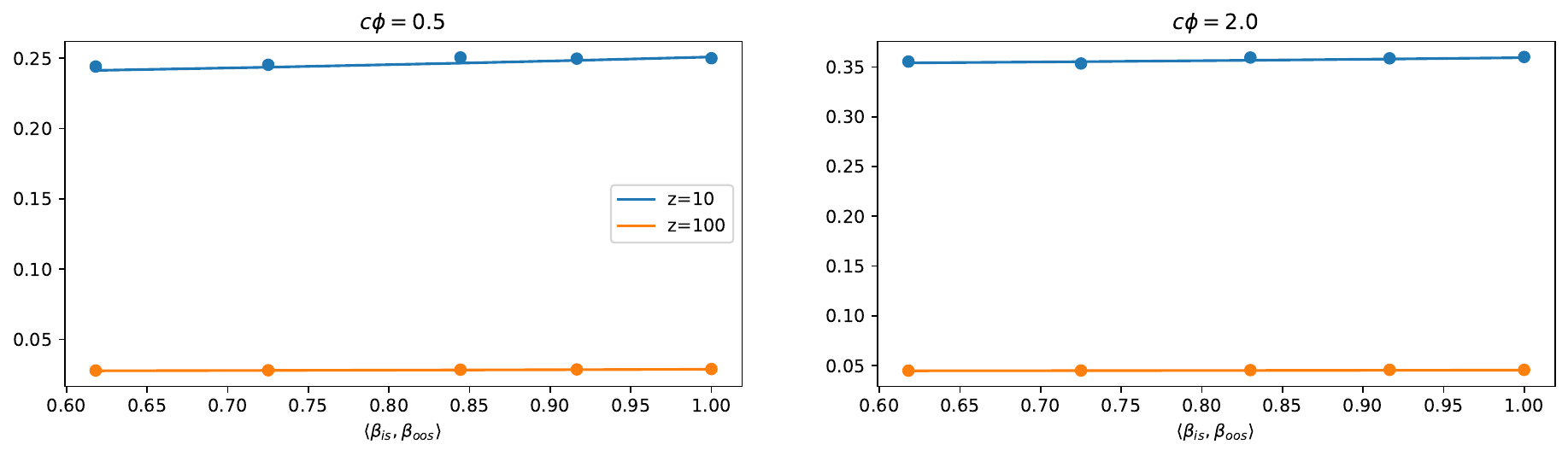}  
\caption{\textbf{Simulated and theoretical volatility of the strategy under strong regularization (concentrated (true) parameters vectors).} \small The model's true complexity is $c=3$ ($p+q = 300$, $n=100$). $z$ is the regularization parameter of the ridge estimator. We plot the theoretical volatility of the strategy as per Proposition \ref{prop:volstratret_iid} (solid line). Simulated returns (dots) are generated according to the process described in Section \ref{subsec:simretconcentrated} and are averaged over 100,000 random draws. \label{fig:volstrongconc}}
\end{figure}

\begin{figure}[!h]
\centering
    \includegraphics[width=16cm]{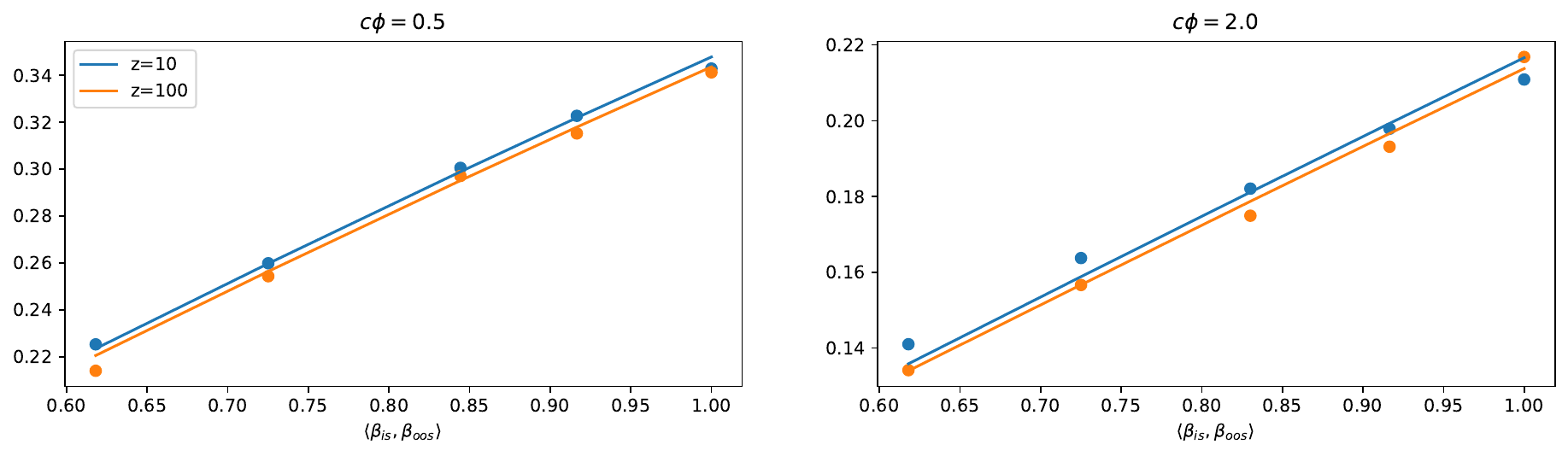}  
\caption{\textbf{Simulated and theoretical Sharpe ratio of the strategy under strong regularization (concentrated (true) parameters vectors).} \small The model's true complexity is $c=3$ ($p+q = 300$, $n=100$). $z$ is the regularization parameter of the ridge estimator. We plot the theoretical Sharpe ratio of the strategy as per Eq. \eqref{eq:sharpe_ratio_iid} (solid line). Simulated returns (dots) are generated according to the process described in Section \ref{subsec:simretconcentrated} and are averaged over 100,000 random draws. \label{fig:srstrongconc}}
\end{figure}

The convergence of average simulated returns towards theoretical returns is illustrated in Figure \ref{fig:convergence}. Note that returns are averaged over 10 million random draws which is already computationally prohibitive when $n>100$. High precision (errors below $0.2 \%$) is already reached for $n=130$. According to Propositions \ref{prop:nonasyboundsdrift} and \ref{prop:nonasyboundsdriftridgeless}, the gaps tend to zero as the number of samples, $n$, increases. Better convergence patterns (strictly decreasing graphically) could also be obtained by averaging over a more substantial number of draws, but at the cost of extended computational times.


\begin{figure}[!h]
\centering
    \centering
  \begin{minipage}[b]{0.49\textwidth}
    \includegraphics[width=\textwidth]{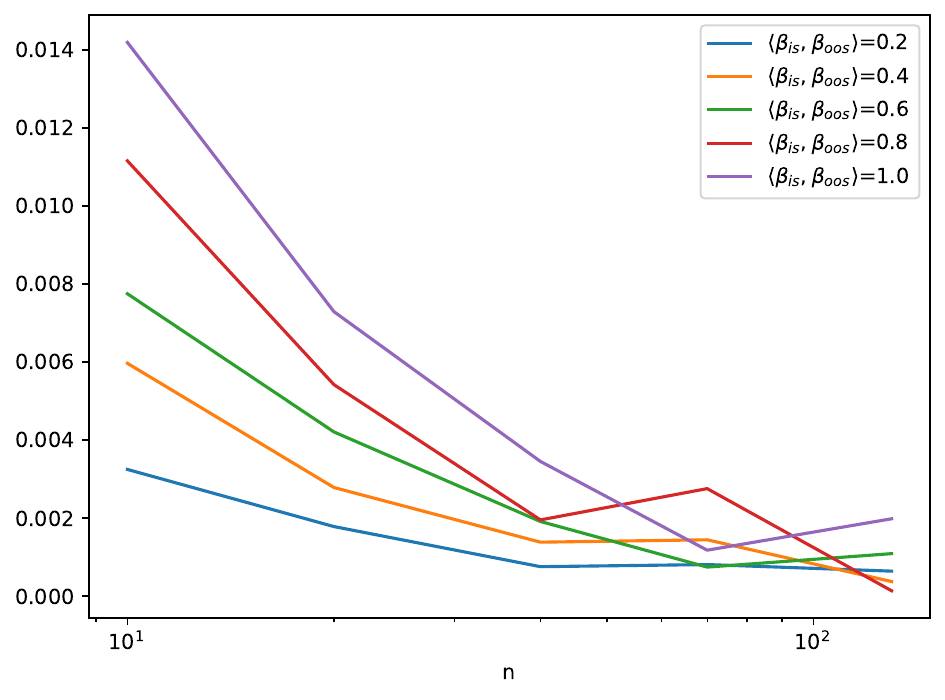}
  \end{minipage}
  \hfill
  \begin{minipage}[b]{0.49\textwidth}
    \includegraphics[width=\textwidth]{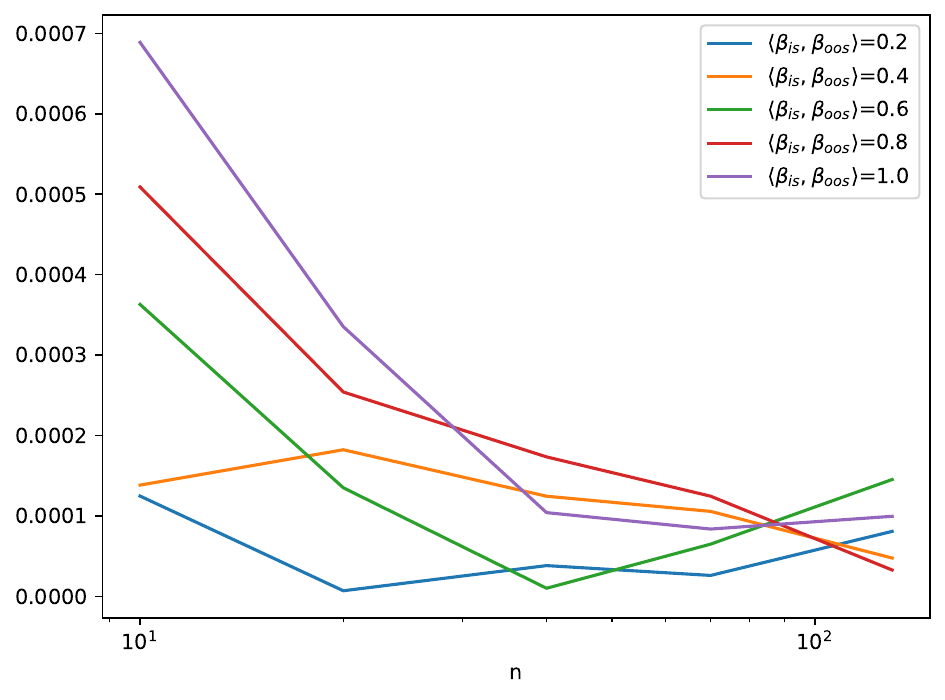}
  \end{minipage}  
\caption{\textbf{Convergence of average simulated returns.} \small We plot the difference between theoretical and average simulated returns. The true model's complexity is $c=3$ while the working model has a complexity of $c \phi = 0.5$. The left plot shows the case when the penalization parameter is $z=0.1$ (mild regularization), while the right one pertains to $z=10$ (strong regularization). Simulated returns are generated according to the process described in Section \ref{sec:numerical} and are averaged over 10 million random draws.\label{fig:convergence}}
\end{figure}

\clearpage

\section{Numerical experiments: correlated features}
\label{sec:app_numericanalysis_correlated}

\subsection{Equidistributed parameters vector}
\label{subsec:numanalysis_correlated_equidist}

We simulate the expected moments of the market timing strategy following the process described in Section \ref{sec:numerical} but introducing correlation among the features. In particular, the covariance matrices $\Sigma_x$ and $\Sigma_w$ have an autoregressive structure. That is, $(\Sigma_x)_{ij} = 0.9^{|i-j|}$ for any $i,j \in \{1, \dots, p\}$ and $(\Sigma_w)_{ij} = 0.9^{|i-j|}$ for any $i,j \in \{1, \dots, q\}$. Moreover, the deterministic matrix $P$ in Assumption \ref{ass:unobscov} is such that $P_{ij} = 1$ if $i=j$ and $0$ otherwise. 

\begin{figure}[!h]
\centering
    \includegraphics[width=16cm]{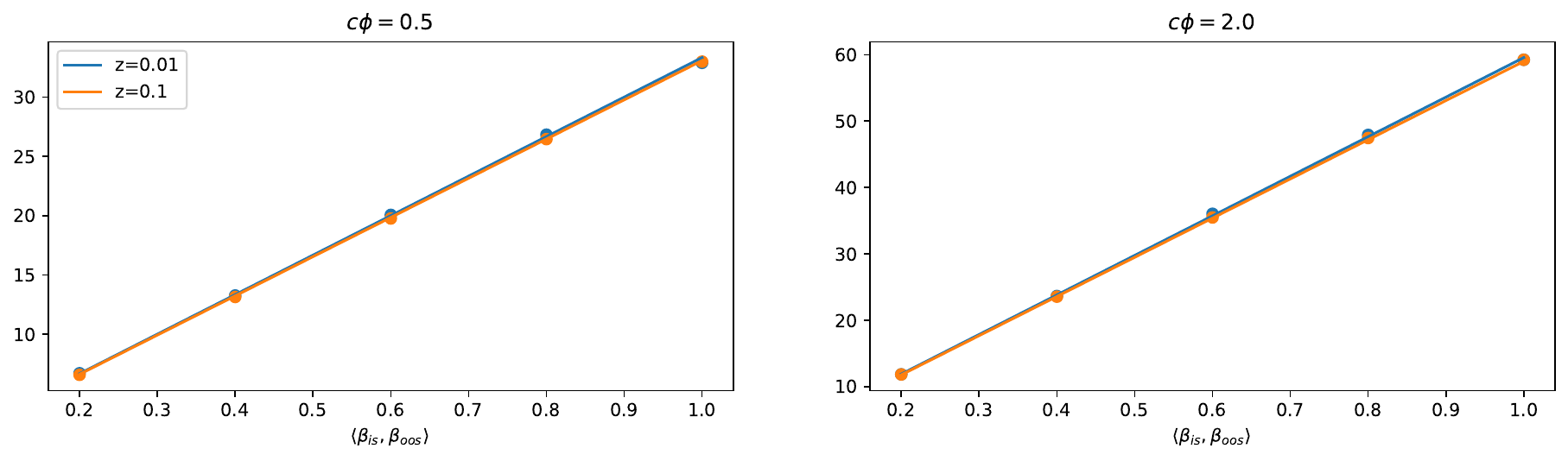}  
\caption{\textbf{Simulated and theoretical returns of the strategy under mild regularization (equidistributed (true) parameters vectors).} \small The model's true complexity is $c=3$ ($p+q = 300$, $n=100$). $z$ is the regularization parameter of the ridge estimator. We plot the theoretical returns of the strategy as per Proposition \ref{prop:expretgeneral} (solid line). Simulated returns (dots) are generated according to the process described in Section \ref{subsec:numanalysis_correlated_equidist} and are averaged over 100,000 random draws. \label{fig:retmildequigeneral}}
\end{figure}

\begin{figure}[!h]
\centering
    \includegraphics[width=16cm]{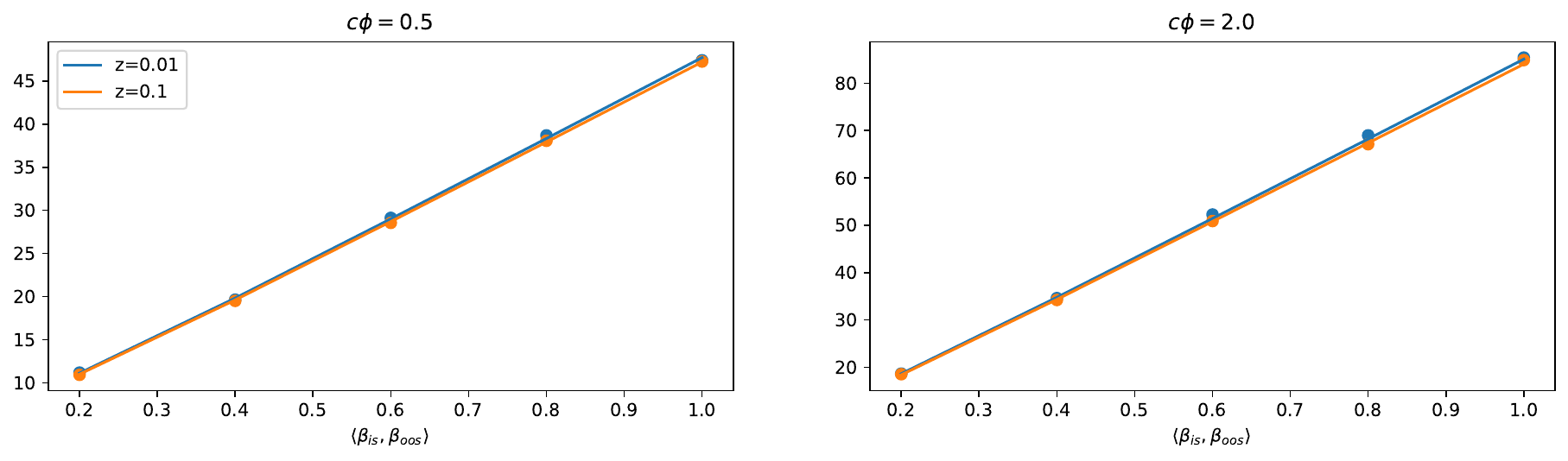}  
\caption{\textbf{Simulated and theoretical volatility of the strategy under mild regularization (equidistributed (true) parameters vectors).} \small The model's true complexity is $c=3$ ($p+q = 300$, $n=100$). $z$ is the regularization parameter of the ridge estimator. We plot the theoretical volatility of the strategy as per Proposition \ref{prop:volstratret} (solid line). Simulated returns (dots) are generated according to the process described in Section \ref{subsec:numanalysis_correlated_equidist} and the sample volatility is computed over 100,000 random draws. \label{fig:volmildequigeneral}}
\end{figure}

\begin{figure}[!h]
\centering
    \includegraphics[width=16cm]{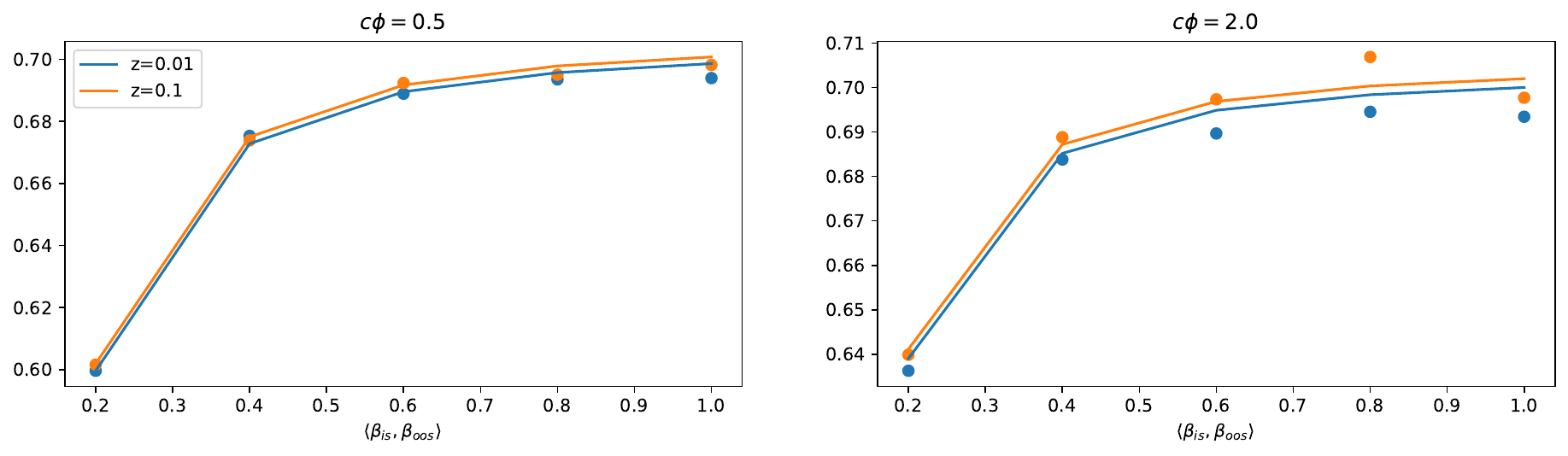}  
\caption{\textbf{Simulated and theoretical Sharpe ratio of the strategy under mild regularization (equidistributed (true) parameters vectors).} \small The model's true complexity is $c=3$ ($p+q = 300$, $n=100$). $z$ is the regularization parameter of the ridge estimator. We plot the theoretical Sharpe ratio of the strategy as the ratio of the expected returns displayed in Proposition \ref{prop:expretgeneral} and the square root of their variance as provided in Proposition \ref{prop:volstratret} (solid line). Simulated returns (dots) are generated according to the process described in Section \ref{subsec:numanalysis_correlated_equidist} and the average  averaged over 100,000 random draws.  \label{fig:srmildequigeneral}}
\end{figure}

\begin{figure}[!h]
\centering
    \includegraphics[width=16cm]{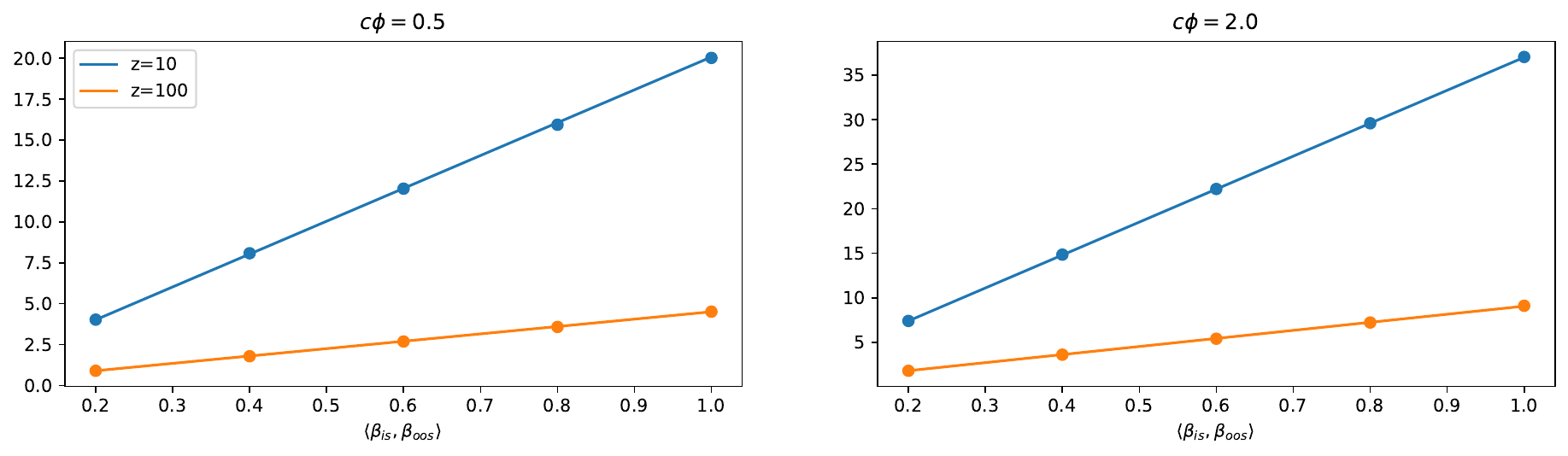}  
\caption{\textbf{Simulated and theoretical returns of the strategy under strong regularization (equidistributed (true) parameters vectors).} \small The model's true complexity is $c=3$ ($p+q = 300$, $n=100$). $z$ is the regularization parameter of the ridge estimator. We plot the theoretical returns of the strategy as per Proposition \ref{prop:expretgeneral} (solid line). Simulated returns (dots) are generated according to the process described in Section \ref{subsec:numanalysis_correlated_equidist} and are averaged over 100,000 random draws. \label{fig:retstrongequigeneral}}
\end{figure}

\begin{figure}[!h]
\centering
    \includegraphics[width=16cm]{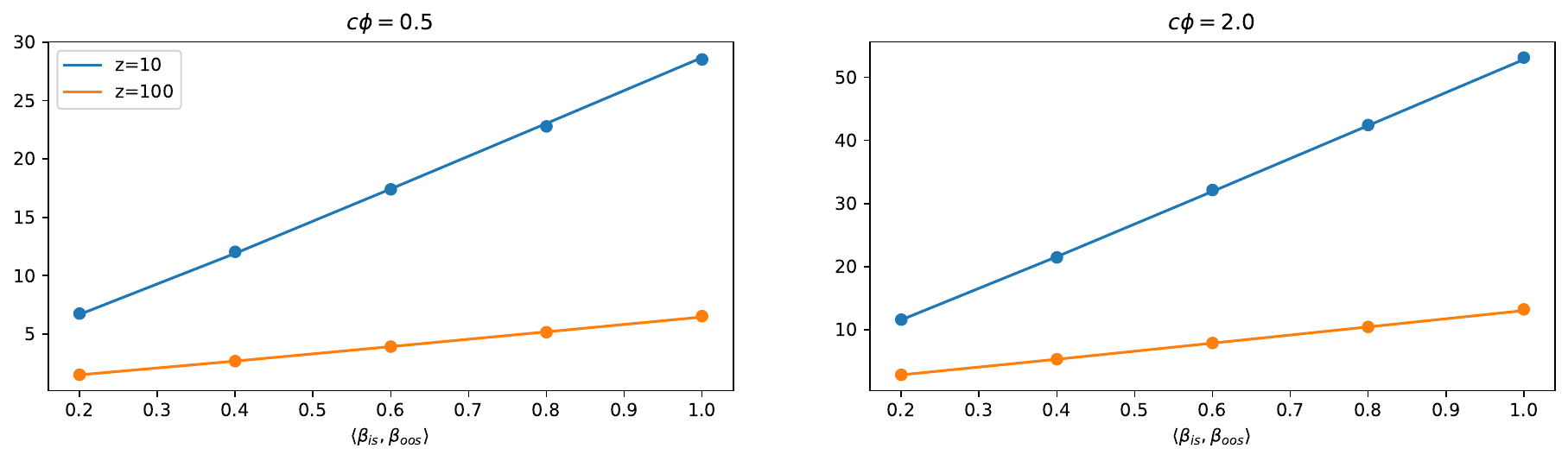}  
\caption{\textbf{Simulated and theoretical volatility of the strategy under strong regularization (equidistributed (true) parameters vectors).} \small The model's true complexity is $c=3$ ($p+q = 300$, $n=100$). $z$ is the regularization parameter of the ridge estimator. We plot the theoretical volatility of the strategy as per Proposition \ref{prop:volstratret} (solid line). Simulated returns (dots) are generated according to the process described in Section \ref{subsec:numanalysis_correlated_equidist} and the sample volatility is computed over 100,000 random draws. \label{fig:volstrongequigeneral}}
\end{figure}

\begin{figure}[!h]
\centering
    \includegraphics[width=16cm]{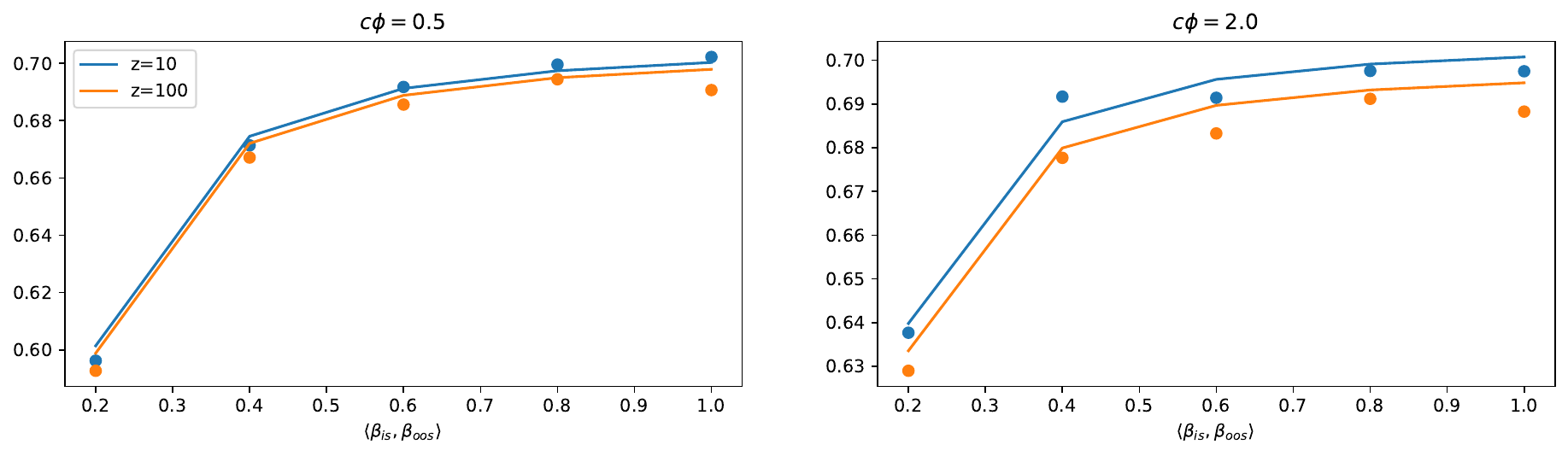}  
\caption{\textbf{Simulated and theoretical Sharpe ratio of the strategy under strong regularization (equidistributed (true) parameters vectors).} \small The model's true complexity is $c=3$ ($p+q = 300$, $n=100$). $z$ is the regularization parameter of the ridge estimator. We plot the theoretical Sharpe ratio of the strategy as the ratio of the expected returns displayed in Proposition \ref{prop:expretgeneral} and the square root of their variance as provided in Proposition \ref{prop:volstratret} (solid line). Simulated returns (dots) are generated according to the process described in Section \ref{subsec:numanalysis_correlated_equidist} and the average  averaged over 100,000 random draws.  \label{fig:srstrongequigeneral}}
\end{figure}

\clearpage

\subsection{Concentrated parameters vector}
\label{subsec:numanalysis_correlated_conc}

We simulate the expected moments of the market timing strategy following the process described in Section \ref{subsec:simretconcentrated} but introducing correlation among the features. As in Section \ref{subsec:numanalysis_correlated_equidist}, the covariance matrices $\Sigma_x$ and $\Sigma_w$ have an autoregressive structure. That is, $(\Sigma_x)_{ij} = 0.9^{|i-j|}$ for any $i,j \in \{1, \dots, p\}$ and $(\Sigma_w)_{ij} = 0.9^{|i-j|}$ for any $i,j \in \{1, \dots, q\}$. Moreover, the deterministic matrix $P$ in Assumption \ref{ass:unobscov} is such that $P_{ij} = 1$ if $i=j$ and $0$ otherwise. 

\begin{figure}[!h]
\centering
    \includegraphics[width=16cm]{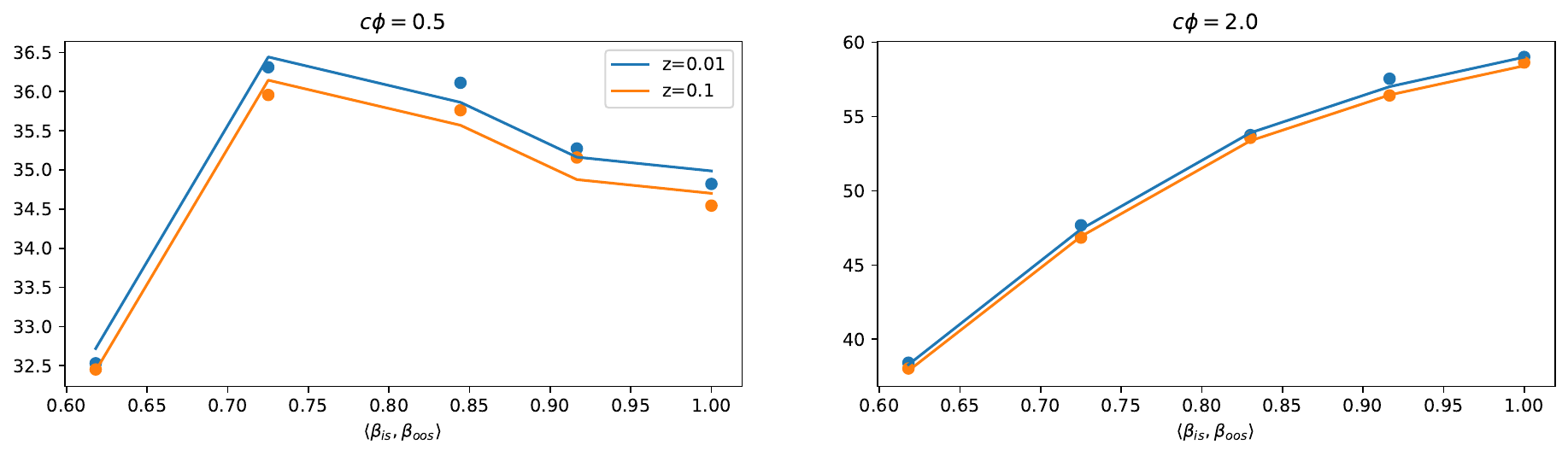}  
\caption{\textbf{Simulated and theoretical returns of the strategy under mild regularization (concentrated (true) parameters vectors).} \small The model's true complexity is $c=3$ ($p+q = 300$, $n=100$). $z$ is the regularization parameter of the ridge estimator. We plot the theoretical returns of the strategy as per Proposition \ref{prop:expretgeneral} (solid line). Simulated returns (dots) are generated according to the process described in Section \ref{subsec:numanalysis_correlated_conc} and are averaged over 100,000 random draws. \label{fig:retmildconcgeneral}}
\end{figure}

\begin{figure}[!h]
\centering
    \includegraphics[width=16cm]{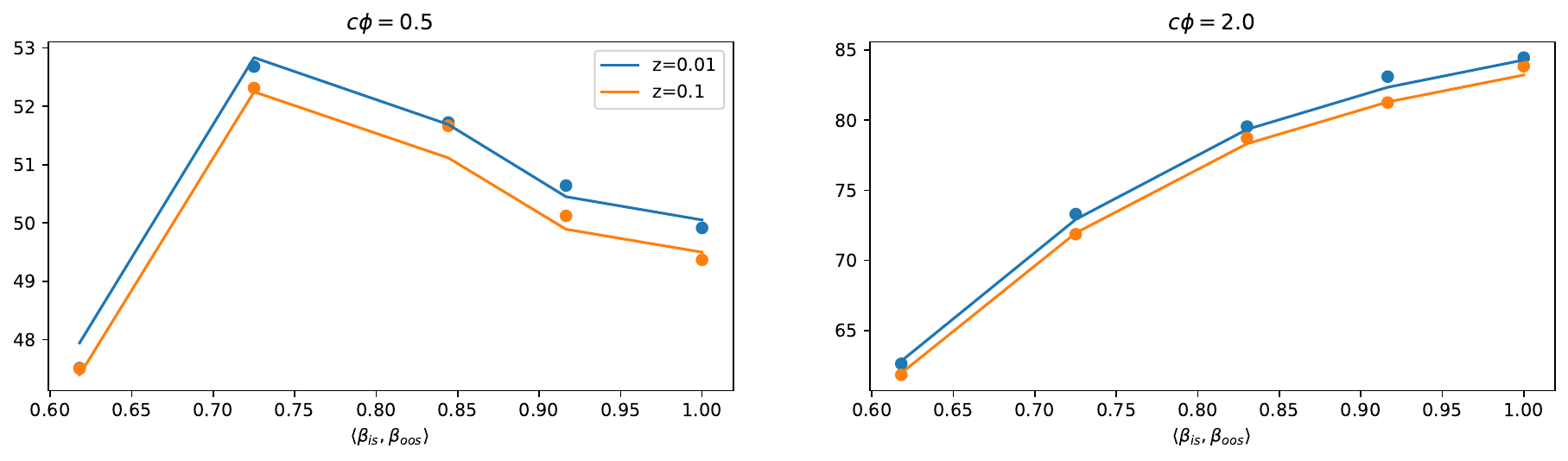}  
\caption{\textbf{Simulated and theoretical volatility of the strategy under mild regularization (concentrated (true) parameters vectors).} \small The model's true complexity is $c=3$ ($p+q = 300$, $n=100$). $z$ is the regularization parameter of the ridge estimator. We plot the theoretical volatility of the strategy as per Proposition \ref{prop:volstratret} (solid line). Simulated returns (dots) are generated according to the process described in Section \ref{subsec:numanalysis_correlated_conc} and the sample volatility is computed over 100,000 random draws. \label{fig:volmildconcgeneral}}
\end{figure}

\begin{figure}[!h]
\centering
    \includegraphics[width=16cm]{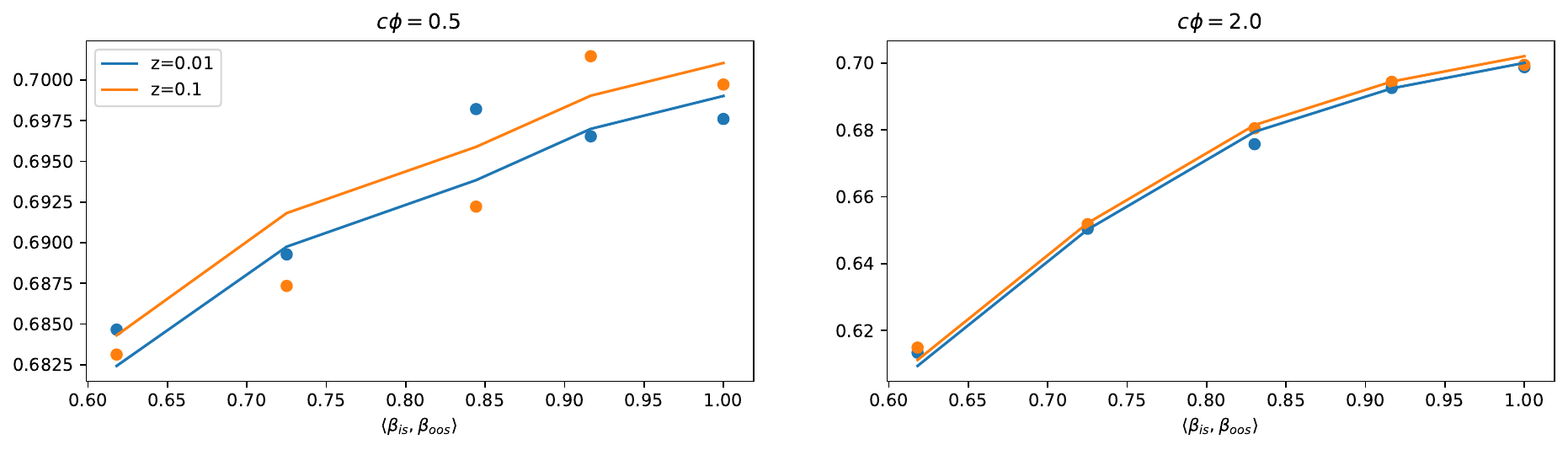}  
\caption{\textbf{Simulated and theoretical Sharpe ratio of the strategy under mild regularization (concentrated (true) parameters vectors).} \small The model's true complexity is $c=3$ ($p+q = 300$, $n=100$). $z$ is the regularization parameter of the ridge estimator. We plot the theoretical Sharpe ratio of the strategy as the ratio of the expected returns displayed in Proposition \ref{prop:expretgeneral} and the square root of their variance as provided in Proposition \ref{prop:volstratret} (solid line). Simulated returns (dots) are generated according to the process described in Section \ref{subsec:numanalysis_correlated_conc} and the average  averaged over 100,000 random draws.  \label{fig:srmildconcgeneral}}
\end{figure}

\begin{figure}[!h]
\centering
    \includegraphics[width=16cm]{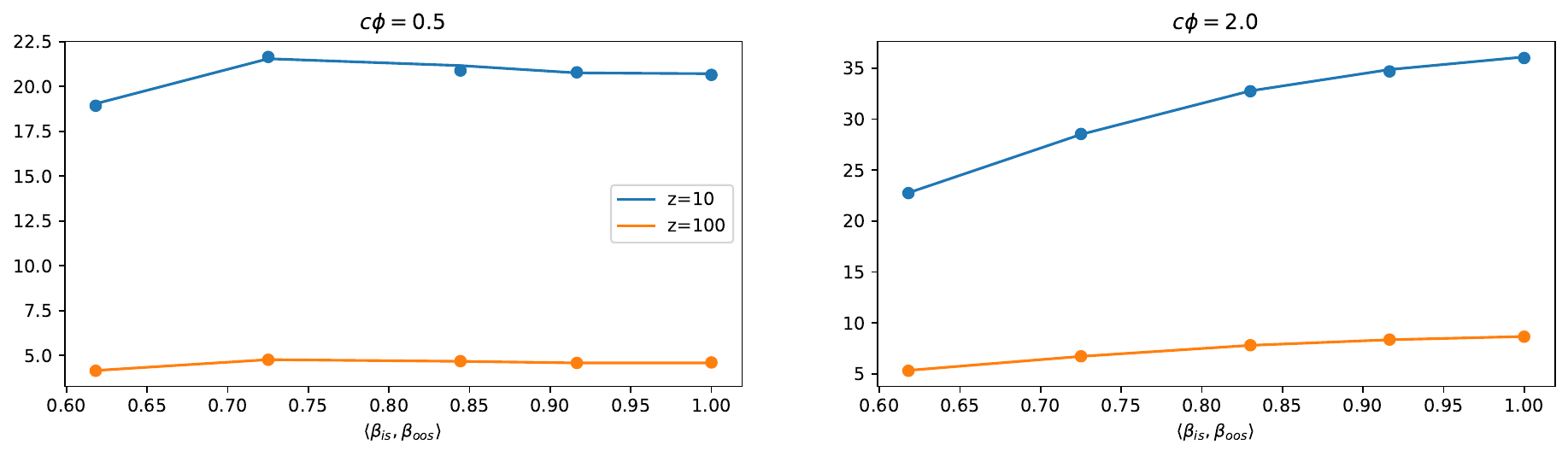}  
\caption{\textbf{Simulated and theoretical returns of the strategy under strong regularization (concentrated (true) parameters vectors).} \small The model's true complexity is $c=3$ ($p+q = 300$, $n=100$). $z$ is the regularization parameter of the ridge estimator. We plot the theoretical returns of the strategy as per Proposition \ref{prop:expretgeneral} (solid line). Simulated returns (dots) are generated according to the process described in Section \ref{subsec:numanalysis_correlated_conc} and are averaged over 100,000 random draws. \label{fig:retstrongconcgeneral}}
\end{figure}

\begin{figure}[!h]
\centering
    \includegraphics[width=16cm]{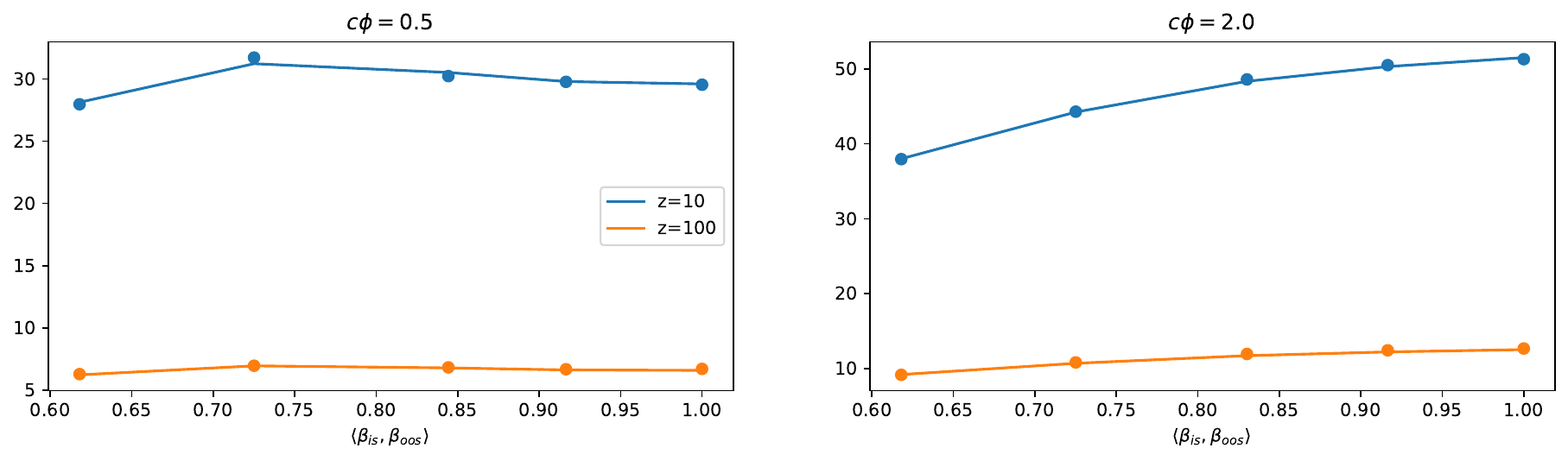}  
\caption{\textbf{Simulated and theoretical volatility of the strategy under strong regularization (concentrated (true) parameters vectors).} \small The model's true complexity is $c=3$ ($p+q = 300$, $n=100$). $z$ is the regularization parameter of the ridge estimator. We plot the theoretical volatility of the strategy as per Proposition \ref{prop:volstratret} (solid line). Simulated returns (dots) are generated according to the process described in Section \ref{subsec:numanalysis_correlated_conc} and the sample volatility is computed over 100,000 random draws. \label{fig:volstrongconcgeneral}}
\end{figure}

\begin{figure}[!h]
\centering
    \includegraphics[width=16cm]{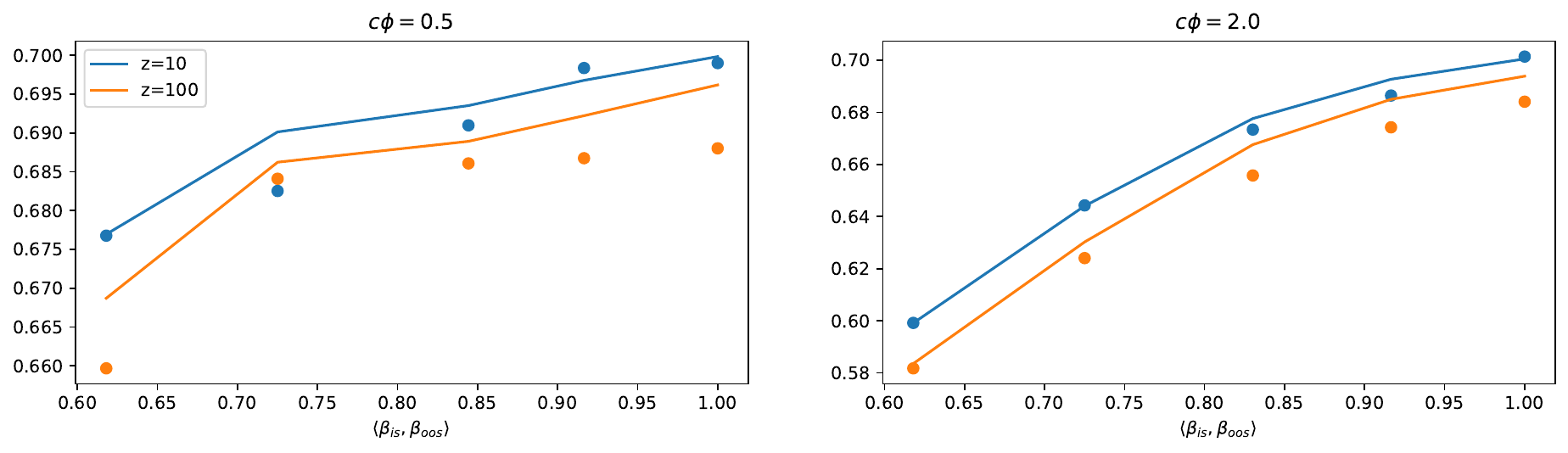}  
\caption{\textbf{Simulated and theoretical Sharpe ratio of the strategy under strong regularization (concentrated (true) parameters vectors).} \small The model's true complexity is $c=3$ ($p+q = 300$, $n=100$). $z$ is the regularization parameter of the ridge estimator. We plot the theoretical Sharpe ratio of the strategy as the ratio of the expected returns displayed in Proposition \ref{prop:expretgeneral} and the square root of their variance as provided in Proposition \ref{prop:volstratret} (solid line). Simulated returns (dots) are generated according to the process described in Section \ref{subsec:numanalysis_correlated_conc} and the average  averaged over 100,000 random draws.  \label{fig:srstrongconcgeneral}}
\end{figure}


\clearpage

\section{Counterfactual returns with strong regularization}

\begin{figure}[!h]
\centering
    \includegraphics[width=16cm]{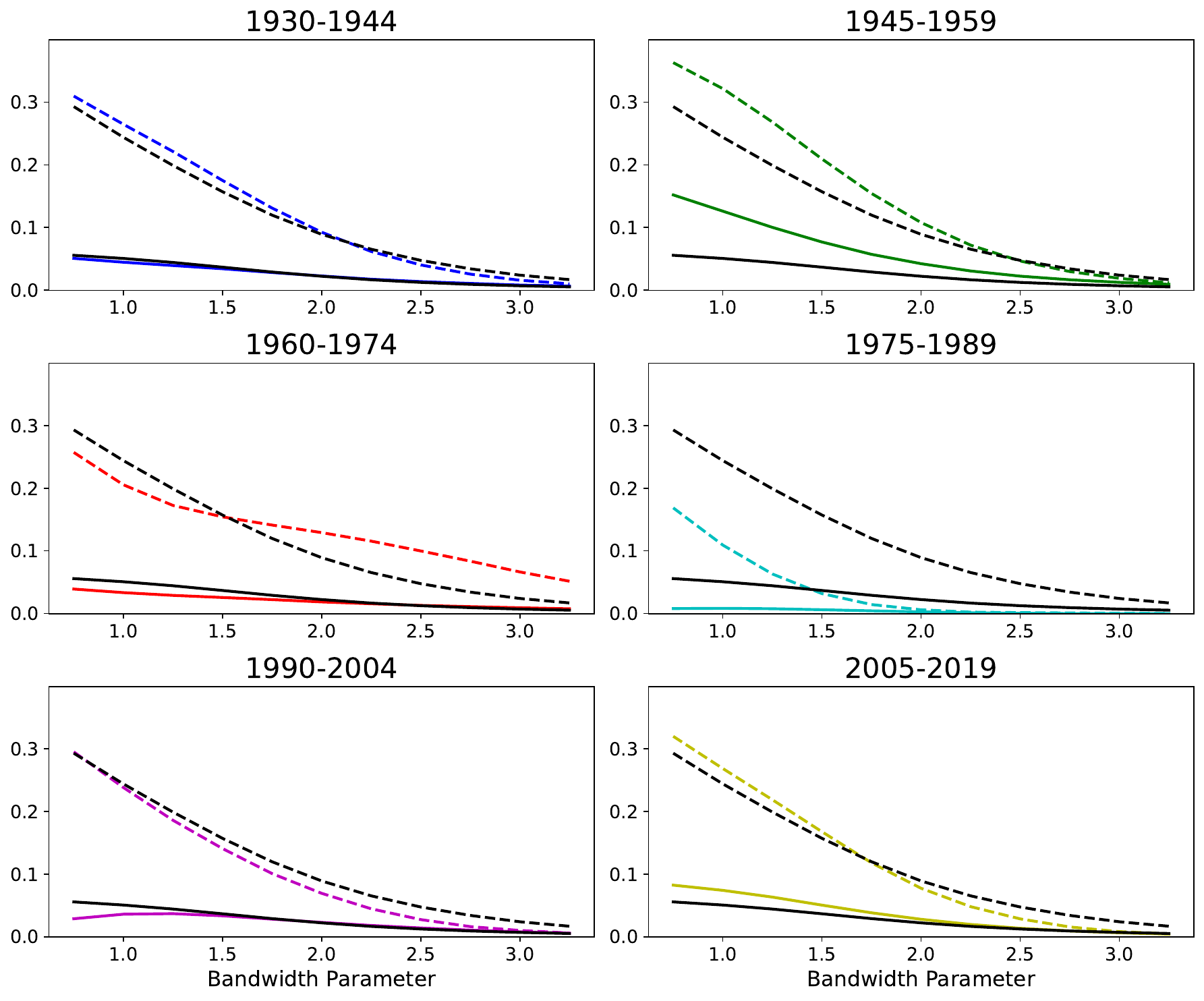}
\caption{\textbf{Portfolio returns without concept drift under strong regularization.} \small We plot the real (with concept drift, solid lines) and the counterfactual (without concept drift, dashed lines) returns of the strategy obtained on six sub-periods for $P=600$ RFFs, as a function of the bandwidth parameter of the kernel from the Random Fourier Features. Colored lines represent the strategy returns for the corresponding sub-periods, while \textbf{black lines} show returns over the full sample. The penalization parameter is $z=100$ (strong regularization). Returns are averaged over 500 random draws of RFFs. \label{fig:cf_strongreg}}
\end{figure}

\clearpage

\section{Proof of Proposition \ref{lem:2}}
\label{app:proof_l2}

\subsection{Preliminary results and discussions}

We start with a few useful results that correspond to the correctly specified model ($\theta=0$) in the isotropic case $\Sigma = I$. This will give us building blocks for subsequent calculations. 

\begin{lemma}
\label{lem:1}
The estimator $\beta_{\text{is}}$ in \eqref{eq:estimator} has out-of-sample bias and variance (defined in Equation \eqref{eq:R_X}):
\begin{align}
    B_X(\hat{\beta}_{\textnormal{is}}, \beta_{\textnormal{oos}})&=  {\beta}_{\textnormal{is}}' \Pi \Sigma \Pi {\beta}_{\textnormal{is}}+ ({\beta}_{\textnormal{oos}}-{\beta}_{\textnormal{is}})'\Sigma ({\beta}_{\textnormal{oos}}-{\beta}_{\textnormal{is}})+2 ({\beta}_{\textnormal{oos}}-{\beta}_{\textnormal{is}})'\Sigma \Pi {\beta}_{\textnormal{is}} \\
    V_X(\hat{\beta}_{\textnormal{is}}, \beta_{\textnormal{oos}})&=n^{-1}\sigma^2 \textnormal{Tr}(\hat{\Sigma}^+\Sigma),
\end{align}
where $\Pi=I-\hat{\Sigma}^+\hat{\Sigma}$, $I$ being the identity matrix.
\end{lemma}
This can be proven as in \cite{hastie2022surprises}, notably by plugging $\E[\hat{\beta}_u|X]=\hat{\Sigma}^+\hat{\Sigma}\beta_u$ into Equation \eqref{eq:R_X}. Compared to the i.i.d. case in which $\beta$ remains the same for both samples, only the bias is changed and has new additional terms that reflect the drift between $\beta_{\text{is}}$ and $\beta_{\text{oos}}$. Asymptotically, as both dimensions increase to infinity such that $p/n \rightarrow \gamma \in (0,1)$, it will hold that $\Pi=0$ as long as all eigenvalues are bounded from below by some strictly positive constant. Thus, compared to Proposition 2 in \cite{hastie2022surprises}, it holds that

\begin{equation}
\label{eq:limup}
    \underset{n,p \rightarrow \infty}{\lim}R_X(\hat{\beta}_{\text{is}},\beta_{\text{oos}}) = \sigma^2 \frac{\gamma}{1-\gamma} + \|{\beta}_{\textnormal{oos}}-{\beta}_{\textnormal{is}}\|^2_\Sigma, \quad p/n \rightarrow \gamma \in (0,1).
\end{equation}

Hence, we clearly see that in contrast with the case when $\beta$ does not change, there is an additional term that pertains to the change between the training loadings ($\beta_{\text{is}}$) and the testing ones ($\beta_\text{oos}$). This quantity is strictly positive if $\Sigma$ is positive definite and it increases with each element-wise absolute distance $|\beta_{\text{is},p}-\beta_{\text{oos},p}|$ when correlations are positive, or when $\Sigma$ is diagonally dominant.

We next turn the overparametrized regime in which the number of predictors lies asymptotically \textit{above} the number of observations, i.e., $\gamma >1$. This case is much more involved and for simplicity, we will only tackle it in the simplified framework in which predictors are independent (and have common variance), i.e., when $\Sigma = I$ - this is commonly referred to as the \textit{isotropic} feature case. 

In this particular case, the variance term is the same as in the original article, that is, $V_X(\hat{\beta}_{\textnormal{is}}, \beta_{\textnormal{oos}})=\sigma^2(\gamma-1)^{-1}$. However, the bias term is directly impacted and, in addition to the original term $\|\beta_{\text{is}} \|^2_2(1-\gamma^{-1})$, there are two quantities, namely $\|\beta_{\text{is}}-\beta_\text{oos} \|^2_2$ and $2({\beta}_{\textnormal{oos}}-{\beta}_{\textnormal{is}})' \Pi {\beta}_{\textnormal{is}} \rightarrow 2(1-\gamma^{-1})({\beta}_{\textnormal{oos}}-{\beta}_{\textnormal{is}})'{\beta}_{\textnormal{is}} $.
The first one is obvious from the second term of the bias in Lemma \ref{lem:1} when $\Sigma=I$. The second limiting value can be determined exactly as in the proof of Theorem 1 in \cite{hastie2022surprises}. Adding all three terms gives
\begin{align}
\underset{n,p \rightarrow \infty}{\lim}R_X(\hat{\beta}_{\text{is}},\beta_{\text{oos}}) &=\sigma^2(\gamma-1)^{-1} +\gamma^{-1}\|\beta_{\text{is}}\|_2^2+\|\beta_{\text{oos}}\|_2^2 -2\gamma^{-1}\beta_{\text{oos}}'\beta_{\text{is}} \\
&= \sigma^2(\gamma-1)^{-1} +\gamma^{-1}\|\beta_{\text{is}}-\beta_{\text{oos}} \|^2_2 + (1-\gamma^{-1}) \|\beta_{\text{oos}} \|_2^2,
\label{eq:limop}
\end{align}
where in the second equation all terms are positive. In particular, it is clear that the risk increases with $|\beta_{\text{is},p}-\beta_{\text{oos},p}|$ which measures the absolute change in the $p^{th}$ loading of the model. Intuitively, the more the data generating process changes, the larger the out-of-sample error. The general case involves the empirical distribution of the eigenvalues of $\Sigma$ and is out of the scope of the present paper, but it can be conjectured that the limiting risk will also involve the change term $\Delta = \beta_{\text{is}}-\beta_{\text{oos}}$.
 
\subsection{The proof}

Let $(x_0,w_0)$ be a random draw of $P_{x,w}$ that is independent from the training data.
\begin{align*}
    R_X^m(\hat{\beta}_{\text{is}},\beta_{\text{oos}}, \theta) &= \E[(x_0\hat{\beta}_{\text{is}} - \E[y_0|x_0,w_0])^2|X] \\
    &= \E[(x_0\hat{\beta}_{\text{is}} - \E[y_0|x_0] + \E[y_0|x_0] - \E[y_0|x_0,w_0])^2|X] \\
    &= \E[(x_0\hat{\beta}_{\text{is}} - \E[y_0|x_0])^2|X] + \E[(\E[y_0|x_0] - \E[y_0|x_0,w_0])^2] \\ &\quad + 2 \E[(x_0\hat{\beta}_{\text{is}} - \E[y_0|x_0])(\E[y_0|x_0] - \E[y_0|x_0,w_0])|X] \\
\end{align*}

In addition, we have $\E[y_0|x_0] = x_0' \beta_{\text{oos}} + \E[w_0'\theta_{oos}|x_0]$ and $\E[y_0|x_0,w_0] = x_0' \beta_{\text{oos}} + w_0'\theta$.

Focusing on the cross-term, note that by the law of iterated expectations we have

\begin{align*}
    & \quad \E[(x_0\hat{\beta}_{\text{is}} - \E[y_0|x_0])(\E[y_0|x_0] - \E[y_0|x_0,w_0])|X] \\
    &= \E[(x_0\hat{\beta}_{\text{is}} - x_0' \beta_{\text{oos}} - \E[w_0'\theta_{oos}|x_0])( \E[w_0'\theta_{oos}|x_0] - w_0'\theta_{oos})|X] \\
    & = \E \Big[ \E[(x_0\hat{\beta}_{\text{is}} - x_0' \beta_{\text{oos}} - \E[w_0'\theta_{oos}|x_0])( \E[w_0'\theta_{oos}|x_0] - w_0'\theta_{oos})|x_0,X] \big|X \Big] \\
    & = \E \Big[ (x_0\hat{\beta}_{\text{is}} - x_0' \beta_{\text{oos}} - \E[w_0'\theta_{oos}|x_0])( \E[w_0'\theta_{oos}|x_0] - \E[w_0'\theta_{oos}|x_0,X]) \big|X \Big]
\end{align*}

$w_0$ is independent from $X$. Then, $\E[w_0'\theta_{oos}|x_0,X] = \E[w_0'\theta_{oos}|x_0]$ and the cross term vanishes. Therefore,

\begin{equation*}
    R_X^m(\hat{\beta}_{\text{is}},\beta_{\text{oos}}, \theta) = \E[(x_0'(\hat{\beta}_{\text{is}}-\beta_{\text{oos}}))^2|X] + \E[(w_0'\theta_{oos} - \E[w_0'\theta_{oos}|x_0])^2]=R_X(\hat{\beta}_{\text{\text{is}}},\beta_{\text{oos}})+M(\theta_{oos})
\end{equation*}
The two cases (under- versus over-parametrized models) are obtained thanks to Equations \eqref{eq:limup} and \eqref{eq:limop}. In fact, when covariates are i.i.d. with unit variance, the training labels can be thought to have variance $\sigma^2 + \|\theta_{is}\|_2^2$. Then, it suffices to plug $\sigma^2 \to \sigma^2 + \|\theta_{is}\|_2^2$ into \eqref{eq:limup} or \eqref{eq:limop} and add the misspecification bias $M(\theta_{oos}) = \|\theta_{oos}\|_2^2$ to obtain the desired prediction risk. This completes the proof.

\section{Asymptotic moments of the strategy's return}
\label{subsec:app_expret}

Let $\lambda_1 \geq \lambda_2 \geq \dots \geq \lambda_{p} \geq 0$ be the eigenvalues of $\Sigma_x$. Denote by $\delta_{\lambda_i}$ a unit mass concentrated at $\lambda_i$. Then, the empirical spectral distribution (ESD) of $\Sigma_x$ is given by

\begin{equation}
\label{eq:esd}
   \hat{\mu} \equiv \hat{\mu}(\lambda) = \frac{1}{p} \sum_{i=1}^p \delta_{\lambda_i}
\end{equation}

For $ i \in \{ 1, \dots , p \}$, denote by $v_i$ the eigenvector of $\Sigma_x$ associated with the eigenvalue $\lambda_i$. Then, we also define the 

\begin{align}
\label{eq:rewesdcomp}
    \begin{split}
    &\hat{\zeta}_{is} := \frac{1}{\|\beta_{is}\|^2} \sum_{i=1}^p | \langle \beta_{is} , v_i \rangle |^2 \delta_{\lambda_i} \\
    &\hat{\zeta}_{oos} := \frac{1}{\|\beta_{oos}\|^2} \sum_{i=1}^p | \langle \beta_{oos} , v_i \rangle |^2 \delta_{\lambda_i} \\
    &\hat{\zeta}_r := \left\{ \begin{array}{c l}
    \frac{1}{\|\beta_{oos} - \beta_{is}\|^2} \sum_{i=1}^p | \langle \beta_{oos} - \beta_{is} , v_i \rangle |^2 \delta_{\lambda_i} \quad &\beta_{is} \neq \beta_{oos}\\
    0 \quad &\beta_{is} = \beta_{oos}
    \end{array} \right. 
    \end{split}
\end{align}

where the subscript $r$ in $\hat{\zeta}_r$ stands for "residual". Note that these three quantities are density measures. In particular, they are reweighted counterparts of $\hat{\mu}$ where the summands account for the projections of $\beta_{is}$, $\beta_{oos}$ or the residual vector $\beta_{oos} - \beta_{is}$  on the eigenvectors of $\Sigma_x$. These densities are often referred to as the eigenvector empirical spectral distributions (the so-called VESDs) or local densities of the states $\beta_{is}$, $\beta_{oos}$ and $\beta_{oos} - \beta_{is}$, respectively. 

Finally, we set 

\begin{align}
\label{eq:rewesddrift}
    \hat{\zeta}_d \equiv \hat{\zeta}_d (\lambda) := \frac{1}{\|\beta_{is}\| \|\beta_{oos}\|} \sum_{i=1}^p \langle \beta_{is} , v_i \rangle \langle v_i, \beta_{oos} \rangle  \delta_{\lambda_i},
\end{align}

where the subscript $d$ stands for "drift". We make two observations regarding $\hat{\zeta}_d$. First, it is a signed measure and no longer a density. Second, it can be decomposed as a function of $\hat{\zeta}_{is}$, $\hat{\zeta}_{oos}$ and $\hat{\zeta}_r$ as follows:

\begin{equation*}
    \hat{\zeta}_d = \frac{1}{2\|\beta_{is}\| \|\beta_{oos}\|} \big( \|\beta_{is}\|^2 \hat{\zeta}_{is} + \|\beta_{oos}\|^2 \hat{\zeta}_{oos} - \|\beta_{oos} - \beta_{is}\|^2 \hat{\zeta}_r \big).
\end{equation*}

As is customary in the literature, we will need a few technical assumptions.

\begin{assump}
\label{ass:opbound}
    $\lambda_1 = \| \Sigma_x \|_{\text{op}} \leq \tau^{-1}$.
\end{assump} 

\begin{assump}
\label{ass:accumulationzero}
    $\hat{\mu}([0,\tau]) \leq 1 - \tau$.
\end{assump}

Condition \ref{ass:opbound} and \ref{ass:accumulationzero} require the eigenvalues of $\Sigma_x$ to be bounded and not to be concentrated at zero.\footnote{As noted in \cite{hastie2020surprises}, this assumption could be relaxed by requiring only $d < p$ of the eigenvalues to be non-vanishing and by consequently redefining $c\phi = d/n$. However, we pursue with the stronger hypothesis to avoid cumbersome notations in the statement of the results and proofs.} 

\begin{assump}
\label{ass:probacvmumis}
    $\hat{\mu}$ weakly converges to a density measure $\mu$. Besides, $supp(\mu) \subseteq [0,\infty)$.
\end{assump}

\begin{assump}
\label{ass:probacvpimis}
    $\hat{\zeta}_{is}$, $\hat{\zeta}_{oos}$ and $\hat{\zeta}_{r}$ weakly converges to some density measures $\zeta_{is}$, $\zeta_{oos}$ and $\zeta_{r}$ respectively. Moreover, these limiting measures have support on $[0,\infty)$.
\end{assump}

Finally, the latter two assumptions ensure that the densities defined in \eqref{eq:esd} and \eqref{eq:rewesdcomp} are well-behaved in the asymptotic proportional regime, i.e. when $n,p \rightarrow \infty$ such that $p/n \rightarrow c\phi$.

\begin{remark}
    Assumption \ref{ass:probacvpimis} implies that 
    
    \begin{align*}
        \hat{\zeta}_{d} = \frac{1}{2\|\beta_{is}\| \|\beta_{oos}\|} \big( \|\beta_{is}\|^2 \hat{\zeta}_{is} + \|\beta_{oos}\|^2 \hat{\zeta}_{oos} - \|\beta_{oos} - \beta_{is}\|^2 \hat{\zeta}_{r} \big)
    \end{align*}

    weakly converges to a (possibly signed) measure $\zeta_{d}$.  
\end{remark}

Next, we define the solutions to two different equations, $m(z;c\phi,\hat{\mu})$ and $s_0(c\phi,\hat{\mu})$, that will appear in the subsequent theoretical results. Denote by $\C_+$ the complex upper half plane, i.e. $\C_+ = \{ z \in \C : \Im(z) > 0 \}$. 

\begin{definition}
\label{def:mn}
    Let $c\phi = p/n > 0$ and $z \in \C_+$. Define $m := m(z;c\phi,\hat{\mu})$ as the unique solution in $\C_+$ to

    \begin{equation*}
        m = \int \frac{1}{\lambda(1-c\phi-c\phi zm)-z} d\hat{\mu} ( \lambda).
    \end{equation*}

    This definition is extended to $\Im(z) = 0$ whenever possible. 
    
\end{definition}

\begin{definition}
\label{def:s0}
    Let $s_0 := s_0(c\phi,\hat{\mu})$ be the unique non-negative solution to

    \begin{equation*}
        1 - \frac{1}{c\phi} = \int \frac{1}{1 + \lambda c\phi s_0} d\hat{\mu}(\lambda).
    \end{equation*}
    
\end{definition}

Lastly, we define the following quantity which repeatedly appear in the different characterizations of the expected return of the strategy. Note that it does not depend on the regularization parameter, $z$. 

\begin{definition}
\label{def:retnoreg}
    \begin{align*}
    \begin{split}
        \Bar{H}(\hat{\zeta}_d)
    & = \langle \beta_{oos}, \beta_{is} \rangle_{\Sigma} \\
    &= \|\beta_{is}\| \|\beta_{oos}\| \int \lambda d \hat{\zeta}_d (\lambda) .
    \end{split}
\end{align*}
\end{definition}

\subsection{Well-specified model and ridge regularization}
\label{subsub:wellspecifiedridge}

Here, we consider that the true model is known. Hence, throughout this section, $\theta = 0$ and consequently, $q=0$, $\phi = 1$ and $\Sigma \equiv \Sigma_x$.

Let 

\begin{equation*}
    F_n(z) = -z \langle \beta_{oos}, \Sigma (\hat{\Sigma}+zI)^{-1}\beta_{is} \rangle,
\end{equation*}

and define:

\begin{align*}
    \Bar{F}(z;\hat{\mu}, \hat{\zeta}_d) &= -z \langle \beta_{oos}, \Sigma ( zI + (czm(-z;c,\hat{\mu}) + 1 - c) \Sigma)^{-1} \beta_{is} \rangle \\
    &= - z \|\beta_{is}\| \|\beta_{oos}\| \int \frac{\lambda}{\lambda(cz m(-z;c,\hat{\mu}) + 1 - c) + z} d \hat{\zeta}_d ( \lambda) \\
\end{align*}

where $m := m(-z;c,\hat{\mu})$ is given in Definition \ref{def:mn}.

Moreover, consider the following hypothesis.

We now state a number of lemmas that will be useful in the remainder of the proofs. In particular, Lemma $\ref{lem:deterministic_equivalent}$ establishes non-asymptotic bounds for the gap between $F_n$ and $\Bar{F}$ which implies almost sure convergence as proved in Lemma $\ref{lem:asympcv}$.

\begin{lemma}
\label{lem:deterministic_equivalent}
    Let $n^{-2/3+\tau} \leq z \leq \tau^{-1}$ and Assumptions (\ref{ass:complexity}), (\ref{ass:opbound}), (\ref{ass:accumulationzero}) and (\ref{ass:iddistcov}) hold. Then, for any $D>0$ and $\varepsilon>0$, there exist $n_0 := n_0(\varepsilon,D)$ such that for $n \geq n_0$, 

    \begin{equation}
    \label{eq:ineqF}
        |F_n(z) - \Bar{F}(z;\hat{\mu}, \hat{\zeta}_d)| \leq \frac{\| \Sigma \beta_{oos}\| \|\beta_{is}\|}{n^{\frac{1}{2}-\varepsilon}z}
    \end{equation}

holds with probability $1-n^{-D}$. 
\end{lemma}

\begin{proof}
    The proof mainly consists in an application of the anisotropic local law derived by \cite{knowles2017anisotropic}.

    First, let $\Tilde{\beta}_{(oos, \Sigma)} = (\Sigma \beta_{oos}) / \| \Sigma \beta_{oos}\| = (\Sigma \beta_{oos}) / \| \beta_{oos}\|_{\Sigma}$ and $\Tilde{\beta}_{is} = \beta_{is} / \| \beta_{is}\|$. Furthermore, define 
    \begin{align*}
        &J_n(z) = -z \langle \Tilde{\beta}_{(oos, \Sigma)}, (\hat{\Sigma}+zI)^{-1}\Tilde{\beta}_{is} \rangle \\
        &\Bar{J}(z;\hat{\mu}, \hat{\zeta}_d) = -z \langle \Tilde{\beta}_{(oos, \Sigma)}, ( zI + (czm(-z;c,\hat{\mu}) + 1 - c) \Sigma)^{-1} \Tilde{\beta}_{is} \rangle,
    \end{align*}
    
    which are standardized counterparts of $F_n$ and $\Bar{F}$.

    Define $\mathbf{D}(\tau,n) = \{ z \in \C_+: |z| \geq \tau, |\Re (z)| \leq \tau^{-1}, 0 \leq \Im (z) \leq \tau^{-1} \}$ and $\mathbf{D}^O (\tau,n) = \{ z \in \mathbf{D}(\tau,n) : \Re (z) \leq -n^{-2/3 + \tau} \}$. For the sake of simplicity, we drop the dependence on $\tau$ and $n$ and let $\mathbf{D}^O := \mathbf{D}^O (\tau,n)$. 

    Consider $z \in \C_-$, i.e. $\Im (z) < 0$, such that $\Re (z) \geq n^{-2/3 + \tau}$. Hence, $-z \in \mathbf{D}^O$. Then, as per Theorem 3.16. (i) and Remark 3.17. in \cite{knowles2017anisotropic}, for any $D>0$ and $\varepsilon>0$, there exist $n_0 := n_0(\varepsilon,D)$ such that for $n \geq n_0$, the following holds with probability $1-n^{-D}$

    \begin{equation*}
        |J_n(z) + \langle \Tilde{\beta}_{(oos, \Sigma)}, ( I + r_n(-z) \Sigma)^{-1} \Tilde{\beta}_{is} \rangle| \leq n^\varepsilon \sqrt{\frac{\Im (r_n(-z))}{n \Im (-z) }},
    \end{equation*} 

    with $r_n = r_n(-z)$ the unique solution in $\C_+$ of the equation

    \begin{equation*}
    \label{eq:fixedr_n}
        \frac{1}{r_n} = - z + c \int \frac{\lambda}{1 + r_n \lambda} d\hat{\mu}(\lambda),\quad \Im(-z) > 0.
    \end{equation*}

    where $\lambda_1 \geq \lambda_2 \geq \dots \geq \lambda_p \geq 0$ are the eigenvalues of $\Sigma$.

    Moreover, $r_n(-z)$ is the Stieltjes transform of a probability measure $\rho_n$ with bounded support in $[0,\infty)$.

    Therefore, we have 

    \begin{align*}
        \Im (r_n(-z)) &= \Im \Big( \int \frac{1}{x + z} d\rho_n (x) \Big) = \int \frac{-\Im(z)}{(x + \Re(z))^2 + \Im(z)^2} d\rho_n(x).
    \end{align*}

    This implies 

    \begin{equation*}
        |\Im(r_n(-z))| \leq \frac{|\Im(z)|}{\Re(z)^2}.
    \end{equation*}

    and consequently,

    \begin{equation*}
        |J_n(z) + \langle \Tilde{\beta}_{(oos, \Sigma)}, ( I + r_n(-z) \Sigma)^{-1} \Tilde{\beta}_{is} \rangle| \leq \frac{1}{n^{1/2 - \varepsilon}\Re(z)}\sqrt{\frac{|\Im(z)|}{\Im (-z) }}.
    \end{equation*} 

    Note that $r_n$ is the companion Stieltjes transform of the ESD defined in \eqref{eq:esd}. However, we would like to relate the deterministic equivalent of $J_n(z)$ to the eigenvalue distribution of the covariance matrix, or equivalently to its Stieltjes transform. This is possible thanks to the following (well-known) relation in random matrix theory (see e.g. \cite{silverstein1995strong})

    \begin{equation}
    \label{eq:relationrm}
        r_n(-z) = c m(-z;c,\hat{\mu}) + \frac{1-c}{z}.
    \end{equation}

    Here, $m := m(-z;c,\hat{\mu})$ is the unique solution in $\C_+$ to

    \begin{equation*}
        m = \int \frac{1}{\lambda(1-c+czm)+z} d \hat{\mu} (\lambda), \quad \Im(-z) > 0.
    \end{equation*}

    Replacing $r_n(-z)$ as per Equation \eqref{eq:relationrm} in $-\langle \Tilde{\beta}_{(oos, \Sigma)}, ( I + r_n(-z) \Sigma)^{-1} \Tilde{\beta}_{is} \rangle$, we obtain

    \begin{equation*}
        |J_n(z) - \Bar{J}(z;\hat{\mu}, \hat{\zeta}_d)| \leq \frac{1}{n^{1/2 - \varepsilon}\Re(z)}\sqrt{\frac{|\Im(z)|}{\Im (-z) }}.
    \end{equation*} 

    Taking $\Im(z) \rightarrow 0$, we get the following inequality with probability $1-n^{-D}$:

    \begin{equation*}
        |J_n(z) - \Bar{J}(z;\hat{\mu}, \hat{\zeta}_d)| \leq \frac{1}{n^{1/2 - \varepsilon}z}, \quad z \in [n^{-2/3+\tau}, \tau^{-1}],
    \end{equation*} 

    for any sufficiently large $n$.

    Finally, multiplying both sides of the inequality by $\| \Sigma \beta_{oos}\| \| \beta_{is}\|$, we obtain \eqref{eq:ineqF}.

    The proof of Lemma \ref{lem:deterministic_equivalent} is complete.
    
\end{proof}

\begin{lemma}
\label{lem:baiyin}
    
    (\cite{yin1988limit}, \cite{bai2008limit})  Let $Y$ be a $n \times p$ matrix with i.i.d. entries, $(y_{ij})_{1 \leq i \leq n, 1 \leq j \leq p}$. Besides, for all $i,j$, $\E[y_{ij}] = 0$, $\E[y_{ij}^2] = 1$ and $\E[|y_{ij}|^{4+a}] \leq C < \infty$ for $a > 0$. Then, in the limit, as $p,n \rightarrow \infty$, $p/n \rightarrow c \in (0,\infty)$,

    \begin{align*}
        \lim \sup \quad \lambda_{\text{max}}\Big(\frac{1}{n} Y'Y - (1+c)I\Big) \leq 2 \sqrt{c}, \quad a.s.
    \end{align*}
    
\end{lemma}

\begin{lemma}
\label{lem:asympcv}
    Let $z > \tau$ and Assumptions (\ref{ass:iddistcov}), (\ref{ass:complexity}), (\ref{ass:opbound}), (\ref{ass:accumulationzero}), (\ref{ass:probacvmumis}), (\ref{ass:probacvpimis}) hold. Then 

    \begin{equation*}
        F_n(z) \xrightarrow[]{a.s.} \Bar{F}(z; \mu, \zeta_d)
    \end{equation*}
    
\end{lemma}

\begin{proof}
    Rewrite $F_n(z;\hat{\Sigma},\Sigma) = F_n(z)$ to highlight the dependence of the function on the sample covariance, $\hat{\Sigma}$, and the population covariance, $\Sigma$.

    Note that 

    \begin{align*}
        | F_n(z;\hat{\Sigma},\Sigma_1) - F_n(z;\hat{\Sigma},\Sigma_2) | &\leq z \|\beta_{is}\| \|\beta_{oos}\| \| (\hat{\Sigma} + zI)^{-1} \|_{\text{op}} \| \Sigma_1 - \Sigma_2\|_{\text{op}} \\
        & \leq \|\beta_{is}\| \|\beta_{oos}\| \| \Sigma_1 - \Sigma_2\|_{\text{op}}
    \end{align*}

    The second inequality follows after observing that $\| (\hat{\Sigma} + zI)^{-1} \|_{\text{op}} \leq z^{-1}$.

    We also have

    \begin{align*}
        | F_n(z;\hat{\Sigma}_1,\Sigma) - F_n(z;\hat{\Sigma}_2,\Sigma) | &\leq z \|\beta_{is}\| \|\beta_{oos}\| \| \Sigma \|_{\text{op}} \| (\hat{\Sigma}_1 + zI)^{-1} - (\hat{\Sigma}_2 + zI)^{-1} \|_{\text{op}} \\
        & \leq z \tau^{-1} \|\beta_{is}\| \|\beta_{oos}\| \| (\hat{\Sigma}_1 + zI)^{-1}(\hat{\Sigma}_1 + zI - (\hat{\Sigma}_2 + zI))(\hat{\Sigma}_2 + zI)^{-1}\|_{\text{op}} \\
        & \leq z \tau^{-1} \|\beta_{is}\| \|\beta_{oos}\| \| (\hat{\Sigma}_1 - \hat{\Sigma}_2)\|_{\text{op}} \|(\hat{\Sigma}_1 + zI)^{-1}\|_{\text{op}} \|(\hat{\Sigma}_2 + zI)^{-1}\|_{\text{op}} \\
        & \leq z^{-1} \tau^{-1} \|\beta_{is}\| \|\beta_{oos}\| (\hat{\Sigma}_1 - \hat{\Sigma}_2)\|_{\text{op}} \\
        & \leq \tau^{-2} \|\beta_{is}\| \|\beta_{oos}\| (\hat{\Sigma}_1 - \hat{\Sigma}_2)\|_{\text{op}} 
    \end{align*}

    Therefore, setting $C = \max \big\{ \|\beta_{is}\| \|\beta_{oos}\|, \tau^{-2} \|\beta_{is}\| \|\beta_{oos}\| \big\}$, we have

    \begin{align*}
        | F_n(z;\hat{\Sigma}_1,\Sigma_1) - F_n(z;\hat{\Sigma}_2,\Sigma_2) | &= | F_n(z;\hat{\Sigma}_1,\Sigma_1) - F_n(z;\hat{\Sigma}_1,\Sigma_2) + F_n(z;\hat{\Sigma}_1,\Sigma_2) - F_n(z;\hat{\Sigma}_2,\Sigma_2) | \\
        & \leq | F_n(z;\hat{\Sigma}_1,\Sigma_1) - F_n(z;\hat{\Sigma}_1,\Sigma_2)| + |F_n(z;\hat{\Sigma}_1,\Sigma_2) - F_n(z;\hat{\Sigma}_2,\Sigma_2) | \\
        & \leq C \| \Sigma_1 - \Sigma_2\|_{\text{op}} + C \|\hat{\Sigma}_1 - \hat{\Sigma}_2\|_{\text{op}}
    \end{align*}

    Let $\hat{\Sigma}_1 = n^{-1} X_1' X_1$ and $\hat{\Sigma}_2 = n^{-1} X_2' X_2$. Then,

    \begin{align}
    \label{eq:ineqlipschitz}
    \begin{split}
        | F_n(z;\hat{\Sigma}_1,\Sigma_1) - F_n(z;\hat{\Sigma}_2,\Sigma_2) | & \leq C \| \Sigma_1 - \Sigma_2\|_{\text{op}} + \frac{C}{n} \| X_1' X_1 - X_1' X_2 + X_1' X_2 - X_2' X_2\|_{\text{op}} \\
        & \leq C \| \Sigma_1 - \Sigma_2\|_{\text{op}} + \frac{C}{n} \| X_1' (X_1 - X_2) + (X_1 - X_2)' X_2\|_{\text{op}} \\
        & \leq C \| \Sigma_1 - \Sigma_2\|_{\text{op}} + \frac{C}{n} \| X_1 - X_2 \|_{\text{op}} \big( \| X_1 \|_{\text{op}} + \| X_2 \|_{\text{op}} \big)
    \end{split}
    \end{align}

    The rest of the proof relies on the same truncation and centralization arguments used in \cite{hastie2020surprises} (Appendix A.4.).

    Let $M > 0$. We decompose the matrices of covariates, $X$, and latent covariates, $Z$, as follows:

    \begin{align*}
        &X = X_M + \Tilde{X}_M, \quad X_M = Z_M \Sigma_M^{1/2}, \quad \Tilde{X}_M = \Tilde{Z}_M \Sigma^{1/2} \\
        &Z = a_M Z_M + \Tilde{Z}_M \\
        & \Sigma_M = a_M^2 \Sigma
    \end{align*}

    where, denoting by $Z_{ij}$ the $(i,j)$-th entry of $Z$,

    \begin{align*}
        &(Z_M)_{ij} = \frac{1}{a_M} \Big( Z_{ij} \mathds{1}_{|Z_{ij}|\leq M} - \E[Z_{ij} \mathds{1}_{|Z_{ij}|\leq M}] \Big), \quad a_M = \E \Big[ \Big( Z_{ij} \mathds{1}_{|Z_{ij}|\leq M} - \E[Z_{ij} \mathds{1}_{|Z_{ij}|\leq M}] \Big)^2\Big]^{1/2} \\
        &(\Tilde{Z}_M)_{ij} = Z_{ij} \mathds{1}_{|Z_{ij}| > M} - \E[Z_{ij} \mathds{1}_{|Z_{ij}| > M}]
    \end{align*}

    Note that $Z_{ij} = a_M (Z_M)_{ij} + (\Tilde{Z}_M)_{ij}$. Each entry of $Z_M$ and $\Tilde{Z}_M$ has zero mean and the entries of $Z_M$ have unit variance and are bounded. Moreover, the fourth moment of $Z_{ij}$ are finite, which entails

    \begin{align*}
        |a_M^2 - 1| \leq \varepsilon_M, \quad \E \Big[ \big| (\Tilde{Z}_M)_{ij} \big|^4 \Big] \leq \varepsilon_M^4
    \end{align*}

    for some $\varepsilon_M \rightarrow 0$ as $M \rightarrow \infty$.

    Using \eqref{eq:ineqlipschitz}, we obtain 

    \begin{align*}
        | F_n(z;\hat{\Sigma}_1,\Sigma_1) - F_n(z;\hat{\Sigma}_2,\Sigma_2) | &\leq C \| \Sigma_M - \Sigma\|_{\text{op}} + \frac{C}{n} \| \Tilde{X}_M \|_{\text{op}} \big( \| X \|_{\text{op}} + \| X_M \|_{\text{op}} \big) \\
        & \leq C \tau^{-1} |a_M^2 - 1| + C\tau^{-1} \Big|\Big| \frac{\Tilde{Z}_M}{\sqrt{n}} \Big|\Big|_{\text{op}} \Big( \Big|\Big| \frac{Z}{\sqrt{n}} \Big|\Big|_{\text{op}} + \Big|\Big| \frac{Z_M}{\sqrt{n}} \Big|\Big|_{\text{op}} \Big) \\
    \end{align*}

    Recall that the operator norm of a real matrix $Y$ is the square root of the maximum eigenvalue of the (symmetric) matrix $Y'Y$. Therefore, Lemma \ref{lem:baiyin} ensures that 

    \begin{align*}
        \underset{n,p \rightarrow \infty}{\lim \sup} \Big( \Big|\Big| \frac{Z}{\sqrt{n}} \Big|\Big|_{\text{op}} + \Big|\Big| \frac{Z_M}{\sqrt{n}} \Big|\Big|_{\text{op}} \Big) \leq K < \infty, \quad \underset{n,p \rightarrow \infty}{\lim \sup} \Big|\Big| \frac{\Tilde{Z}_M}{\sqrt{n}} \Big|\Big|_{\text{op}} < K \varepsilon_M, \quad a.s.
    \end{align*}

    which yields

    \begin{align}
    \label{eq:limsup1}
        \underset{n,p \rightarrow \infty}{\lim \sup} | F_n(z;\hat{\Sigma},\Sigma) - F_n(z;n^{-1} X_M' X_M,\Sigma_M) | \leq C \tau^{-1} K \varepsilon_M = C_1 \varepsilon_M
    \end{align}

    Next, denote by $\hat{\mu}^M$ and $\hat{\zeta}_d^M$ the ESD and reweighted ESD of $\Sigma_M$ defined analogously to the measures in \eqref{eq:esd} and \eqref{eq:rewesddrift} for $\Sigma$.

    As each entry of $Z_M$ has bounded moments of all orders, we can apply Lemma \ref{lem:deterministic_equivalent} to $F_n(z; n^{-1} X_M' X_M,\Sigma_M)$. Then, taking $\varepsilon = 0.1$ and $D = 2$, we have 

    \begin{align*}
        |F_n(z; n^{-1} X_M' X_M,\Sigma_M) - \Bar{F}(z;\hat{\mu}^M, \hat{\zeta}_d^M)| \leq \frac{\| \Sigma_M \beta_{oos}\| \|\beta_{is}\|}{n^{0.49}z} = \frac{C}{n^{0.49}}
    \end{align*}

    with probability $1 - n^{-2}$ for $n$ large enough.

    Concomitantly, we have the following relations for the spectral measures:
    
    \begin{align*}
        \hat{\mu}^M(x) = \hat{\mu}\Big(\frac{x}{a_M^2}\Big), \quad \hat{\zeta}_d^M(x) = \hat{\zeta}_d \Big(\frac{x}{a_M^2}\Big)
    \end{align*}

    Therefore, $\hat{\mu}^M(x)$ (resp. $\hat{\zeta}^M(x)$) converges weakly to $\mu^M(x) = \mu(x/a_M^2)$ (resp. $\zeta_d^M(x) = \zeta_d (x/a_M^2)$). By Borel-Cantelli lemma, this yields

    \begin{align}
    \label{eq:limsup2}
        \underset{n \rightarrow \infty}{\lim \sup} |F_n(z; n^{-1} X_M' X_M,\Sigma_M) - \Bar{F}(z;\mu^M, \zeta_d^M)| = 0.
    \end{align}

    Finally, note that $\mu^M$ (resp. $\zeta_d^M$) weakly converges to $\mu$ (resp. $\zeta_d$) as $M \rightarrow \infty$. Consequently, taking $M \rightarrow \infty$, Equations $\eqref{eq:limsup1}$ and $\eqref{eq:limsup2}$ yield 

    \begin{equation*}
        F_n(z;\hat{\Sigma},\Sigma) \xrightarrow[]{a.s.} \Bar{F}(z; \mu, \zeta_d)
    \end{equation*}

    The proof of Lemma \ref{lem:asympcv} is complete.
    
\end{proof}

\begin{lemma}
\label{lem:qt}
    Let $\xi = (e_1, e_2, \dots , e_n)'$ with $(e_i)_i$ as in \eqref{eq:dgpmis}. Then,

    \begin{equation*}
        \frac{1}{n^2} \E [ (X' \xi) (X' \xi)' | X ] = \frac{1}{n} \hat{\Sigma}
    \end{equation*}
\end{lemma}

\begin{proof}

    \begin{align*}
        \frac{1}{n^2} \E [ (X' \xi) (X' \xi)' | X ] &= \frac{1}{n^2} X' \E [\xi \xi' | X ] X = \frac{1}{n^2} X' \E [\xi \xi' ] X \\
        & =  \frac{1}{n^2} X' X  = \frac{1}{n} \hat{\Sigma}
    \end{align*}

\end{proof}

\begin{lemma}
\label{lem:nonasyboundres}
    Let $z > 0$, $\xi = (e_1, e_2, \dots , e_n)'$ with $(e_i)_i$ as in \eqref{eq:dgpmis} and Assumptions (\ref{ass:iddistcov}), (\ref{ass:complexity}), (\ref{ass:opbound}) hold. Besides, assume $\E [\lambda_{max}(n^{-1}Z'Z)] \leq K < \infty$, for any $n \in \N_{>0}$. Then, there exist $C := C(\tau,K)$ such that with probability at least $1 - n^{-1/5}$, 

    \begin{equation*}
        \big| \frac{1}{n}\beta_{oos}' \Sigma (zI + \hat{\Sigma})^{-1} X' \xi \big| \leq \frac{C \| \Sigma \beta_{oos} \|}{n^{2/5}}
    \end{equation*}
\end{lemma}

\begin{proof}
    Using Chebyshev's inequality, for some $t>0$, we have

    \begin{align}
    \label{eq:chebyshev}
        \mathds{P} \Big( \big| \frac{1}{n}\beta_{oos}' \Sigma (zI + \hat{\Sigma})^{-1} X' \xi \big| \geq t \Big) \leq \frac{\mathds{V}\big[ \frac{1}{n}\beta_{oos}' \Sigma (zI + \hat{\Sigma})^{-1} X' \xi \big]}{t^2} = \frac{\E\big[ \big| \frac{1}{n}\beta_{oos}' \Sigma (zI + \hat{\Sigma})^{-1} X' \xi \big|^2 \big]}{t^2}
    \end{align}

    The equality follows from the independence of $X$ and $\xi$. Besides, 

    \begin{align*}
        \E \Big[ \big|\frac{1}{n}\beta_{oos}' \Sigma (zI + \hat{\Sigma})^{-1} X' \xi \big|^2 \Big] &= \E \Big[ \big(\frac{1}{n}\beta_{oos}' \Sigma (zI + \hat{\Sigma})^{-1} X' \xi \big) \big( \frac{1}{n}\beta_{oos}' \Sigma (zI + \hat{\Sigma})^{-1} X' \xi \big)' \Big] \\
        & = \frac{1}{n^2} \E \Big[ \beta_{oos}' \Sigma (zI + \hat{\Sigma})^{-1} X' \xi \xi' X (zI + \hat{\Sigma})^{-1} \Sigma \beta_{oos} \Big] \\
        & = \E \Big[ \beta_{oos}' \Sigma (zI + \hat{\Sigma})^{-1} \frac{1}{n^2} \E [ (X' \xi) (X' \xi)' | X ] (zI + \hat{\Sigma})^{-1} \Sigma \beta_{oos} \Big]
    \end{align*}

    The last equality results from the law of iterated expectations. Next, using Lemma \ref{lem:qt}, we have

    \begin{align}
    \label{eq:l2conv}
    \begin{split}
        \E \Big[ \big|\frac{1}{n}\beta_{oos}' \Sigma (zI + \hat{\Sigma})^{-1} X' \xi \big|^2 \Big] &= \frac{1}{n} \E \Big[ \beta_{oos}' \Sigma (zI + \hat{\Sigma})^{-1} \hat{\Sigma} (zI + \hat{\Sigma})^{-1} \Sigma \beta_{oos} \Big] \\
        & = \frac{1}{n} \beta_{oos}' \Sigma \E \Big[ (zI + \hat{\Sigma})^{-1} \Sigma^{1/2} \Big( \frac{1}{n} Z'Z \Big) \Sigma^{1/2} (zI + \hat{\Sigma})^{-1} \Big] \Sigma \beta_{oos} \\
        & \leq \frac{1}{n} \| \Sigma \beta_{oos} \|^2 \|\Sigma^{1/2}\|_{\text{op}}^2 \E \Big[ \|(zI + \hat{\Sigma})^{-1}\|_{\text{op}}^2 \Big| \Big| \frac{1}{n} Z'Z \Big| \Big|_{\text{op}} \Big] \\
        & \leq \frac{1}{n} \| \Sigma \beta_{oos} \|^2 \tau^{-1} z^{-2}\E \Big[ \Big| \Big| \frac{1}{n} Z'Z \Big| \Big|_{\text{op}} \Big] \\
        & = \frac{C \| \Sigma \beta_{oos} \|^2}{n z^2}
    \end{split}
    \end{align}

    Taking $t = C^{1/2} \| \Sigma \beta_{oos} \| n^{-1/4}z$ in \eqref{eq:chebyshev}, we get

    \begin{align*}
        \mathds{P} \Big( \big| \frac{1}{n}\beta_{oos}' \Sigma (zI + \hat{\Sigma})^{-1} X' \xi \big| \geq \frac{C^{1/2} \| \Sigma \beta_{oos} \|}{n^{1/4}z} \Big) \leq \frac{1}{n^{1/2}}
    \end{align*}
    
\end{proof}

\begin{corollary}
\label{cor:cvresasy}
    Let $z > 0$, $\xi = (e_1, e_2, \dots , e_n)'$ with $(e_i)_i$ as in \eqref{eq:dgpmis} and Assumptions (\ref{ass:iddistcov}), (\ref{ass:complexity}), (\ref{ass:opbound}) hold. Then, in the limit, as $p,n \rightarrow \infty$, $p/n \rightarrow c \in (0,\infty)$,

    \begin{equation*}
        \frac{1}{n}\beta_{oos}' \Sigma (zI + \hat{\Sigma})^{-1} X' \xi \xrightarrow[]{L_2} 0
    \end{equation*}
\end{corollary}

\begin{proof}

    Using Bai-Yin theorem (\cite{yin1988limit}, \cite{bai2008limit}), we have $\underset{n,p \rightarrow \infty}{\lim \sup} \| n^{-1}Z'Z \|_{\text{op}} \leq K < \infty$. Then, the result follows from taking $n \rightarrow \infty$ in \eqref{eq:l2conv}.
    
\end{proof}

\begin{proposition}
\label{prop:nonasyboundsdrift}
    Let $z > \tau$, $\xi = (e_1, e_2, \dots , e_n)'$ with $(e_i)_i$ as in \eqref{eq:dgpmis} and Assumptions (\ref{ass:iddistcov}), (\ref{ass:complexity}), (\ref{ass:opbound}), (\ref{ass:accumulationzero}) hold.  Then, there exists $C := C(\tau)$ such that 

    \begin{equation}
    \label{eq:nonasyexpret}
        \big| \mathbb{E}\left[ r^{(s)}_{t+1}(z) | X  \right] - (\Bar{H}(\hat{\zeta}_d) + \Bar{F}(z; \hat{\mu}, \hat{\zeta}_d)) \big| \leq \frac{C \| \beta_{is} \| \| \beta_{oos}\|}{n^{2/5}z}
    \end{equation}
    
    holds with probability at least $1-n^{-1/5}$ for $n$ sufficiently large.  
\end{proposition}

\begin{proof}
    We wish to study the out-of-sample expected return of the market timing strategy. Then, we integrate over the randomness in a new (independent) observation $x_{t} \in \R^p$ and in the errors $e_{t} \in \R$. Therefore,

    \begin{align}
    \label{eq:oosexpret}
    \begin{split}
        \mathbb{E}\left[ r^{(s)}_{t+1}(z)  |X \right] & = \E [ \hat{\pi}_t r_{t+1} | X] \\
        & = \E [ \hat{\beta}_{is}' x_t (x_t' \beta_{oos} + e_{t})|X] \\
        & = \E [ \hat{\beta}_{is}' x_t x_t' \beta_{oos} |X] \\
        & = \E [ \beta_{oos}' x_t x_t' \hat{\beta}_{is} |X]
    \end{split}
    \end{align}

    Let $\xi = (e_{t-n-1}, e_{t - n}, \dots , e_{t-1})'$. Then, 

    \begin{align}
    \label{eq:betahatdecomp}
    \begin{split}
        \hat{\beta}_{is} & = \frac{1}{n} (zI + \hat{\Sigma})^{-1} X' (X \beta_{is} + \xi) \\
        & = (zI + \hat{\Sigma})^{-1} (\hat{\Sigma} \beta_{is} + \frac{1}{n} X' \xi)
    \end{split}
    \end{align}

    Injecting this into \eqref{eq:oosexpret}, we get

    \begin{align*}
        \mathbb{E}\left[ r^{(s)}_{t+1}(z) | X  \right] & = \E [ \beta_{oos}' x_t x_t' (zI + \hat{\Sigma})^{-1} (\hat{\Sigma} \beta_{is} + \frac{1}{n} X' \xi) |X] \\
        & = \E [ \beta_{oos}' x_t x_t' (zI + \hat{\Sigma})^{-1} \hat{\Sigma} \beta_{is}|X ] + \E \left[ \frac{1}{n} \beta_{oos}' x_t x_t' (zI + \hat{\Sigma})^{-1}  X' \xi|X \right] \\
        & = \beta_{oos}' \E[x_t x_t'] (zI + \hat{\Sigma})^{-1} \hat{\Sigma} \beta_{is} + \frac{1}{n} \beta_{oos}' \E[x_t x_t'] (zI + \hat{\Sigma})^{-1}  X' \xi \\
        & = \beta_{oos}' \Sigma (zI + \hat{\Sigma})^{-1} \hat{\Sigma} \beta_{is} + \frac{1}{n} \beta_{oos}' \Sigma (zI + \hat{\Sigma})^{-1}  X' \xi \\
        & = \beta_{oos}' \Sigma (zI + \hat{\Sigma})^{-1} (\hat{\Sigma} + zI - zI) \beta_{is} + \frac{1}{n} \beta_{oos}' \Sigma (zI + \hat{\Sigma})^{-1}  X' \xi \\
        & = \beta_{oos}' \Sigma \beta_{is} - z \beta_{oos}' \Sigma (zI + \hat{\Sigma})^{-1} \beta_{is} + \frac{1}{n} \beta_{oos}' \Sigma (zI + \hat{\Sigma})^{-1}  X' \xi.
    \end{align*}
    
    The inequality \eqref{eq:nonasyexpret} follows from Lemma \ref{lem:nonasyboundres} and by taking $\varepsilon = 10^{-1}$ and $D = 1/5$ in Lemma \ref{lem:deterministic_equivalent}. 
\end{proof}

We define the limiting out-of-sample expected return of the strategy under concept drift if the model is well-specified

\begin{align}
\label{eq:expretdrift}
    \begin{split}
        \mathcal{E}^{(d)}(z;c,\mu, \zeta_d)
    & = \langle \beta_{oos}, \beta_{is} \rangle_{Q(z)} \\
    &= \|\beta_{is}\| \|\beta_{oos}\| \int \Big( \lambda - \frac{z\lambda}{\lambda(cz m(-z;c,\mu) + 1 - c) + z} \Big) d \zeta_d (\lambda) 
    \end{split}
\end{align}

with $Q(z) \equiv Q(z;c,\Sigma) := \Sigma (I - z( zI + (czm(-z;c,\mu) + 1 - c) \Sigma)^{-1})$ and $m := m(-z;c,\mu)$ is given in Definition \ref{def:mn}.

\begin{proposition}
\label{prop:cvprobdrift}
Let $z > \tau$ and Assumptions (\ref{ass:iddistcov}), (\ref{ass:complexity}), (\ref{ass:opbound}), (\ref{ass:accumulationzero}), (\ref{ass:probacvmumis}), (\ref{ass:probacvpimis}) hold. Then,
    \begin{equation}
        \mathbb{E}\left[ r^{(s)}_{t+1}(z) | X  \right]  \xrightarrow[\substack{\mathstrut n,p \rightarrow \infty \\ p/n \rightarrow c }]{\mathds{P}} \mathcal{E}^{(d)}(z;c,\mu, \zeta_d)  \label{eq:E_prop7}
    \end{equation}

\end{proposition}

\begin{proof}

    By Lemma \ref{lem:asympcv}, we have 
    
    \begin{equation*}
        - z \beta_{oos}' \Sigma (zI + \hat{\Sigma})^{-1} \beta_{is} \xrightarrow[\substack{\mathstrut n,p \rightarrow \infty \\ p/n \rightarrow c }]{a.s.} -z \beta_{oos}' \Sigma ( zI + (czm(-z;c,\mu) + 1 - c) \Sigma)^{-1} \beta_{is} ,
    \end{equation*}

    and by Corollary \ref{cor:cvresasy},

    \begin{equation*}
        \frac{1}{n} \beta_{oos}' \Sigma (zI + \hat{\Sigma})^{-1}  X' \xi \xrightarrow[\substack{\mathstrut n,p \rightarrow \infty \\ p/n \rightarrow c }]{\mathds{P}} 0.
    \end{equation*}

    Consequently,

    \begin{align*}
        \mathbb{E}\left[ r^{(s)}_{t+1}(z) | X  \right] & \xrightarrow[\substack{\mathstrut n,p \rightarrow \infty \\ p/n \rightarrow c }]{\mathds{P}} \beta_{oos}' \Sigma \beta_{is} - z \beta_{oos}' \Sigma ( zI + (czm(-z;c,\mu) + 1 - c) \Sigma)^{-1} \beta_{is} 
    \end{align*}

    The proof of Proposition \ref{prop:cvprobdrift} is complete.
\end{proof}

We define the limiting out-of-sample expected return of the strategy when the model is well-specified and no concept drift is considered: 

\begin{align}
\label{eq:expretwodrift}
    \begin{split}
        \mathcal{E}(z;c,\mu, \zeta_{is}) &= \| \beta_{is} \|_{Q(z)}^2 \\
    &= \|\beta_{is}\|^2 \int \Big( \lambda - \frac{z\lambda}{\lambda(cz m(-z;c,\hat{\mu}) + 1 - c) + z} \Big) d \zeta_{is} (\lambda) 
    \end{split}
\end{align}

with $Q(z) \equiv Q(z;c,\Sigma) := \Sigma (I - z( zI + (czm(-z;c,\hat{\mu}) + 1 - c) \Sigma)^{-1})$ and $m := m(-z;c,\hat{\mu})$ is given in Definition \ref{def:mn}.

\begin{corollary}
    \label{cor:cvprobwodrift}
    Let $z > \tau$ and Assumptions (\ref{ass:iddistcov}), (\ref{ass:complexity}), (\ref{ass:opbound}), (\ref{ass:accumulationzero}), (\ref{ass:probacvmumis}), (\ref{ass:probacvpimis}) hold. Then,
    
        \begin{equation*}
            \mathbb{E}\left[ r^{(s)}_{t+1}(z) | X, \beta_{is} = \beta_{oos} \right]  \xrightarrow[\substack{\mathstrut n,p \rightarrow \infty \\ p/n \rightarrow c }]{\mathds{P}} \mathcal{E}(z;c,\mu, \zeta_{is}) .
        \end{equation*}
\end{corollary}

\begin{proof}
    The proof is straightforward from Proposition \ref{prop:cvprobdrift}.
\end{proof}

\begin{lemma}
\label{lem:polarization}
    Let $z > \tau$ and Assumptions (\ref{ass:iddistcov}), (\ref{ass:complexity}), (\ref{ass:opbound}), (\ref{ass:accumulationzero}), (\ref{ass:probacvmumis}), (\ref{ass:probacvpimis}) hold. 

    \begin{align*}
        \mathcal{E}^{(d)}(z;c,\mu,\zeta_d)
        & = \frac{1}{2} \big( \| \beta_{oos} \|_{Q(z)}^2 + \| \beta_{is} \|_{Q(z)}^2 - \| \beta_{is} - \beta_{oos} \|_{Q(z)}^2 \big)
    \end{align*}

    where $Q(z) \equiv Q(z;c,\Sigma) := \Sigma (I - z( zI + (czm(-z;c,\hat{\mu}) + 1 - c) \Sigma)^{-1})$, or alternatively

    \begin{align*}
        \mathcal{E}^{(d)}(z;c,\mu,\zeta_d) &= \frac{ \|\beta_{is}\|^2}{2} \int \Big( \lambda - \frac{z\lambda}{\lambda(cz m(-z;c,\hat{\mu}) + 1 - c) + z} \Big) d \zeta_{is}(\lambda) \\
        &\quad + \frac{ \|\beta_{oos}\|^2}{2} \int \Big( \lambda - \frac{z\lambda}{\lambda(cz m(-z;c,\hat{\mu}) + 1 - c) + z} \Big) d \zeta_{oos}(\lambda) \\
        & \quad - \frac{ \|\beta_{is} - \beta_{oos}\|^2}{2} \int \Big( \lambda - \frac{z\lambda}{\lambda(cz m(-z;c,\hat{\mu}) + 1 - c) + z} \Big) d \zeta_{r}(\lambda).
    \end{align*}
\end{lemma}

\begin{proof}
    This result follows from the polarization identity $\langle (u-v),A(u-v) \rangle = \langle u, Au \rangle + \langle v, Av \rangle - 2 \langle u, Av \rangle$ for $u,v \in \R^p$ and $0 \preccurlyeq A \in \R^{p \times p}$.
\end{proof}

The following proposition establishes the conditions under which the performance of the strategy deteriorates in the presence of concept drift.

\begin{proposition}
\label{prop:condwell}
    Let $z > \tau$ and Assumptions (\ref{ass:iddistcov}), (\ref{ass:complexity}), (\ref{ass:opbound}), (\ref{ass:accumulationzero}), (\ref{ass:probacvmumis}), (\ref{ass:probacvpimis}) hold. Consider $\mathcal{E}^{(d)}(z;c,\mu, \zeta_d)$ and $\mathcal{E}(z;c,\mu, \zeta_{is})$ as defined in \eqref{eq:expretdrift} and \eqref{eq:expretwodrift} respectively. Then,

    \begin{equation*}
        \mathcal{E}^{(d)}(z;c,\mu, \zeta_d) < \mathcal{E}(z;c,\mu, \zeta_{is})
    \end{equation*}

    if and only if, 

    \begin{equation}
    \label{eq:condprod}
        \langle (\beta_{oos} - \beta_{is}), \beta_{is} \rangle_{Q(z)} < 0
    \end{equation}

    or alternatively, if and only if,

    \begin{equation}
    \label{eq:condnorm}
        \| \beta_{oos} \|_{Q(z)}^2 - \| \beta_{is} \|_{Q(z)}^2 < \| \beta_{is} - \beta_{oos} \|_{Q(z)}^2
    \end{equation}
\end{proposition}

\begin{proof}
    Condition \eqref{eq:condprod} directly follows from the definition of $\mathcal{E}^{(d)}(z;c,\mu, \zeta_d)$ and $\mathcal{E}(z;c,\mu, \zeta_{is})$, i.e.,

    \begin{align*}
        & \quad \mathcal{E}^{(d)}(z;c,\mu, \zeta_d) < \mathcal{E}(z;c,\mu, \zeta_{is}) \\
        & \Leftrightarrow \langle \beta_{oos}, \beta_{is} \rangle_{Q(z)} < \| \beta_{is} \|_{Q(z)}^2 \\
        & \Leftrightarrow \langle \beta_{oos}, Q(z) \beta_{is} \rangle < \langle \beta_{is}, Q(z) \beta_{is} \rangle \\
        & \Leftrightarrow \langle (\beta_{oos} - \beta_{is}), Q(z) \beta_{is} \rangle < 0 \\
        & \Leftrightarrow \langle (\beta_{oos} - \beta_{is}), \beta_{is} \rangle_{Q(z)} < 0
    \end{align*}

    Next, we use Lemma \ref{lem:polarization} to derive Condition \eqref{eq:condnorm}. Indeed, note that we can rewrite \\ $\mathcal{E}^{(d)}(z;c,\mu, \zeta_d)$ as 

    \begin{align*}
        \mathcal{E}^{(d)}(z;c,\mu, \zeta_d) = \mathcal{E}(z;c,\mu, \zeta_{is}) + \frac{1}{2} \big( \| \beta_{oos} \|_{Q(z)}^2 - \| \beta_{is} \|_{Q(z)}^2 - \| \beta_{is} - \beta_{oos} \|_{Q(z)}^2 \big).
    \end{align*}

    Then,

    \begin{align*}
        &\quad \mathcal{E}^{(d)}(z;c,\mu, \zeta_d) < \mathcal{E}(z;c,\mu, \zeta_{is}) \\
        &\Leftrightarrow \frac{1}{2} \big( \| \beta_{oos} \|_{Q(z)}^2 - \| \beta_{is} \|_{Q(z)}^2 - \| \beta_{is} - \beta_{oos} \|_{Q(z)}^2 \big) < 0 \\
        &\Leftrightarrow \| \beta_{oos} \|_{Q(z)}^2 - \| \beta_{is} \|_{Q(z)}^2 < \| \beta_{is} - \beta_{oos} \|_{Q(z)}^2 
    \end{align*}

    The proof of Proposition \ref{prop:condwell} is complete.
\end{proof}

The next proposition establishes the limiting out-of-sample expected return of the strategy under concept drift in the standard isotropic scenario, i.e. the covariates are i.i.d. with $\Sigma = I$.

\begin{proposition}
\label{prop:isotropic}
    Let $z>\tau$. Assume $x_i = (x_{i1}, x_{i2}, \dots, x_{ip})$ has independent and identically distributed entries and, in addition, that for all $j \in \{1,\dots,p\}$, $\E[z_{ij}] = 0$, $\E[z_{ij}^2] = 1$ and $\E[|z_{ij}|^{4+a}] \leq C < \infty$ for $a > 0$. Then, 
    \begin{equation}
        \mathbb{E}\left[ r^{(s)}_{t+1}(z) \big| X  \right]  \xrightarrow[\substack{\mathstrut n,p \rightarrow \infty \\ p/n \rightarrow c }]{\mathds{P}} f(z;c) \langle \beta_{is}, \beta_{oos} \rangle  \label{eq:E_prop9}
    \end{equation}

    where 

    \begin{align}
    \label{eq:f_isotropic}
    \begin{split}
        f(z;c) &= 1 - \frac{z}{z(1+cm(-z;c)) + 1 - c} \in [0,1), \\
        \quad \text{ with } \quad m(-z;c) &= \frac{c-1-z + \sqrt{(1-c+z)^2 + 4cz}}{2cz}.
    \end{split}
    \end{align}

    Furthermore, $z \mapsto f(z;c) \langle \beta_{is}, \beta_{oos} \rangle$ is monotone decreasing (resp. increasing) in $z$ when $\langle \beta_{is}, \beta_{oos} \rangle \geq 0$ (resp. $\langle \beta_{is}, \beta_{oos} \rangle \leq 0$).
    
\end{proposition}

\begin{proof}
    Let $\Sigma = I$. Then,

    \begin{align*}
        \mathcal{E}^{(d)}(z;c,\mu, \zeta_d) = \beta_{oos}' Q(I) \beta_{is} \\
    \end{align*}
    
    with 

    \begin{align*}
        Q(I) &= I (I - z( zI + (czm(-z;c,\hat{\mu}) + 1 - c) I)^{-1}) \\
        & = \Big( 1 - \frac{z}{z + czm(-z;c,\hat{\mu}) +1-c} \Big) I
    \end{align*}
    
    We are left with finding $m(-z;c,\hat{\mu})$ with $\Im(-z) > 0$. Recall that $m := m(-z;c,\hat{\mu})$ is the unique solution in $\C_+$ to
    
    \begin{align*}
        m & = \int \frac{1}{\lambda(1-c+czm)+z} d\hat{\mu} ( \lambda) \\
        & = \frac{1}{p} \sum_{i=1}^p \frac{1}{1-c+czm+z} \\
        & = \frac{1}{1-c+czm+z}
    \end{align*}

    This is a quadratic equation in $m$ whose roots are:

    \begin{align*}
        m &= \frac{c-1-z \pm \sqrt{(1-c+z)^2 + 4cz}}{2cz}. \\
    \end{align*}

    By the definition of the Stieltjes transform, $m(-z)$ must behave like $1/z$ as $|z| \rightarrow \infty$. This is only possible when the positive branch cut is chosen. This yields:
    
    \begin{equation*}
        m(-z;c) = \frac{c-1-z + \sqrt{(1-c+z)^2 + 4cz}}{2cz} .
    \end{equation*}

    Note that the quantity $czm(-z;c) + 1 - c$ is always positive. Indeed,

    \begin{align*}
        &czm(-z;c) + 1 - c > 0 \\
        \Leftrightarrow \quad  & c-1-z + \sqrt{(1-c+z)^2 + 4cz} > 2(c-1) \\
        \Leftrightarrow \quad & (1-c+z)^2 - (1-c-z)^2 + 4cz > 0 \\
        \Leftrightarrow \quad & 4z(1-c) + 4cz > 0 \\
        \Leftrightarrow \quad & z > 0
    \end{align*}

    which is always verified as $z > \tau >0$. Therefore, 

    \begin{align*}
        f(z;c) &= 1 - \frac{z}{z(1+cm(-z;c)) + 1 - c} \in [0,1).
    \end{align*}

    Regarding the monotonicity of the limiting expected return in $z$, notice that

    \begin{align*}
        & \quad m(z;c) = \frac{1}{1-c-czm(z;c)-z} \\
        & \Rightarrow m'(z;c) = \frac{cm(z;c)+czm'(z;c)+1}{(1-c-czm(z;c)-z)^2} \\
        & \Rightarrow m'(z;c) = \frac{1 + cm(z;c)}{[1-c-czm(z;c)-z]^2-cz}.
    \end{align*}

    Hence, 

    \begin{equation*}
        m'(-z;c) = \frac{1 + cm(-z;c)}{[1-c+czm(-z;c)+z]^2+cz} \geq 0.
    \end{equation*}

    Finally, we have 

    \begin{align}
    \label{eq:derivativez}
        \frac{\partial}{\partial z} \Big[ \Big( 1 - \frac{z}{z(1+cm(-z;c)) + 1 - c} \Big) \langle \beta_{is}, \beta_{oos} \rangle \Big] = - \frac{cm'(-z;c)z^2 + 1 + c}{[z(1+cm(-z;c))+1-c]^2} \langle \beta_{is}, \beta_{oos} \rangle.
    \end{align}

    The derivative in \eqref{eq:derivativez} is negative (resp. positive) when $\langle \beta_{is}, \beta_{oos} \rangle \geq 0$ (resp. $\langle \beta_{is}, \beta_{oos} \rangle \leq 0$), and the proof of Proposition \ref{prop:isotropic} is complete. 
    
\end{proof}

\subsection{Well-specified model and ridgeless regression}
\label{subsub:wellspecifiedminnorm}

In this section, we extend the results presented in the previous section to the ridgeless case, i.e. taking $z \to 0$. Here again, we consider that the true model is known. Hence, $\theta = 0$, $q=0$, $\phi = 1$ and $\Sigma \equiv \Sigma_x$.

Let 

\begin{equation*}
    G_n = \lim_{z \to 0} -z \langle \beta_{oos}, \Sigma (\hat{\Sigma}+zI)^{-1}\beta_{is} \rangle ,
\end{equation*}

and define:

\begin{align*}
    \bar{G}(\hat{\mu}, \hat{\zeta}_d) &= - \langle \beta_{oos}, \Sigma ( I + c s_0 \Sigma)^{-1} \beta_{is} \rangle \\
    &= - \|\beta_{is}\| \|\beta_{oos}\| \int \frac{\lambda}{1+\lambda c s_0} d \hat{\zeta}_d ( \lambda) ,\\
\end{align*}

where $s_0 := s_0(c,\hat{\mu})$ is given in Definition \ref{def:s0}.

Moreover, consider the following hypothesis.

\begin{assump}
\label{ass:interpolationtreshold}
    $|c - 1| \geq \tau$.
\end{assump}

The complexity of the model is supposed to be bounded away from the interpolation threshold, $c = 1$. 

\begin{assump}
\label{ass:mineigen}
    $\lambda_p > \tau$.
\end{assump}

This assumption imposes positive-definiteness to $\Sigma$ by requiring the minimum eigenvalue of the covariance matrix to be strictly positive. 

We now state a number of lemmas that will be useful in the remainder of the proofs. In particular, Lemma $\ref{lem:deterministic_equivalent_ridgeless}$ establishes non-asymptotic bounds for the gap between $G_n$ and $\bar{G}$.

\begin{lemma}
\label{lem:hastieanalysissn}
    Let Assumptions (\ref{ass:complexity}), (\ref{ass:opbound}), (\ref{ass:accumulationzero}), (\ref{ass:iddistcov}), (\ref{ass:interpolationtreshold}) and (\ref{ass:mineigen}) hold. Define $s_n(-z;c,\hat{\mu}) = m(-z;c,\hat{\mu}) - (c-1)/cz$ where $m(-z;c,\hat{\mu})$ is given in Definition \ref{def:mn} and $s_0(c,\hat{\mu}) = \lim_{z \to 0} s_n(-z;c,\hat{\mu})$ which is given in Definition \ref{def:s0}. Then,

    \begin{equation*}
        m(-z;c,\hat{\mu}) = \frac{c-1}{cz} + s_0(c,\hat{\mu}) + O_*(z),
    \end{equation*}

    where $O_*(z)$ denotes a quantity uniformly bounded as $|O_*(z)| \leq C(\tau)z$.
\end{lemma}

\begin{proof}
    The proof is given in Appendix B.1. in \cite{hastie2020surprises}.
\end{proof}

\begin{lemma}
\label{lem:deterministic_equivalent_ridgeless}
    Let Assumptions (\ref{ass:complexity}), (\ref{ass:opbound}), (\ref{ass:accumulationzero}), (\ref{ass:iddistcov}), (\ref{ass:interpolationtreshold}) and (\ref{ass:mineigen}) hold. Then, for any $D>0$ and $\varepsilon >0$, there exist $C := C(\tau)$ and  $n_0 := n_0(\varepsilon,D)$ such that for $n \geq n_0$, 

    \begin{equation*}
        |G_n - \bar{G}(\hat{\mu}, \hat{\zeta}_d)| \leq \frac{C\| \beta_{oos}\| \|\beta_{is}\|}{n^{\frac{1}{4}-\varepsilon}},
    \end{equation*}

holds with probability $1-n^{-D}$. 
\end{lemma}

\begin{proof}

    We proceed in three steps. First, we show that $\big| G_n - F_n(z)  \big| \leq C(\tau) \|\beta_{oos}\| \|\beta_{is}\|z$. Then, we show that $\big| \bar{G}(\hat{\mu}, \hat{\zeta}_d) - \Bar{F}(z;\hat{\mu}, \hat{\zeta}_d)  \big| \leq C(\tau) \|\beta_{oos}\| \|\beta_{is}\|z$. Finally, using Lemma \ref{lem:deterministic_equivalent} and combining both results, we prove the desired bound. \\

    \textbf{1) Bound on $\big| G_n - F_n(z)  \big|$}.

    \begin{align}
    \label{eq:gaprandom}
    \begin{split}
        \big| G_n - F_n(z)  \big| &=  \big| \lim_{z \to 0} \big( -z  \beta_{oos}' \Sigma (\hat{\Sigma}+zI)^{-1}\beta_{is} \big) + z\beta_{oos}' \Sigma ( zI + \hat{\Sigma})^{-1} \beta_{is}  \big| \\
        & = \big| \lim_{z \to 0} \big( -z  \beta_{oos}' \Sigma (\hat{\Sigma}+zI)^{-1}\beta_{is} \big) + \beta_{oos}' \Sigma \beta_{is} \\
        & \quad - \beta_{oos}' \Sigma \beta_{is} + z\beta_{oos}' \Sigma ( zI + \hat{\Sigma})^{-1} \beta_{is}  \big| \\
        & = \big| \lim_{z \to 0} \big(  \beta_{oos}' \Sigma (\hat{\Sigma}+zI)^{-1}\hat{\Sigma}\beta_{is} \big) - \beta_{oos}' \Sigma ( zI + \hat{\Sigma})^{-1} \hat{\Sigma}\beta_{is}  \big| \\
        & = \big|  \beta_{oos}' \Sigma \hat{\Sigma}^+\hat{\Sigma}\beta_{is} - \beta_{oos}' \Sigma ( zI + \hat{\Sigma})^{-1} \hat{\Sigma}\beta_{is}  \big|.
    \end{split}
    \end{align}

    Denote by $r_1 \geq r_2 \geq \dots \geq r_{p} \geq 0$ the eigenvalues of $\hat{\Sigma}$. Let $\hat{\Sigma} = PDP^{-1}$ be the eigendecomposition of $\hat{\Sigma}$. Then, $P \in \R^{p \times p}$ is orthogonal and $D \in \R^{p \times p}$ is a diagonal matrix with the eigenvalues of $\hat{\Sigma}$ on the main diagonal. By construction, $\hat{\Sigma}^+ = PEP^{-1}$ with $E$ a diagonal matrix such that

    \begin{align*}
        E_{ii} = 
        \begin{cases}
            r_i^{-1}, & \quad r_i > 0 \\
            0, & \quad r_i = 0.
        \end{cases} 
    \end{align*}

    Therefore, we have $\hat{\Sigma}^+ \hat{\Sigma} = PEDP^{-1} = PQP^{-1}$ with $Q$ diagonal such that, for any $i \in \{ 1, \dots, p \}$, $Q_{ii} = 1$ if $r_i > 0$ and $0$ otherwise. Coming back to \eqref{eq:gaprandom}, we have

    \begin{align*}
        \big| G_n - F_n(z)  \big| &= \big|  \beta_{oos}' \Sigma \hat{\Sigma}^+\hat{\Sigma}\beta_{is} - \beta_{oos}' \Sigma ( zI + \hat{\Sigma})^{-1} \hat{\Sigma}\beta_{is}  \big| \\
        & = \big|  \beta_{oos}' \Sigma P \big( Q - (zI + D)^{-1} D\big) P^{-1} \beta_{is} \big| .
    \end{align*}

    Let $T = Q - (zI + D)^{-1} D$. Then, $T \in  \R^{p \times p}$ is a diagonal matrix such that 

    \begin{align*}
        T_{ii} = 
        \begin{cases}
            \frac{z}{r_i + z}, & \quad r_i > 0 \\
            0, & \quad r_i = 0.
        \end{cases}
    \end{align*}

    Hence, introducing $r_{min}$ and $\lambda_{min}(n^{-1}Z'Z)$ the smallest non-vanishing eigenvalues of $\hat{\Sigma}$ and $n^{-1}Z'Z$ respectively, we get the following bound

    \begin{align*}
        \big| G_n - F_n(z)  \big| &= \big|  \beta_{oos}' \Sigma P T P^{-1} \beta_{is} \big| \\
        & \leq \|\beta_{oos}\| \| \Sigma \|_{\text{op}} \| P \|_{\text{op}}^2 \| T \|_{\text{op}} \|\beta_{is}\| \\
        & \leq \|\beta_{oos}\|  \|\beta_{is}\| \tau^{-1} \frac{z}{r_{min}} \\
        & \leq \|\beta_{oos}\|  \|\beta_{is}\| \tau^{-1} \frac{z}{\| \Sigma^{1/2} \|_{\text{op}}^2 \lambda_{min}(n^{-1}Z'Z)} \\
        & \leq \|\beta_{oos}\|  \|\beta_{is}\| \tau^{-2} \frac{z}{ \lambda_{min}(n^{-1}Z'Z)}.
    \end{align*}

    By Assumption (\ref{ass:interpolationtreshold}), the Bai-Yin theorem (\cite{bai2008limit}) yields that \\ 
    $\lambda_{min}(n^{-1}Z'Z) \geq \kappa(\tau) > 0$ with probability at least $1 - n^{-D}$ for any $D>0$ and $n$ sufficiently large. Finally, with the same probability, we have 

    \begin{align*}
        \big| G_n - F_n(z)  \big| &\leq \|\beta_{oos}\|  \|\beta_{is}\| \tau^{-2} \frac{z}{\kappa(\tau)} \\
        &\leq C(\tau) \|\beta_{oos}\|  \|\beta_{is}\| z.
    \end{align*}

    \textbf{2) Bound on $\big| \bar{G}(\hat{\mu}, \hat{\zeta}_d) - \Bar{F}(z;\hat{\mu}, \hat{\zeta}_d)  \big|$}. 

    Recall the following fundamental relation presented in Equation \eqref{eq:relationrm}

    \begin{equation*}
        r_n(-z) = c m(-z;c,\hat{\mu}) + \frac{1-c}{z},
    \end{equation*}

    or equivalently,

    \begin{equation}
    \label{eq:mnfunctionsn}
        m(-z;c,\hat{\mu}) = s_n(-z;c,\hat{\mu}) + \frac{c-1}{cz}.
    \end{equation}

    with $s_n(-z;c,\hat{\mu}) = r_n(-z)/c$. Let $s_0(c,\hat{\mu}) = \lim_{z \to 0} s_n(-z;c,\hat{\mu})$.

    According to Lemma \ref{lem:hastieanalysissn}, we have 

    \begin{equation*}
        m(-z;c,\hat{\mu}) = \frac{c-1}{cz} + s_0(c,\hat{\mu}) + O_*(z)
    \end{equation*}

    where $O_*(z)$ denotes a quantity uniformly bounded as $|O_*(z)| \leq C(\tau)z$. 

    Injecting this into $\Bar{F}(z;\hat{\mu}, \hat{\zeta}_d)$, we get

    \begin{align*}
        \frac{\Bar{F}(z;\hat{\mu}, \hat{\zeta}_d)}{\|\beta_{is}\| \|\beta_{oos}\|} &= - z \int \frac{\lambda}{\lambda(cz m(-z;c,\hat{\mu}) + 1 - c) + z} d \hat{\zeta}_d ( \lambda) \\
         &= - z \int \frac{\lambda}{\lambda(cz [(c-1)/cz + s_0(c,\hat{\mu}) + O_*(z)] + 1 - c) + z} d \hat{\zeta}_d ( \lambda) \\
         &= - z \int \frac{\lambda}{\lambda(c-1 + czs_0(c,\hat{\mu}) + czO_*(z)] + 1 - c) + z} d \hat{\zeta}_d ( \lambda) \\
         & = - z \int \frac{\lambda}{\lambda (czs_0(c,\hat{\mu}) + czO_*(z)) + z} d \hat{\zeta}_d ( \lambda) \\
         & = - \int \frac{\lambda}{\lambda (cs_0(c,\hat{\mu}) + cO_*(z)) + 1} d \hat{\zeta}_d ( \lambda) .
    \end{align*}

    As per Assumption (\ref{ass:complexity}), $cO_*(z)$ is also uniformly bounded as $|cO_*(z)| \leq C(\tau)z$. Hence, $\big| \bar{G}(\hat{\mu}, \hat{\zeta}_d) - \Bar{F}(z;\hat{\mu}, \hat{\zeta}_d)  \big| \leq C(\tau)\|\beta_{oos}\|  \|\beta_{is}\|z$. \vspace{3mm}

    \textbf{3) Conclusion}.

    \begin{align*}
        |G_n - \bar{G}(\hat{\mu}, \hat{\zeta}_d)| &= |G_n - F_n(z) + F_n(z) - \bar{G}(\hat{\mu}, \hat{\zeta}_d)| \\
        & \leq |G_n - F_n(z)| + |F_n(z) - \bar{G}(\hat{\mu}, \hat{\zeta}_d)| \\
        & \leq |G_n - F_n(z)| + |F_n(z) - \Bar{F}(z;\hat{\mu}, \hat{\zeta}_d) + \Bar{F}(z;\hat{\mu}, \hat{\zeta}_d) - \bar{G}(\hat{\mu}, \hat{\zeta}_d)| \\
        & \leq |G_n - F_n(z)| + |F_n(z) - \Bar{F}(z;\hat{\mu}, \hat{\zeta}_d)| + |\Bar{F}(z;\hat{\mu}, \hat{\zeta}_d) - \bar{G}(\hat{\mu}, \hat{\zeta}_d)| \\
        & \leq C(\tau) \|\beta_{oos}\|  \|\beta_{is}\| \Big(z + \frac{1}{n^{1/2-\varepsilon}z} \Big).
    \end{align*}

    The last inequality follows from results 1) and 2) above and Lemma \ref{lem:deterministic_equivalent}. 

    Taking $z = n^{-1/4} \geq n^{-2/3+\tau}$, we get

    \begin{equation*}
        |G_n - \bar{G}(\hat{\mu}, \hat{\zeta}_d)| \leq \frac{C(\tau) \|\beta_{oos}\|  \|\beta_{is}\|}{n^{1/4-\varepsilon}}
    \end{equation*}

    with probability $1-n^{-D}$.

    The proof of Lemma \ref{lem:deterministic_equivalent_ridgeless} is complete.

\end{proof}

\begin{corollary}
\label{cor:asycvridgeless}
    Let Assumptions (\ref{ass:iddistcov}), (\ref{ass:complexity}), (\ref{ass:opbound}), (\ref{ass:accumulationzero}), (\ref{ass:probacvmumis}), (\ref{ass:probacvpimis}), (\ref{ass:interpolationtreshold}) and (\ref{ass:mineigen}) hold. Then,

    \begin{equation}
    \label{eq:ineqG}
        G_n \xrightarrow[]{a.s.} \bar{G}(\mu, \zeta_d)
    \end{equation}
\end{corollary}

\begin{proof}
    Rewrite $G_n(\hat{\Sigma},\Sigma) = G_n$ to highlight the dependence of the function on the sample covariance, $\hat{\Sigma}$, and the population covariance, $\Sigma$.

    Note that 

    \begin{align*}
        |-z \langle \beta_{oos}, \Sigma_1 (\hat{\Sigma}+zI)^{-1}\beta_{is} \rangle + z \langle \beta_{oos}, \Sigma_2 (\hat{\Sigma}+zI)^{-1}\beta_{is} \rangle| &\leq z \|\beta_{is}\| \|\beta_{oos}\| \| (\hat{\Sigma} + zI)^{-1} \|_{\text{op}} \| \Sigma_1 - \Sigma_2\|_{\text{op}} \\
        & \leq \|\beta_{is}\| \|\beta_{oos}\| \| \Sigma_1 - \Sigma_2\|_{\text{op}}
    \end{align*}

    The second inequality follows after observing that $\| (\hat{\Sigma} + zI)^{-1} \|_{\text{op}} \leq z^{-1}$. This implies 

    \begin{align*}
        | G_n(\hat{\Sigma},\Sigma_1) - G_n(\hat{\Sigma},\Sigma_2) | \leq \|\beta_{is}\| \|\beta_{oos}\| \| \Sigma_1 - \Sigma_2\|_{\text{op}}.
    \end{align*}

    We also have

    \begin{align*}
        & \quad |-z \langle \beta_{oos}, \Sigma (\hat{\Sigma}_1+zI)^{-1}\beta_{is} \rangle + z \langle \beta_{oos}, \Sigma (\hat{\Sigma}_2+zI)^{-1}\beta_{is} \rangle| \\ 
        & \leq z \|\beta_{is}\| \|\beta_{oos}\| \| \Sigma \|_{\text{op}} \| (\hat{\Sigma}_1 + zI)^{-1} - (\hat{\Sigma}_2 + zI)^{-1} \|_{\text{op}} \\
        & \leq z \tau^{-1} \|\beta_{is}\| \|\beta_{oos}\| \| (\hat{\Sigma}_1 + zI)^{-1}(\hat{\Sigma}_1 + zI - (\hat{\Sigma}_2 + zI))(\hat{\Sigma}_2 + zI)^{-1}\|_{\text{op}} \\
        & \leq z \tau^{-1} \|\beta_{is}\| \|\beta_{oos}\| \| (\hat{\Sigma}_1 + zI)^{-1}(\hat{\Sigma}_1 + zI - (\hat{\Sigma}_2 + zI))(\hat{\Sigma}_2 + zI)^{-1}\|_{\text{op}} \\
        & \leq z \tau^{-1} \|\beta_{is}\| \|\beta_{oos}\| \| (\hat{\Sigma}_1 - \hat{\Sigma}_2)\|_{\text{op}} \|(\hat{\Sigma}_1 + zI)^{-1}\|_{\text{op}} z^{-1} \\
        & \leq \tau^{-1} \|\beta_{is}\| \|\beta_{oos}\| \| (\hat{\Sigma}_1 - \hat{\Sigma}_2)\|_{\text{op}} \|(\hat{\Sigma}_1 + zI)^{-1}\|_{\text{op}}. 
    \end{align*} 

    Let $\hat{\Sigma}_1 = n^{-1}X_1'X_1 = n^{-1} \Sigma^{1/2}Z_1'Z_1 \Sigma^{1/2}$. Then,

    \begin{align*}
        & \quad |-z \langle \beta_{oos}, \Sigma (\hat{\Sigma}_1+zI)^{-1}\beta_{is} \rangle + z \langle \beta_{oos}, \Sigma (\hat{\Sigma}_2+zI)^{-1}\beta_{is} \rangle| \\ 
        & \leq \tau^{-1} \|\beta_{is}\| \|\beta_{oos}\| \| (\hat{\Sigma}_1 - \hat{\Sigma}_2)\|_{\text{op}} \frac{1}{\| \Sigma^{1/2} \|_{\text{op}}^2 \lambda_{min}(n^{-1}Z_1'Z_1)} \\
        & \leq \|\beta_{is}\| \|\beta_{oos}\| \| (\hat{\Sigma}_1 - \hat{\Sigma}_2)\|_{\text{op}} \frac{1}{\lambda_{min}(n^{-1}Z_1'Z_1)} .
    \end{align*} 

    By Assumption (\ref{ass:interpolationtreshold}), the Bai-Yin theorem (\cite{bai2008limit}) yields that \\ 
    $\lambda_{min}(n^{-1}Z_1'Z_1) \geq \kappa(\tau) > 0$, almost surely. Hence, 

    \begin{align*}
        |-z \langle \beta_{oos}, \Sigma (\hat{\Sigma}_1+zI)^{-1}\beta_{is} \rangle + z \langle \beta_{oos}, \Sigma (\hat{\Sigma}_2+zI)^{-1}\beta_{is} \rangle| \leq \kappa(\tau)^{-1} \|\beta_{is}\| \|\beta_{oos}\| \| (\hat{\Sigma}_1 - \hat{\Sigma}_2)\|_{\text{op}},
    \end{align*}

    which implies 

    \begin{equation*}
        | G_n(\hat{\Sigma}_1,\Sigma) - G_n(\hat{\Sigma}_2,\Sigma) | \leq \kappa(\tau)^{-1} \|\beta_{is}\| \|\beta_{oos}\| \| (\hat{\Sigma}_1 - \hat{\Sigma}_2)\|_{\text{op}}.
    \end{equation*}

    Taking $C = \max \big\{ \|\beta_{is}\| \|\beta_{oos}\|, \kappa(\tau)^{-1} \|\beta_{is}\| \|\beta_{oos}\| \big\}$, the rest of the proof is similar to the one of Lemma \ref{lem:asympcv}.
\end{proof}

\begin{lemma}
\label{lem:nonasyboundresridgeless}
    Let $\xi = (e_1, e_2, \dots , e_n)'$ with $(e_i)_i$ as in \eqref{eq:dgpmis} and Assumptions (\ref{ass:iddistcov}), (\ref{ass:complexity}), (\ref{ass:opbound}), (\ref{ass:interpolationtreshold}) and (\ref{ass:mineigen}) hold. Besides, assume $\lambda_{min}(n^{-1}Z'Z) \geq \kappa(\tau) > 0$, for any $n \in \N_{>0}$. Then, for $\varepsilon>0$, there exist $C := C(\tau)$ such that with probability at least $1 - n^{-1/2}$, 

    \begin{equation*}
        \big| \frac{1}{n}\beta_{oos}' \Sigma \hat{\Sigma}^{+} X' \xi \big| \leq \frac{C \| \beta_{oos} \|}{n^{1/4-\varepsilon}}.
    \end{equation*}
\end{lemma}

\begin{proof}
    Using Chebyshev's inequality, for some $t>0$, we have

    \begin{align}
    \label{eq:chebyshevridgeless}
        \mathds{P} \Big( \big| \frac{1}{n}\beta_{oos}' \Sigma \hat{\Sigma}^{+} X' \xi \big| \geq t \Big) \leq \frac{\mathds{V}\big[ \frac{1}{n}\beta_{oos}' \Sigma \hat{\Sigma}^{+} X' \xi \big]}{t^2} = \frac{\E\big[ \big| \frac{1}{n}\beta_{oos}' \Sigma \hat{\Sigma}^{+} X' \xi \big|^2 \big]}{t^2}.
    \end{align}

    The equality follows from the independence of $X$ and $\xi$. Besides, 

    \begin{align*}
        \E \Big[ \big|\frac{1}{n}\beta_{oos}' \Sigma \hat{\Sigma}^{+} X' \xi \big|^2 \Big] &= \E \Big[ \big(\frac{1}{n}\beta_{oos}' \Sigma \hat{\Sigma}^{+} X' \xi \big) \big( \frac{1}{n}\beta_{oos}' \Sigma \hat{\Sigma}^{+} X' \xi \big)' \Big] \\
        & = \frac{1}{n^2} \E \Big[ \beta_{oos}' \Sigma \hat{\Sigma}^{+} X' \xi \xi' X \hat{\Sigma}^{+} \Sigma \beta_{oos} \Big] \\
        & = \E \Big[ \beta_{oos}' \Sigma \hat{\Sigma}^{+} \frac{1}{n^2} \E [ (X' \xi) (X' \xi)' | X ] \hat{\Sigma}^{+} \Sigma \beta_{oos} \Big].
    \end{align*}

    The last equality results from the law of iterated expectations. Next, using Lemma \ref{lem:qt}, we have

    \begin{align}
    \label{eq:l2convridgeless}
    \begin{split}
        \E \Big[ \big|\frac{1}{n}\beta_{oos}' \Sigma \hat{\Sigma}^{+} X' \xi \big|^2 \Big] &= \frac{1}{n} \E \Big[ \beta_{oos}' \Sigma \hat{\Sigma}^{+} \hat{\Sigma} \hat{\Sigma}^{+} \Sigma \beta_{oos} \Big] \\
        & = \frac{1}{n} \E \Big[ \beta_{oos}' \Sigma \hat{\Sigma}^{+} \Sigma \beta_{oos} \Big] \\
        & \leq \frac{1}{n} \| \beta_{oos} \|^2 \| \Sigma \|_{\text{op}}^2 \frac{1}{\| \Sigma^{1/2} \|_{\text{op}}^2 \lambda_{min}(n^{-1}Z'Z)} \\
        & \leq \frac{1}{n} \| \beta_{oos} \|^2 \tau^{-3} \frac{1}{ \lambda_{min}(n^{-1}Z'Z)} \\
        &\leq \frac{1}{n} C(\tau) \| \beta_{oos} \|^2 .
    \end{split}
    \end{align}

    Let $C := C(\tau)$. Taking $t = C^{1/2} \| \beta_{oos} \| n^{-1/4}$ in \eqref{eq:chebyshevridgeless}, we get

    \begin{align*}
        \mathds{P} \Big( \big| \frac{1}{n}\beta_{oos}' \Sigma \hat{\Sigma}^{+} X' \xi \big| \geq \frac{C^{1/2} \| \beta_{oos} \|}{n^{1/4}} \Big) \leq \frac{1}{n^{1/2}}.
    \end{align*}
    
\end{proof}

\begin{corollary}
\label{cor:asyboundresridgeless}
    Let $\xi = (e_1, e_2, \dots , e_n)'$ with $(e_i)_i$ as in \eqref{eq:dgpmis} and Assumptions (\ref{ass:iddistcov}), (\ref{ass:complexity}), (\ref{ass:opbound}), (\ref{ass:interpolationtreshold}) and (\ref{ass:mineigen}) hold. Then,

    \begin{equation*}
        \frac{1}{n}\beta_{oos}' \Sigma \hat{\Sigma}^+ X' \xi \xrightarrow[]{L_2} 0.
    \end{equation*}
\end{corollary}

\begin{proof}
    Using Bai-Yin theorem (\cite{bai2008limit}) and Assumption (\ref{ass:interpolationtreshold}), we have \\ $\underset{n,p \rightarrow \infty}{\lim \sup} \lambda_{min}(n^{-1}Z'Z) \quad \geq \kappa(\tau) > 0$. Then, the result follows from taking $n \rightarrow \infty$ in \eqref{eq:l2convridgeless}.
\end{proof}

\begin{proposition}
\label{prop:nonasyboundsdriftridgeless}
    Let Assumptions (\ref{ass:iddistcov}), (\ref{ass:complexity}), (\ref{ass:opbound}), (\ref{ass:accumulationzero}), (\ref{ass:interpolationtreshold}) and (\ref{ass:mineigen}) hold.  Besides, assume $\lambda_{min}(n^{-1}Z'Z) \geq \kappa(\tau) > \infty$, for any $n \in \N_{>0}$. Then, for $\varepsilon>0$ and $n$ sufficiently large, there exists $C := C(\tau)$ such that

    \begin{equation}
    \label{eq:ridgelessnonasy}
        \big| \lim_{z \to 0} \mathbb{E}\left[ r^{(s)}_{t+1}(z) | X  \right] - \big(\Bar{H}(\hat{\zeta}_d) + \bar{G}(\hat{\mu}, \hat{\zeta}_d) \mathds{1}_{\{c>1\}} \big) \big| \leq \frac{C \| \beta_{is} \| \| \beta_{oos}\|}{n^{1/4 - \varepsilon}}
    \end{equation}
    
    holds with probability at least $1-n^{-1/2}$.

\end{proposition}

\begin{proof}
    Recall that

    \begin{align*}
        \mathbb{E}\left[ r^{(s)}_{t+1}(z) | X  \right] & = \beta_{oos}' \Sigma (zI + \hat{\Sigma})^{-1} \hat{\Sigma} \beta_{is} + \frac{1}{n} \beta_{oos}' \Sigma (zI + \hat{\Sigma})^{-1}  X' \xi .
    \end{align*}

    If $c < 1$, $\lim_{z \to 0} \beta_{oos}' \Sigma (zI + \hat{\Sigma})^{-1} \hat{\Sigma} \beta_{is} = \beta_{oos}' \Sigma \beta_{is}$. Hence, from Lemma \ref{lem:nonasyboundresridgeless}, we have 

    \begin{align*}
        \big| \lim_{z \to 0} \mathbb{E}\left[ r^{(s)}_{t+1}(z) | X  \right] - \Bar{H}(\hat{\zeta}_d) \big| \leq \frac{C \| \beta_{is} \| \| \beta_{oos}\|}{n^{1/4 - \varepsilon}}.
    \end{align*}

    Next, note that

    \begin{align*}
        \mathbb{E}\left[ r^{(s)}_{t+1}(z) | X  \right] & = \beta_{oos}' \Sigma (zI + \hat{\Sigma})^{-1} \hat{\Sigma} \beta_{is} + \frac{1}{n} \beta_{oos}' \Sigma (zI + \hat{\Sigma})^{-1}  X' \xi \\
        & = \beta_{oos}' \Sigma (zI + \hat{\Sigma})^{-1} (\hat{\Sigma} + zI - zI) \beta_{is} + \frac{1}{n} \beta_{oos}' \Sigma (zI + \hat{\Sigma})^{-1}  X' \xi \\
        & = \beta_{oos}' \Sigma \beta_{is} - z \beta_{oos}' \Sigma (zI + \hat{\Sigma})^{-1} \beta_{is} + \frac{1}{n} \beta_{oos}' \Sigma (zI + \hat{\Sigma})^{-1}  X' \xi.
    \end{align*}
    
    Thus, taking $D = 1/2$ in Lemma \ref{lem:deterministic_equivalent_ridgeless}, we get

    \begin{align*}
        \big| \lim_{z \to 0} \mathbb{E}\left[ r^{(s)}_{t+1}(z) | X  \right] - \big(\Bar{H}(\hat{\zeta}_d) + \bar{G}(\hat{\mu}, \hat{\zeta}_d) \big) \big| \leq \frac{C \| \beta_{is} \| \| \beta_{oos}\|}{n^{1/4 - \varepsilon}}.
    \end{align*}
\end{proof}

We define the limiting out-of-sample expected return of the strategy under concept drift in the ridgeless case if the model is well-specified

\begin{align}
\label{eq:expretdriftridgeless}
    \begin{split}
        \mathcal{E}^{(d)}(c,\mu, \zeta_d)
    & = \langle \beta_{oos}, \beta_{is} \rangle_{R} \\
    &= \|\beta_{is}\| \|\beta_{oos}\| \int \Big( \lambda - \frac{\lambda}{1+\lambda c s_0(c,\hat{\mu})}\mathds{1}_{\{c>1\}} \Big) d \zeta_d (\lambda) ,
    \end{split}
\end{align}

with 

\begin{align*}
    R \equiv R(c,\Sigma) := \begin{cases}
        \Sigma, \quad & c<1 \\
        \Sigma (I - (c s_0(c,\hat{\mu}) \Sigma + I)^{-1}), \quad & c>1
    \end{cases}
\end{align*}

and $s_0 := s_0(c,\hat{\mu})$ is given in Definition \ref{def:s0}.

\begin{proposition}
\label{prop:cvprobdriftridgeless}
Let Assumptions (\ref{ass:iddistcov}), (\ref{ass:complexity}), (\ref{ass:opbound}), (\ref{ass:accumulationzero}), (\ref{ass:interpolationtreshold}), (\ref{ass:mineigen}), (\ref{ass:probacvmumis}) and (\ref{ass:probacvpimis}) hold. Then,
    \begin{equation}
        \mathbb{E}\left[ r^{(s)}_{t+1}(z) | X  \right]  \xrightarrow[\substack{\mathstrut n,p \rightarrow \infty \\ p/n \rightarrow c \\ z \rightarrow 0}]{\mathds{P}} \mathcal{E}^{(d)}(c,\mu, \zeta_d) .
        \label{eq:E_prop11}
    \end{equation}

\end{proposition}

\begin{proof}

    Recall that

    \begin{align*}
        \mathbb{E}\left[ r^{(s)}_{t+1}(z) | X  \right] & = \beta_{oos}' \Sigma (zI + \hat{\Sigma})^{-1} \hat{\Sigma} \beta_{is} + \frac{1}{n} \beta_{oos}' \Sigma (zI + \hat{\Sigma})^{-1}  X' \xi .
    \end{align*}

    If $c < 1$, $\lim_{z \to 0} \beta_{oos}' \Sigma (zI + \hat{\Sigma})^{-1} \hat{\Sigma} \beta_{is} = \beta_{oos}' \Sigma \beta_{is}$. Hence, the result in the underparameterized regime follows from Corollary \ref{cor:asyboundresridgeless}.

    Now, consider $c>1$. By Corollary \ref{cor:asycvridgeless}, we have 
    
    \begin{equation*}
        - z \beta_{oos}' \Sigma (zI + \hat{\Sigma})^{-1} \beta_{is} \xrightarrow[\substack{\mathstrut n,p \rightarrow \infty \\ p/n \rightarrow c \\ z \to 0 }]{a.s.} - \beta_{oos} \Sigma (c s_0(c,\mu) \Sigma + I)^{-1} \beta_{is},
    \end{equation*}

    and by Corollary \ref{cor:asyboundresridgeless},

    \begin{equation*}
        \frac{1}{n} \beta_{oos}' \Sigma \hat{\Sigma}^+  X' \xi \xrightarrow[\substack{\mathstrut n,p \rightarrow \infty \\ p/n \rightarrow c }]{\mathds{P}} 0.
    \end{equation*}

    Consequently, for $c>1$,

    \begin{align*}
        \mathbb{E}\left[ r^{(s)}_{t+1}(z) | X  \right] & \xrightarrow[\substack{\mathstrut n,p \rightarrow \infty \\ p/n \rightarrow c \\ z \to 0}]{\mathds{P}} \beta_{oos}' \Sigma \beta_{is} - \beta_{oos} \Sigma (c s_0(c,\mu) \Sigma + I)^{-1}  \beta_{is}. 
    \end{align*}

    The proof of Proposition \ref{prop:cvprobdriftridgeless} is complete.
\end{proof}

The next proposition establishes the limiting out-of-sample expected return of the strategy under concept drift in the ridgeless scenario when the covariates are i.i.d. with $\Sigma = I$ (isotropic features).

\begin{proposition}
\label{prop:isotropicridgeless}
    Assume $x_i = (x_{i1}, x_{i2}, \dots, x_{ip})$ has independent and identically distributed entries and, in addition, that for all $j \in \{1,\dots,p\}$, $\E[z_{ij}] = 0$, $\E[z_{ij}^2] = 1$ and $\E[|z_{ij}|^{4+a}] \leq C < \infty$ for $a > 0$. Then, 
    \begin{equation}
        \mathbb{E}\left[ r^{(s)}_{t+1}(z) | X  \right]  \xrightarrow[\substack{\mathstrut n,p \rightarrow \infty \\ p/n \rightarrow c \\ z \to 0 }]{\mathds{P}} f(c) \langle \beta_{is}, \beta_{oos} \rangle 
        \label{eq:E_prop12}
    \end{equation}

    with 

    \begin{equation}
    \label{eq:isotropicridgeless}
        f(c) = 1 - \frac{c-1}{c}\mathds{1}_{\{c>1\}}  \in (0,1].
    \end{equation}
    
\end{proposition}

\begin{proof}
    From \eqref{eq:expretdriftridgeless}, if $c<1$, we have $\mathcal{E}^{(d)}(c,\mu, \zeta_d) = \langle \beta_{is}, \beta_{oos} \rangle$. This yields the result in the underparameterized regime. 

    Now, consider $c>1$. Again, from \eqref{eq:expretdriftridgeless},

    \begin{align}
    \label{eq:isotropiccsup1}
        \mathcal{E}^{(d)}(c,\mu, \zeta_d) = \Big( 1 - \frac{1}{1 + cs_0} \Big) \langle \beta_{is}, \beta_{oos} \rangle,
    \end{align}

    and from Definition \ref{def:s0}, we have 

    \begin{align*}
        1 - \frac{1}{c} = \frac{1}{1+cs_0},
    \end{align*}

    which implies $s_0 = [c(c-1)]^{-1}$. Injecting this into \eqref{eq:isotropiccsup1}, we get

    \begin{equation*}
        \mathcal{E}^{(d)}(c,\mu, \zeta_d) = \big( 1 - \frac{c-1}{c} \big) \langle \beta_{is}, \beta_{oos} \rangle .
    \end{equation*}

    The proof of Proposition \ref{prop:isotropicridgeless} is complete.
\end{proof}

\subsection{Misspecified model (General dependence for features)}
\label{subsub:misspecifiedgeneral}

We now assume the data generating process (DGP) described in \eqref{eq:dgpmis}. 

\vskip 5mm

\noindent \textbf{Proof of Proposition \ref{prop:expretmis}}: 

Let $x_t \in \R^p$ (resp. $w_t \in \R^q$) be a new independent sample of observable (resp. unobservable) features. Then,

\begin{align}
\label{eq:expretmis}
\begin{split}
    \mathbb{E}\left[ r^{(s,m)}_{t+1}(z) \big| X  \right] & = \mathbb{E}\left[ \hat{\pi}_t r_{t+1} \big| X  \right] \\
    & = \mathbb{E}\left[ (\beta_{oos}'x_t + \theta_{oos}'w_t)  x_t' \hat{\beta}_{is} \big| X  \right] .
\end{split}
\end{align}

Note that in the misspecified setting, we are also integrating over the randomness in the unobserved features $W$. 

Similarly to \eqref{eq:betahatdecomp} and by taking into account the presence of unobserved features, we can decompose $\hat{\beta}_{is}$ as follows

\begin{align*}
    \hat{\beta}_{is} & = \frac{1}{n} (zI + \hat{\Sigma}_x)^{-1} X' (X \beta_{is} + W \theta_{is}  + \xi) \\
    & = (zI + \hat{\Sigma}_x)^{-1} (\hat{\Sigma}_x \beta_{is} + \hat{\Sigma}_{xw} \theta_{is} + \frac{1}{n} X' \xi).
\end{align*}

Therefore, coming back to \eqref{eq:expretmis}, 

\begin{align*}
    \mathbb{E}\left[ r^{(s,m)}_{t+1}(z) \big| X  \right] &= \mathbb{E}\left[ (\beta_{oos}'x_t + \theta_{oos}'w_t)  x_t' (zI + \hat{\Sigma}_x)^{-1} (\hat{\Sigma}_x \beta_{is} + \hat{\Sigma}_{xw} \theta_{is} + \frac{1}{n} X' \xi) \big| X  \right] \\
    & = \beta_{oos}' \E \big[x_t x_t' \big] (zI + \hat{\Sigma}_x)^{-1} \hat{\Sigma}_x \beta_{is} \\
    & \quad + \beta_{oos}' \E \big[x_t x_t' (zI + \hat{\Sigma}_x)^{-1} \hat{\Sigma}_{xw} \big| X \big] \theta_{is} \\ & \quad + \frac{1}{n} \beta_{oos}' \E \big[x_t x_t' \big] (zI + \hat{\Sigma}_x)^{-1} X' \xi \\
    & \quad + \theta_{oos}' \E \big[w_t x_t' \big] (zI + \hat{\Sigma}_x)^{-1} \hat{\Sigma}_x \beta_{is} \\
    & \quad + \theta_{oos}' \E \big[w_t x_t' (zI + \hat{\Sigma}_x)^{-1} \hat{\Sigma}_{xw} \big| X \big] \theta_{is} \\ & \quad + \frac{1}{n} \theta_{oos}' \E \big[w_t x_t' \big] (zI + \hat{\Sigma}_x)^{-1} X' \xi.
\end{align*}

By the law of iterated expectations, we have

\begin{align*}
    \beta_{oos}' \E \big[x_t x_t' (zI + \hat{\Sigma}_x)^{-1} \hat{\Sigma}_{xw} \big| X \big] \theta_{is} &= \beta_{oos}' \E \Big[ \E \big[ x_t x_t' (zI + \hat{\Sigma}_x)^{-1} \hat{\Sigma}_{xw} \big| X,W \big] \Big| X \Big] \theta_{is} \\
    & = \beta_{oos}' \E \Big[ \E \big[ x_t x_t'\big| X,W \big] (zI + \hat{\Sigma}_x)^{-1} \hat{\Sigma}_{xw}  \Big| X \Big] \theta_{is} \\
    & = \beta_{oos}' \E \Big[ \Sigma_x (zI + \hat{\Sigma}_x)^{-1} \hat{\Sigma}_{xw}  \Big| X \Big] \theta_{is} \\
    & = \beta_{oos}' \Sigma_x (zI + \hat{\Sigma}_x)^{-1} \E \Big[\hat{\Sigma}_{xw}\Big| X \Big] \theta_{is}, \\
\end{align*}

and 

\begin{align*}
    \theta_{oos}' \E \big[w_t x_t' (zI + \hat{\Sigma}_x)^{-1} \hat{\Sigma}_{xw} \big| X \big] \theta_{is} &= \theta_{oos}' \E \Big[ \E \big[ w_t x_t' (zI + \hat{\Sigma}_x)^{-1} \hat{\Sigma}_{xw} \big| X,W \big] \Big| X \Big] \theta_{is} \\
    & = \theta_{oos}' \E \Big[ \E \big[ w_t x_t'\big| X,W \big] (zI + \hat{\Sigma}_x)^{-1} \hat{\Sigma}_{xw}  \Big| X \Big] \theta_{is} \\
    & = \theta_{oos}' \E \Big[ \Sigma_{xw}' (zI + \hat{\Sigma}_x)^{-1} \hat{\Sigma}_{xw}  \Big| X \Big] \theta_{is} \\
    & = \theta_{oos}' \Sigma_{xw}' (zI + \hat{\Sigma}_x)^{-1} \E \Big[\hat{\Sigma}_{xw}\Big| X \Big] \theta_{is} . 
\end{align*}

Therefore, 

\begin{align*}
    \mathbb{E}\left[ r^{(s,m)}_{t+1}(z) \big| X  \right] & = \beta_{oos}' \Sigma_x (zI + \hat{\Sigma}_x)^{-1} \hat{\Sigma}_x \beta_{is} \\
    & \quad + \beta_{oos}' \Sigma_x (zI + \hat{\Sigma}_x)^{-1} \E \Big[\hat{\Sigma}_{xw}\Big| X \Big] \theta_{is} \\ & \quad + \frac{1}{n} \beta_{oos}' \Sigma_x (zI + \hat{\Sigma}_x)^{-1} X' \xi \\
    & \quad + \theta_{oos}' \Sigma_{xw}' (zI + \hat{\Sigma}_x)^{-1} \hat{\Sigma}_x \beta_{is} \\
    & \quad + \theta_{oos}' \Sigma_{xw}' (zI + \hat{\Sigma}_x)^{-1} \E \Big[\hat{\Sigma}_{xw}\Big| X \Big] \theta_{is} \\ & \quad + \frac{1}{n} \theta_{oos}' \Sigma_{xw}' (zI + \hat{\Sigma}_x)^{-1} X' \xi.
\end{align*}

    The proof of Proposition \ref{prop:expretmis} is complete. \qed

\vspace{10pt}

We introduce the following VESDs

\begin{align}
\label{eq:VESD2}
\begin{split}
    \hat{\eta} &\equiv \hat{\eta} (\lambda) := \frac{1}{\|\beta_{oos}\| \|P'\Sigma_w^{1/2} \theta_{is}\|} \sum_{i=1}^p \langle \beta_{oos} , v_i \rangle \langle v_i, P'\Sigma_w^{1/2} \theta_{is} \rangle  \delta_{\lambda_i}, \\
    \hat{\nu} &\equiv \hat{\nu} (\lambda) := \frac{1}{\|P'\Sigma_w^{1/2}\theta_{is}\| \|P'\Sigma_w^{1/2}\theta_{oos}\|} \sum_{i=1}^p \langle P'\Sigma_w^{1/2}\theta_{is} , v_i \rangle \langle v_i, P'\Sigma_w^{1/2}\theta_{oos} \rangle  \delta_{\lambda_i}, \\
    \hat{\psi} &\equiv \hat{\psi} (\lambda) := \frac{1}{\|\beta_{is}\| \|P'\Sigma_w^{1/2}\theta_{oos}\|} \sum_{i=1}^p \langle \beta_{is} , v_i \rangle \langle v_i, P'\Sigma_w^{1/2}\theta_{oos} \rangle  \delta_{\lambda_i}.
\end{split}
\end{align}

We assume that these three measures weakly converge. This is summed up in the following assumption. 

\begin{assump}
    \label{ass:weakconvsignedmeasures}
    $\hat{\eta}$, $\hat{\nu}$ and $\hat{\psi}$ weakly converge to some (possibly signed) measures, denoted by $\eta$, $\nu$ and $\psi$, respectively. 
\end{assump}

The following proposition characterizes the asymptotic behavior of the out-of-sample expected return of the timing strategy under posterior drift. 

\begin{proposition}
\label{prop:expretgeneral}
    Let $z>0$ and Assumptions (\ref{ass:iddistcov})-(\ref{ass:complexity}) and (\ref{ass:opbound})-(\ref{ass:weakconvsignedmeasures}) hold. Then, the average return of the strategy in the well-specified case ($\theta=0$) is given by: 
    \begin{align}
        \mathbb{E}\left[ r^{(s,w)}_{t+1}(z) \big| X  \right] &\xrightarrow[\substack{\mathstrut n,p,q \rightarrow \infty \\ p/n \rightarrow c\phi }]{\mathds{P}} \quad \|\beta_{\text{is}}\| \|\beta_{\text{oos}}\| \int h(\lambda;z,c\phi,\mu) d \zeta_{d}(\lambda) ,
        \label{eq:E_prop13}
    \end{align}
   and the additive term arising from mispecification ($\theta \neq 0$) has the following limit: 
    \begin{align*}
        \mathcal{M}(z) &\xrightarrow[\substack{\mathstrut n,p,q \rightarrow \infty \\ p/n \rightarrow c\phi }]{\mathds{P}} \quad \mathcal{J}_1(z;c\phi,\mu,\eta) + \mathcal{J}_2(z;c\phi,\mu,\nu) +\mathcal{J}_3(z;c\phi,\mu,\psi), 
    \end{align*}
    
    where $h(\cdot)$ is defined in \eqref{eq:h} and the quantities in the limiting expression of $\mathcal{M}(z)$ are given in Eq. \eqref{eq:Js}. 
\end{proposition}

\begin{proof}
    
By the assumptions on the covariates, we have 

\begin{align*}
    \E \Big[\hat{\Sigma}_{xw}\Big| X \Big] &= \frac{1}{n} X' \E \big[W\big| X \big] \\
    &= \frac{1}{n} X'ZP'\Sigma_w^{1/2} \\ 
    &= \hat{\Sigma}_x \Sigma_x^{-1/2}P'\Sigma_w^{1/2}
\end{align*}

Further, note that $\Sigma_{xw} = \E_{z}[x w'] = \E_z[\Sigma_x^{1/2}zz'P'\Sigma_w^{1/2}] = \Sigma_x^{1/2}P'\Sigma_w^{1/2}$. Thus, we can rewrite the formula for the expected return of the strategy as 

\begin{align*}
    \mathbb{E}\left[ r^{(s,m)}_{t+1}(z) \big| X  \right] & = \beta_{oos}' \Sigma_x \beta_{is} - z\beta_{oos}' \Sigma_x (zI + \hat{\Sigma}_x)^{-1} \beta_{is} \\
    & \quad + \beta_{oos}' \Sigma_x^{1/2}P'\Sigma_w^{1/2} \theta_{is} - z\beta_{oos}' \Sigma_x (zI + \hat{\Sigma}_x)^{-1} \Sigma_x^{-1/2}P'\Sigma_w^{1/2} \theta_{is} \\ 
    & \quad + \theta_{oos}' \Sigma_w^{1/2} P \Sigma_x^{1/2} \beta_{is} - z\theta_{oos}' \Sigma_w^{1/2} P \Sigma_x^{1/2} (zI + \hat{\Sigma}_x)^{-1} \beta_{is} \\
    & \quad + \theta_{oos}' \Sigma_w^{1/2} PP' \Sigma_w^{1/2} \theta_{is} - z\theta_{oos}' \Sigma_w^{1/2} P \Sigma_x^{1/2} (zI + \hat{\Sigma}_x)^{-1} \Sigma_x^{-1/2}P'\Sigma_w^{1/2} \theta_{is} \\ 
    & \quad + \frac{1}{n} (\beta_{oos}' \Sigma_x + \theta_{oos}' \Sigma_{xw}') (zI + \hat{\Sigma}_x)^{-1} X' \xi .
\end{align*}

where we also used the decomposition trick $\hat{\Sigma}_x = \hat{\Sigma}_x + zI - zI$.

From Corollary \ref{cor:cvresasy} and \ref{cor:asyboundresridgeless}, we have both in the ridge and ridgeless case

\begin{equation*}
        \frac{1}{n} (\beta_{oos}' \Sigma_x + \theta_{oos}' \Sigma_{xw}') (zI + \hat{\Sigma})^{-1} X' \xi \xrightarrow[]{L_2} 0.
\end{equation*} 

By Lemma \ref{lem:asympcv}, for any vectors $u,v \in \R^p$ such that $\|u\|_2 \to c_u < \infty$ and $\|v\|_2 \to c_v < \infty$, we have 
    
    \begin{equation*}
        - z u' (zI + \hat{\Sigma})^{-1} v \xrightarrow[\substack{\mathstrut n,p,q \rightarrow \infty \\ p/n \rightarrow c\phi }]{a.s.} -z u' ( zI + (c\phi zm(-z;c\phi,\mu) + 1 - c\phi) \Sigma)^{-1} v ,
    \end{equation*}

In the ridgeless scenario, by Corollary \ref{cor:asycvridgeless}, we have 
    
    \begin{equation*}
        - z u' (zI + \hat{\Sigma})^{-1} v \xrightarrow[\substack{\mathstrut n,p,q \rightarrow \infty \\ p/n \rightarrow c\phi \\ z \to 0 }]{a.s.} - u' (c\phi s_0(c\phi,\mu) \Sigma + I)^{-1} v \mathds{1}_{c\phi>1},
    \end{equation*}

    Hence, applying the last three limiting formulas to the different terms in $\mathbb{E}\left[ r^{(s,m)}_{t+1}(z) \big| X  \right]$ yields 
    
    \begin{align*}
        \mathbb{E}\left[ r^{(s,w)}_{t+1}(z) \big| X  \right] &\xrightarrow[\substack{\mathstrut n,p,q \rightarrow \infty \\ p/n \rightarrow c\phi }]{\mathds{P}} \quad \|\beta_{is}\| \|\beta_{oos}\| \int h(\lambda;z,c\phi,\mu) \lambda d \zeta_{d}(\lambda) \\
        \mathcal{M}(z) &\xrightarrow[\substack{\mathstrut n,p,q \rightarrow \infty \\ p/n \rightarrow c\phi }]{\mathds{P}} \quad \mathcal{J}_1(z;c\phi,\mu,\eta) + \mathcal{J}_2(z;c\phi,\mu,\nu) +\mathcal{J}_3(z;c\phi,\mu,\psi), 
    \end{align*}
    where the limiting quantities are defined as

    \begin{align}
    \label{eq:Js}
    \begin{split}
        \mathcal{J}_1(z;c\phi,\mu,\eta) &:= \|\beta_{oos}\| \|P\Sigma_w^{1/2} \theta_{is}\| \int h(\lambda;z,c\phi,\mu) \sqrt{\lambda} d \eta(\lambda), \\
        \mathcal{J}_2(z;c\phi,\mu,\nu) &:= \|P\Sigma_w^{1/2} \theta_{is}\| \|P\Sigma_w^{1/2} \theta_{oos}\| \int h(\lambda;z,c\phi,\mu) d \nu(\lambda), \\
        \mathcal{J}_3(z;c\phi,\mu,\psi) &:= \|\beta_{is}\| \|P\Sigma_w^{1/2} \theta_{oos}\| \int h(\lambda;z,c\phi,\mu) \sqrt{\lambda} d \psi(\lambda).
    \end{split}
    \end{align}

    with 

    \begin{align}
    \label{eq:h}
    \begin{split}
        h(x;z,c\phi,\mu) := \left\{ \begin{array}{c l}
    1 - \frac{z}{x(c\phi z m(-z;c \phi,\mu) + 1 - c\phi) + z} & \quad z>0 \\[10pt]
     1 - \frac{1}{1+x c\phi s_0(c\phi,\mu)}\mathds{1}_{\{c\phi>1\}}  &\quad z \to 0.
    \end{array} \right. 
    \end{split}
    \end{align}

\end{proof}

\vspace{10pt}

Recall that $\omega_u := \Sigma_x^{1/2} \beta_u + P' \Sigma_w^{1/2} \theta_u$ with $u \in \{is, oos\}$. Note that by Assumptions \ref{ass:probacvpimis} and \ref{ass:weakconvsignedmeasures}, we have that 

\begin{align*}
    \hat{\iota}_{is} &\equiv \hat{\iota}_{is} (\lambda) := \frac{1}{\|\omega_{is}\|^2} \sum_{i=1}^p \langle \omega_{is} , v_i \rangle^2  \delta_{\lambda_i} \\
    \hat{\iota}_{oos} &\equiv \hat{\iota}_{oos} (\lambda) := \frac{1}{\|\omega_{oos}\|^2} \sum_{i=1}^p \langle \omega_{oos} , v_i \rangle^2  \delta_{\lambda_i} \\
    \hat{\iota}_d &\equiv \hat{\iota}_d (\lambda) := \frac{1}{\|\omega_{is}\| \|\omega_{oos}\|} \sum_{i=1}^p \langle \omega_{is} , v_i \rangle \langle v_i, \omega_{oos} \rangle  \delta_{\lambda_i}
\end{align*}

weakly converge to some (possibly signed) measures, denoted by $\iota_{is}$, $\iota_{oos}$ and $\iota_d$, respectively. 

Then, the following corollary follows from Proposition \ref{prop:expretgeneral}. It gives the mean return of the strategy in the general case.

\begin{proposition}
\label{prop:asyret_general}
    The asymptotic average return of the strategy in the general case $(\theta\neq0)$ is given by 
    \begin{align}
    \label{eq:asyretgeneral}
        \mathbb{E}\left[ r^{(s)}_{t+1}(z) \big| X  \right] &\xrightarrow[\substack{\mathstrut n,p,q \rightarrow \infty \\ p/n \rightarrow c\phi }]{\mathds{P}} \quad \|\omega_{is}\| \|\omega_{oos}\| \int h(\lambda;z,c\phi,\mu) d \iota_d(\lambda) := \mathcal{E}^{(d)}(z;c\phi,\iota_d),
    \end{align}

    with $h(\cdot)$ defined in Eq. \eqref{eq:h}.
\end{proposition}

\begin{proof}
    The proof directly follows from combining the asymptotic expressions of Proposition \ref{prop:expretgeneral} and noting that 

    \begin{align*}
        \iota_d :=\frac{\|\beta_{is}\|\|\beta_{oos}\| \zeta_d +  \|\beta_{oos}\|\|P \Sigma_w^{1/2}\theta_{is}\| \eta + \|P \Sigma_w^{1/2}\theta_{is}\|\|P \Sigma_w^{1/2}\theta_{oos}\| \nu + \|\beta_{is}\|\|P \Sigma_w^{1/2}\theta_{is}\| \psi}{\|\omega_{is}\|\|\omega_{oos}\|}.
    \end{align*}
\end{proof}



Denote $m_4 = \E[z_i^4]$. Moreover, let $\tilde{\lambda}_1(z) \geq \Tilde{\lambda}_2(z) \geq \dots \geq \Tilde{\lambda}_p(z)$ be the eigenvalues of the matrix

\begin{align*}
    \Big( I - z(zI+[1-c\phi+c\phi z m(-z)]\Sigma_x)^{-1}\Big)\diag(\omega_{oos})^2 \Big( I - z(zI+[1-c\phi+c\phi z m(-z)]\Sigma_x)^{-1}\Big)
\end{align*}

if $z>0$ or of the matrix $\big(I - ( I + c \phi s_0 \Sigma_x)^{-1} \big) \diag(\omega_{oos})^2 \big(I - ( I + c \phi s_0 \Sigma_x)^{-1} \big)$ in the limit $z \to 0^+$.

Similarly, denote by $\overline{\lambda}_1(z) \geq \overline{\lambda}_2(z) \geq \dots \geq \overline{\lambda}_p(z)$ the eigenvalues of $(zI+[1-c\phi+c\phi z m(-z)]\Sigma_x)^{-1} \diag(\omega_{oos})^2 (zI+[1-c\phi+c\phi z m(-z)]\Sigma_x)^{-1}$ if $z>0$ or of the matrix $( I + c \phi s_0 \Sigma_x)^{-1} \diag(\omega_{oos})^2( I + c \phi s_0 \Sigma_x)^{-1}$ in the limit $z \to 0^+$. For $ i \in \{ 1, \dots , p \}$, denote by $\tilde{v}_i(z)$ and $\overline{v}_i(z)$ the eigenvectors associated with the eigenvalue $\tilde{\lambda}_i(z)$ and $\overline{\lambda}_i(z)$, respectively. We make the following assumption on the local density of the state $\omega_{is}$.

\begin{assump}
\label{ass:esvd_diag_omega_oos}
The VESDs $\hat{\varpi}_1^z \equiv \hat{\varpi}_1^z(\lambda) := \| \omega_{is}\|^{-2} \sum_{i=1}^p \langle \omega_{is} , \tilde{v}_i(z) \rangle^2  \delta_{\tilde{\lambda}_i(z)}$ and $\hat{\varpi}_2^z \equiv \hat{\varpi}_2^z(\lambda) := \| \omega_{is}\|^{-2} \sum_{i=1}^p \langle \omega_{is} , \overline{v}_i(z) \rangle^2  \delta_{\overline{\lambda}_i(z)}$ weakly converge to some density measures $\varpi_1$ and $\varpi_2$, respectively. Moreover, these limiting measures have support on $[0,\infty)$.
\end{assump}

The following proposition gives the asymptotic second moment of the strategy's return. 

\begin{proposition}
\label{prop:secondmomstratret}
    Let $z>0$ and Assumptions (\ref{ass:iddistcov})-(\ref{ass:complexity}) and (\ref{ass:opbound})-(\ref{ass:esvd_diag_omega_oos}) hold. Denote $m_4 = \E[z_i^4]$. Then, 

    \begin{align*}
        \mathbb{E}\left[ \Big(r^{(s)}_{t+1}(z) \Big)^2 \Big| X  \right] &\xrightarrow[\substack{\mathstrut n,p,q \rightarrow \infty \\ p/n \rightarrow c\phi }]{\mathds{P}} \quad (1+\|\omega_{oos}\|^2) \mathcal{L} + 2 \big( \mathcal{E}^{(d)}(z;c\phi,\iota_d)\big)^2 + (m_4-3)\mathcal{K}
    \end{align*}

    where $\mathcal{E}^{(d)}(z;c\phi,\iota_d)$ is the asymptotic return of the strategy and is given in Eq. \eqref{eq:asyretgeneral}, $\mathcal{L}$ is the (limiting) leverage of the strategy, that is $\E[|\hat{\pi}_t(z)|^2 |X] \to \mathcal{L}$ in probability when $n,p,q \to \infty$ such that $p/n \to c\phi$, and is defined in the Appendix in Eq. \eqref{eq:leverage_gen_ridge} for ridge regularization and in Eq. \eqref{eq:leverage_ridgeless} when the penalty level is vanishing. 
    
    Similarly, $\mathcal{K}$ is provided in Eq. \eqref{eq:K1_K2_gen_ridge} while $\lim_{z \to 0^+} \mathcal{K}$ is given in Eq. \eqref{eq:K1_K2_ridgeless}.
\end{proposition}

\begin{proof}
\textbf{Ridge regularization}:
    For notational simplicity, we drop the "is" subscript for the vector of estimated coefficients, i.e. $\hat{\beta} \equiv \hat{\beta}_{is}$. The second moment of the strategy's return is given by

    \begin{align*}
        \mathbb{E}\left[ \Big(r^{(s)}_{t+1}(z) \Big)^2 \Big| X  \right] &= \mathbb{E}\left[ \Big(\hat{\beta}'x_tr_{t+1} \Big)^2 \Big| X  \right] \\
        &= \mathbb{E}\left[ \hat{\beta}'x_tr_{t+1} r_{t+1} x_t' \hat{\beta} \Big| X  \right] \\
        &= \mathbb{E}\left[ \hat{\beta}'x_t (x_t'\beta_{oos}+w_t'\theta_{oos}+e_t) (\beta_{oos}'x_t+\theta_{oos}'w_t+e_t) x_t' \hat{\beta} \Big| X  \right] \\
        &=\mathbb{E}\left[ \hat{\beta}'x_t \big[(x_t'\beta_{oos}+w_t'\theta_{oos})(\beta_{oos}'x_t+\theta_{oos}'w_t)+1 \big] x_t' \hat{\beta} \Big| X  \right] \\
        &= \mathbb{E}\left[ \hat{\beta}'x_t z_t' \omega_{oos} \omega_{oos}' z_t x_t' \hat{\beta} \Big| X  \right] + \mathbb{E}\left[ \hat{\beta}' \Sigma_x \hat{\beta} \Big| X  \right] \\
        &= \omega_{oos}' \mathbb{E}\left[ z_t z_t' \Sigma_x^{1/2} \hat{\beta} \hat{\beta}' \Sigma_x^{1/2} z_t z_t' \Big| X  \right] \omega_{oos} + \mathbb{E}\left[ \hat{\beta}' \Sigma_x \hat{\beta} \Big| X  \right]
    \end{align*}

    Direct calculation yields that, for any p.s.d. matrix $A$, we have the following decomposition 

    \begin{align}
    \label{eq:decomp_zzAzz}
        \mathbb{E}\left[ z_t z_t' A z_t z_t'  \right] = 2A + \tr(A)I +(m_4-3) \diag(A).
    \end{align}

    Choosing $A = \Sigma_x^{1/2} \hat{\beta} \hat{\beta}' \Sigma_x^{1/2}$ gives 

    \begin{align}
    \label{eq:secondmomentret_decomp}
    \begin{split}
        \mathbb{E}\left[ \Big(r^{(s)}_{t+1}(z) \Big)^2 \Big| X  \right] &= 2 \omega_{oos}' \Sigma_x^{1/2}\mathbb{E}\left[ \hat{\beta} \hat{\beta}'\Big| X  \right] \Sigma_x^{1/2} \omega_{oos} + \tr(\Sigma_x^{1/2}\mathbb{E}\left[ \hat{\beta} \hat{\beta}'\Big| X  \right] \Sigma_x^{1/2}) \|\omega_{oos}\|^2 \\
        &\quad + (m_4-3)\omega_{oos}'\diag(\Sigma_x^{1/2}\mathbb{E}\left[ \hat{\beta} \hat{\beta}'\Big| X  \right] \Sigma_x^{1/2})\omega_{oos} + \mathbb{E}\left[ \hat{\beta}' \Sigma_x \hat{\beta} \Big| X  \right] \\
        &= \big( 1+\|\omega_{oos}\|^2 \big)\tr(\Sigma_x^{1/2}\mathbb{E}\left[ \hat{\beta} \hat{\beta}'\Big| X  \right] \Sigma_x^{1/2}) \\
        & \quad + 2 \omega_{oos}' \Sigma_x^{1/2}\mathbb{E}\left[ \hat{\beta} \hat{\beta}'\Big| X  \right] \Sigma_x^{1/2} \omega_{oos} \\
        & \quad + (m_4-3)\omega_{oos}'\diag(\Sigma_x^{1/2}\mathbb{E}\left[ \hat{\beta} \hat{\beta}'\Big| X  \right] \Sigma_x^{1/2})\omega_{oos} \\
        & = T_1 + T_2 + (m_4-3) T_3
    \end{split}
    \end{align}

    The key term to study is the second moment of the regression coefficients vector conditional on the observed features, that is 
    
    \begin{align*}
        \mathbb{E}\left[ \hat{\beta} \hat{\beta}'\Big| X  \right] &= \mathbb{E}\left[ \frac{1}{n^2} (\hat{\Sigma}_x+zI)^{-1}X'yy'X(\hat{\Sigma}_x+zI)^{-1}\Big| X  \right] 
    \end{align*}

    Remark that we can rewrite the central outer product as

    \begin{align*}
        yy' &= (X \beta_{is} + W\theta_{is}) (X \beta_{is} + W\theta_{is})' + \xi\xi' + \xi (X \beta_{is} + W\theta_{is})' +  (X \beta_{is} + W\theta_{is}) \xi' 
    \end{align*}

    Injecting this in the previous equation yields

    \begin{align*}
        \mathbb{E}\left[ \hat{\beta} \hat{\beta}'\Big| X  \right] &=  R\hat{\Sigma}_x \Tilde{\omega}_{is} \Tilde{\omega}_{is}' \hat{\Sigma}_xR + \frac{1}{n^2} RX'\xi\xi'XR \\
        & \quad + \frac{1}{n^2} RX'(\xi (X \beta_{is} + W\theta_{is})' +  (X \beta_{is} + W\theta_{is}) \xi')XR
    \end{align*}

    where $\Tilde{\omega}_{is} = \Sigma_x^{-1/2} \omega_{is}$ and $R := (\hat{\Sigma}_x+zI)^{-1}$.

    Applying the law of large numbers, we have $n^{-2} RX'(\xi (X \beta_{is} + W\theta_{is})' +  (X \beta_{is} + W\theta_{is}) \xi')XR \to 0$ in $L^2$. Further, noticing that $n^{-2}\E[X'\xi \xi'X] = n^{-1} \hat{\Sigma}_x$, we have

    \begin{align*}
        \mathbb{E}\left[ \hat{\beta} \hat{\beta}'\Big| X  \right] \to  R\hat{\Sigma}_x \Tilde{\omega}_{is} \Tilde{\omega}_{is}' \hat{\Sigma}_xR + \frac{1}{n} R\hat{\Sigma}_xR
    \end{align*}

    in probability. Then, rewriting $\hat{\Sigma}_x = \hat{\Sigma}_x + zI - zI$, we can decompose the first term in the rhs as

    \begin{align*}
        R\hat{\Sigma}_x \Tilde{\omega}_{is} \Tilde{\omega}_{is}' \hat{\Sigma}_xR &= \Tilde{\omega}_{is} \Tilde{\omega}_{is}' + z^2 R \Tilde{\omega}_{is} \Tilde{\omega}_{is}' R \\
        &\quad - z R\Tilde{\omega}_{is} \Tilde{\omega}_{is}' - z \Tilde{\omega}_{is} \Tilde{\omega}_{is}'R
    \end{align*}

    Thus, we can come back to \eqref{eq:secondmomentret_decomp} and plug our decomposition of $\mathbb{E}\left[ \hat{\beta} \hat{\beta}'\Big| X  \right]$ into the different terms in the rhs.

    \underline{Analysis of $T_1$}:
    Recall that $T_1 = \big( 1+\|\omega_{oos}\|^2 \big)\tr(\Sigma_x^{1/2}\mathbb{E}\left[ \hat{\beta} \hat{\beta}'\Big| X  \right] \Sigma_x^{1/2})$. We focus on the non-constant factor,

    \begin{align*}
        \tr(\Sigma_x^{1/2}\mathbb{E}\left[ \hat{\beta} \hat{\beta}'\Big| X  \right] \Sigma_x^{1/2}) &= \tr(\omega_{is}\omega_{is}') + z^2 \tr(\Sigma_x^{1/2} R \Sigma_x^{-1/2} \omega_{is} \omega_{is}' \Sigma_x^{-1/2}R\Sigma_x^{1/2}) \\
        & \quad - 2z \tr(\Sigma_x^{1/2} R \Sigma_x^{-1/2} \omega_{is} \omega_{is}') + \frac{1}{n}\tr(\Sigma_x^{1/2} R\hat{\Sigma}_xR\Sigma_x^{1/2}) + o_n(1) \\
        &= \|\omega_{is}\|^2 + z^2 \omega_{is}' \Sigma_x^{-1/2}R\Sigma_xR\Sigma_x^{-1/2}\omega_{is} \\
        &\quad - 2z \omega_{is}'\Sigma_x^{1/2} R \Sigma_x^{-1/2} \omega_{is} + \frac{1}{n}\tr(\Sigma_x R\hat{\Sigma}_xR) + o_n(1)
    \end{align*}

    Substituting $\beta_{oos}$ and $\beta_{is}$ by $\Sigma_x^{-1/2} \omega_{is}$ in the definition of $F_n(z)$, we obtain by Lemma \ref{lem:asympcv},

    \begin{align*}
        - 2z \omega_{is}'\Sigma_x^{1/2} R \Sigma_x^{-1/2} \omega_{is} \xrightarrow[]{a.s.} -2z \omega_{is}'(zI+[1-c\phi+c\phi z m(-z)]\Sigma_x)^{-1} \omega_{is}
    \end{align*}

    Similarly, replacing $\beta$ by $\Sigma_x^{-1/2} \omega_{is}$ in [\cite{hastie2020surprises}, Theorem 5 \& 6], we have 

    \begin{align*}
       & z^2 \omega_{is}' \Sigma_x^{-1/2}R\Sigma_xR\Sigma_x^{-1/2}\omega_{is} \xrightarrow[]{a.s.} z^2 [1+c\phi m_1(-z)] \omega_{is}'(zI+[1-c\phi+c\phi z m(-z)]\Sigma_x)^{-2} \omega_{is}
    \end{align*}

    with 

    \begin{align*}
        m_1(z) \equiv m_1(z;c\phi,\mu) := \frac{\int \frac{\lambda^2[1-c\phi-c\phi z m(z)]}{(\lambda[1-c\phi - c \phi z m(z)]-z)^2} \dd \mu(\lambda)}{1 - c \phi \int \frac{z\lambda}{(\lambda[1-c\phi - c \phi z m(z)]-z)^2}\dd \mu(\lambda)}.
    \end{align*}

    Further, remark that 

    \begin{align*}
        &I + z^2 [1+c\phi m_1(-z)] (zI+[1-c\phi+c\phi z m(-z)]\Sigma_x)^{-2} -2z (zI+[1-c\phi+c\phi z m(-z)]\Sigma_x)^{-1} \\
        = & (zI+[1-c\phi+c\phi z m(-z)]\Sigma_x)^{-2} \Big[ z^2 c \phi m_1(-z)I + ([1-c\phi+c\phi z m(-z)]\Sigma_x)^2\Big]
    \end{align*}

    Finally, note that

    \begin{align*}
        \frac{1}{n}\tr(\Sigma_xR\hat{\Sigma}_xR) = c \phi \frac{\partial}{\partial z} \Big\{ \frac{z}{p} \tr(\Sigma_xR) \Big\}.
    \end{align*}

    From the proof of [\cite{hastie2020surprises}, Theorem 5 \& 6], we have 

    \begin{align*}
        \frac{1}{n}\tr(R\hat{\Sigma}_xR) \xrightarrow[]{a.s.} c\phi \int \frac{\lambda^2[1 - c\phi + c\phi z^2 m'(-z)]}{\big(z+ \lambda[1-c\phi+c\phi z m(-z)]\big)^2} \dd \mu (\lambda).
    \end{align*}

    Put together, we have 

    \begin{align}
    \label{eq:T1_lim_gen_ridge}
    \begin{split}
        T_1 \xrightarrow[]{a.s.}& \big( 1+\|\omega_{oos}\|^2 \big) \Bigg[ \|\omega_{is}\|^2 \int\frac{z^2 c\phi m_1(-z) + \lambda^2[1-c\phi+czm(-z)]^2}{(z + \lambda[1-c\phi+czm(-z)])^2} \dd \iota_{is}(\lambda)  \\
        & \quad + c\phi \int \frac{\lambda^2[1 - c\phi + c\phi z^2 m'(-z)]}{\big(z+ \lambda[1-c\phi+c\phi z m(-z)]\big)^2} \dd \mu (\lambda)\Bigg]
    \end{split}
    \end{align}

    \underline{Analysis of $T_2$}:

    Recall that $T_2 = 2 \omega_{oos}' \Sigma_x^{1/2}\mathbb{E}\left[ \hat{\beta} \hat{\beta}'\Big| X  \right] \Sigma_x^{1/2} \omega_{oos}$. Using the decomposition of $\mathbb{E}\left[ \hat{\beta} \hat{\beta}'\Big| X  \right]$, we have 

    \begin{align*}
        T_2 &= 2 (\omega_{is}'\omega_{oos})^2 + 2 z^2 \omega_{oos}' \Sigma_x^{1/2} R \Sigma_x^{-1/2}\omega_{is}\omega_{is}'\Sigma_x^{-1/2}R\Sigma_x^{1/2}\omega_{oos} \\
        &\quad -4z \omega_{oos}' \Sigma_x^{1/2} R \Sigma_x^{-1/2}\omega_{is}\omega_{is}'\omega_{oos} + \frac{2}{n} \omega_{oos}' \Sigma_x^{1/2} R\hat{\Sigma}_xR\Sigma_x^{1/2} \omega_{oos} + o_n(1) \\
        & = 2 (\omega_{is}'\omega_{oos})^2 + 2 \big(z \omega_{oos}' \Sigma_x^{1/2} R \Sigma_x^{-1/2}\omega_{is}\big)^2 - 4 \big(z \omega_{oos}' \Sigma_x^{1/2} R \Sigma_x^{-1/2}\omega_{is}\big) \big(\omega_{is}'\omega_{oos}\big) \\
        & \quad + \frac{2}{n} \omega_{oos}' \Sigma_x^{1/2} R\hat{\Sigma}_xR\Sigma_x^{1/2} \omega_{oos} + o_n(1)\\
        &= 2 \big( \omega_{is}'\omega_{oos} - z \omega_{oos}' \Sigma_x^{1/2} R \Sigma_x^{-1/2}\omega_{is}\big)^2 + \frac{2}{n} \omega_{oos}' \Sigma_x^{1/2} R\hat{\Sigma}_xR\Sigma_x^{1/2} \omega_{oos} + o_n(1)
    \end{align*}

    Substituting $\beta_{oos}$ and $\beta_{is}$ by $\Sigma_x^{-1/2} \omega_{oos}$ and $\Sigma_x^{-1/2} \omega_{is}$, respectively, in the definition of $F_n(z)$, we obtain by Lemma \ref{lem:asympcv},

    \begin{align*}
        - z \omega_{oos}' \Sigma_x^{1/2} (\hat{\Sigma}_x+zI)^{-1} \Sigma_x^{-1/2}\omega_{is} \xrightarrow[]{a.s.} -z \omega_{oos}'(zI+[1-c\phi+c\phi z m(-z)]\Sigma_x)^{-1} \omega_{is}.
    \end{align*}

    For the second term in the rhs of $T_2$, we can replace the trace in [\cite{hastie2020surprises}, Section A.3] by the quadratic form in $\Sigma_x^{1/2}\omega_{oos}$, i.e., $\omega \mapsto \omega'R\hat{\Sigma}_xR\omega$, and proceed as therein. Indeed, [\cite{knowles2017anisotropic}, Theorem 3.16 (i)] still holds in this setting. Thus, for any $D,\varepsilon>0$, there exists $C := C(\varepsilon,D)$ such that 

    \begin{align*}
        \Bigg| \frac{1}{n} \omega_{oos}' \Sigma_x^{1/2} R\hat{\Sigma}_xR\Sigma_x^{1/2} \omega_{oos} - \frac{c \phi \|\omega_{oos}\|_2^2}{p} \int \frac{\lambda^2[1 - c\phi + c\phi z^2 m'(-z)]}{\big(z+ \lambda[1-c\phi+c\phi z m(-z)]\big)^2} \dd \hat{\iota}_{oos} (\lambda) \Bigg| \leq \frac{C \|\omega_{oos}\|_2^2}{z^2 n^{\frac{1-\varepsilon}{2}}}
    \end{align*}

    holds with probability at least $1 - Cn^{-D}$. Notice that, for $z>0$,

    \begin{align*}
        \frac{c \phi \|\omega_{oos}\|_2^2}{p} \int \frac{\lambda^2[1 - c\phi + c\phi z^2 m'(-z)]}{\big(z+ \lambda[1-c\phi+c\phi z m(-z)]\big)^2} \dd \hat{\iota}_{oos} (\lambda) &\leq \frac{c \phi \|\omega_{oos}\|_2^2\lambda_{\max}[1 - c\phi + c\phi z^2 m'(-z)]}{zp}
    \end{align*}

    where $\lambda_{\max} = \|\Sigma_x\|_{op} < \infty$. The upper bound term converges to $0$ as $p \to \infty$. Terefore, proceeding as in [\cite{hastie2020surprises}, Section A.4], we obtain

    \begin{align*}
        \frac{1}{n} \omega_{oos}' \Sigma_x^{1/2} R\hat{\Sigma}_xR\Sigma_x^{1/2} \omega_{oos} \xrightarrow[]{a.s.} 0
    \end{align*}

    as $n,p \to \infty$.

    Finally, after some rearrangements, we obtain

    \begin{align}
    \label{eq:T2_lim_gen_ridge}
        T_2 \xrightarrow[]{a.s.} 2 \|\omega_{is}\|^2 \|\omega_{oos}\|^2 \Big(\int \frac{\lambda [1 - c\phi + c\phi z m(-z)]}{(z + [1 - c\phi + c \phi z m(-z)]\lambda)} \dd\iota_d\Big)^2  = \big(\mathcal{E}^{(d)}(z;c\phi,\iota_d)\big)^2,
    \end{align}

    where the equality follows from Proposition \ref{prop:asyret_general}.

    \underline{Analysis of $T_3$}:

    Recall that $T_3 = \omega_{oos}'\diag(\Sigma_x^{1/2}\mathbb{E}\left[ \hat{\beta} \hat{\beta}'\Big| X  \right] \Sigma_x^{1/2})\omega_{oos}$. Observe that, 

    \begin{align*}
        \omega_{oos}'\diag(\Sigma_x^{1/2}\mathbb{E}\left[ \hat{\beta} \hat{\beta}'\Big| X  \right] \Sigma_x^{1/2})\omega_{oos} = \tr(\diag(\omega_{oos})^2\Sigma_x^{1/2}\mathbb{E}\left[ \hat{\beta} \hat{\beta}'\Big| X  \right] \Sigma_x^{1/2}),
    \end{align*}

    where $\diag(\omega_{oos})$ is a diagonal matrix whose diagonal is the vector $\omega_{oos}$. 

    Thus, decomposing $\mathbb{E}\left[ \hat{\beta} \hat{\beta}'\Big| X  \right]$ once again, we have 

    \begin{align*}
        T_3 &= \tr(\diag(\omega_{oos})^2 \omega_{is}\omega_{is}') + z^2 \tr(\diag(\omega_{oos})^2 \Sigma_x^{1/2} R \Sigma_x^{-1/2}\omega_{is}\omega_{is}'\Sigma_x^{-1/2}R\Sigma_x^{1/2}) \\
        &\quad -2z\tr(\diag(\omega_{oos})^2 \Sigma_x^{1/2} R \Sigma_x^{-1/2}\omega_{is}\omega_{is}')  + \frac{1}{n} \tr(\diag(\omega_{oos})^2 \Sigma_x^{1/2} R \hat{\Sigma}_xR\Sigma_x^{1/2}) + o_n(1)\\
        &= \|\diag(\omega_{oos}) \omega_{is}\|^2 + z^2 \omega_{is}'\Sigma_x^{-1/2}R\Sigma_x^{1/2}\diag(\omega_{oos})^2 \Sigma_x^{1/2} R \Sigma_x^{-1/2}\omega_{is} \\
        &\quad -2z\omega_{is}'\diag(\omega_{oos})^2 \Sigma_x^{1/2} R \Sigma_x^{-1/2}\omega_{is} + \frac{1}{n} \tr(\Sigma_x^{1/2}\diag(\omega_{oos})^2 \Sigma_x^{1/2} R \hat{\Sigma}_xR) + o_n(1)
    \end{align*}

    Slightly modifying the proof of [\cite{hastie2020surprises}, Theorem 5 \& 6] (i.e., replacing $\beta$ by $\Sigma_x^{-1/2} \omega_{is}$ and $\Sigma$ by $\Sigma_x^{1/2}\diag(\omega_{oos})^2 \Sigma_x^{1/2}$) yields that  

    \begin{align*}
       & z^2 \omega_{is}'\Sigma_x^{-1/2}(\hat{\Sigma}_x+zI)^{-1}\Sigma_x^{1/2}\diag(\omega_{oos})^2 \Sigma_x^{1/2} (\hat{\Sigma}_x+zI)^{-1} \Sigma_x^{-1/2}\omega_{is} \\
        &\quad \xrightarrow[]{a.s.} z^2 [1+c\phi m_1(-z)] \\
        & \quad \times\omega_{is}'(zI+[1-c\phi+c\phi z m(-z)]\Sigma_x)^{-1} \diag(\omega_{oos})^2 (zI+[1-c\phi+c\phi z m(-z)]\Sigma_x)^{-1}\omega_{is}
    \end{align*}

    In addition, substituting $\beta_{oos}$ and $\beta_{is}$ by $\Sigma_x^{-1/2} \diag(\omega_{oos})^2 \omega_{is}$ and $\Sigma_x^{-1/2} \omega_{is}$, respectively, in the definition of $F_n(z)$, we obtain by Lemma \ref{lem:asympcv},

    \begin{align*}
        &- z\omega_{is}'\diag(\omega_{oos})^2 \Sigma_x^{1/2} (\hat{\Sigma}_x+zI)^{-1} \Sigma_x^{-1/2}\omega_{is} \\
        &\quad \xrightarrow[]{a.s.} -z \omega_{is}'\diag(\omega_{oos})^2(zI+[1-c\phi+c\phi z m(-z)]\Sigma_x)^{-1} \omega_{is}.
    \end{align*}

    Finally, we proceed as for the second term in the rhs of $T_2$ to determine the asymptotic limit of the last term in $T_3$. Namely, we can replace $\Sigma$ in the proof of [\cite{hastie2020surprises}, Theorem 5] by $\Sigma_x^{1/2} \diag(\omega_{oos})^2 \Sigma_x^{1/2}$, and proceed as therein. Indeed, [\cite{knowles2017anisotropic}, Theorem 3.16 (i)] still holds in this setting. Thus, for any $D,\varepsilon>0$, there exists $C := C(\varepsilon,D)$ such that 

    \begin{align*}
        &\Bigg| \frac{1}{n} \tr(\Sigma_x^{1/2}\diag(\omega_{oos})^2 \Sigma_x^{1/2} R \hat{\Sigma}_xR) \\
        &- \frac{c \phi}{p} [1 - c\phi + c\phi z^2 m'(-z)]\tr(\diag(\omega_{oos})^2 \Sigma_x^2 [zI + (1-c\phi+c\phi z m(-z))\Sigma_x]^{-2}) \Bigg| \leq \frac{C \|\omega_{oos}\|_2^2}{z^2 n^{\frac{1-\varepsilon}{2}}}
    \end{align*}

    holds with probability at least $1 - Cn^{-D}$. Notice that, for $z>0$,

    \begin{align*}
        & \quad \frac{c \phi}{p} [1 - c\phi + c\phi z^2 m'(-z)]\tr(\diag(\omega_{oos})^2 \Sigma_x^2 [zI + (1-c\phi+c\phi z m(-z))\Sigma_x]^{-2}) \\
        &\leq \frac{c \phi \|\omega_{oos}\|_2^2\lambda_{\max}^2[1 - c\phi + c\phi z^2 m'(-z)]}{zp}
    \end{align*}

    where $\lambda_{\max} = \|\Sigma_x\|_{op} < \infty$. The upper bound term converges to $0$ as $p \to \infty$. Terefore, proceeding as in [\cite{hastie2020surprises}, Section A.4], we obtain

    \begin{align*}
        \frac{1}{n} \tr(\Sigma_x^{1/2}\diag(\omega_{oos})^2 \Sigma_x^{1/2} R \hat{\Sigma}_xR) \xrightarrow[]{a.s.} 0
    \end{align*}

    as $n,p \to \infty$. 

    Rearranging the limiting terms, we get 

    \begin{align}
    \label{eq:T3_lim_gen_ridge}
    \begin{split}
        T_3 &\xrightarrow[]{a.s.} \omega_{is}' \big(I - z (zI+[1-c\phi+c\phi z m(-z)]\Sigma_x)^{-1} \big) \diag(\omega_{oos})^2 \\
        &\quad \times\big(I - z (zI+[1-c\phi+c\phi z m(-z)]\Sigma_x)^{-1} \big)\omega_{is} \\
        & \quad + c \phi z^2 m_1(-z) \\
        &\quad \times \omega_{is}'(zI+[1-c\phi+c\phi z m(-z)]\Sigma_x)^{-1} \diag(\omega_{oos})^2 (zI+[1-c\phi+c\phi z m(-z)]\Sigma_x)^{-1}\omega_{is} \\
        & = \|\omega_{is}\|^2 \Big( \int \lambda \dd \varpi_1^z(\lambda) + c \phi z^2 m_1(-z) \int \lambda \dd \varpi_2^z(\lambda) \Big)
    \end{split}
    \end{align}

    \underline{Conclusion}:

    Recall that from Eq. \eqref{eq:secondmomentret_decomp},

    \begin{align*}
        \mathbb{E}\left[ \Big(r^{(s)}_{t+1}(z) \Big)^2 \Big| X  \right] = T_1+T_2+(m_4-3)T_3.
    \end{align*}

    The desired result follows from taking the limits of the three terms in the rhs, which are given in Eqs. \eqref{eq:T1_lim_gen_ridge}, \eqref{eq:T2_lim_gen_ridge} and \eqref{eq:T3_lim_gen_ridge}, respectively. That is, 
    
    \begin{align*}
        \mathbb{E}\left[ \Big(r^{(s)}_{t+1}(z) \Big)^2 \Big| X  \right] &\xrightarrow[\substack{\mathstrut n,p,q \rightarrow \infty \\ p/n \rightarrow c\phi }]{\mathds{P}} \quad (1+\|\omega_{oos}\|^2) \mathcal{L}(z) + 2 \big(\mathcal{E}^{(d)}(z;c\phi,\iota_d)\big)^2 + (m_4-3)\mathcal{K}(z)
    \end{align*}

    where $\mathcal{L}(z)$ is the (limiting) leverage of the strategy, that is $\E[|\hat{\pi}_t(z)|^2 |X] \to \mathcal{L}(z)$ in probability when $n,p,q \to \infty$ such that $p/n \to c\phi$, and is equal to

    \begin{align}
    \label{eq:leverage_gen_ridge}
    \begin{split}
        \mathcal{L}(z) \equiv \mathcal{L}(z,c\phi, \mu, \iota_{is}) &:= \Bigg[ \|\omega_{is}\|^2 \int\frac{z^2 c\phi m_1(-z) + \lambda^2[1-c\phi+czm(-z)]^2}{(z + \lambda[1-c\phi+czm(-z)])^2} \dd \iota_{is}(\lambda)  \\
        & \quad + c\phi \int \frac{\lambda^2[1 - c\phi + c\phi z^2 m'(-z)]}{\big(z+ \lambda[1-c\phi+c\phi z m(-z)]\big)^2} \dd \mu (\lambda)\Bigg],
    \end{split}
    \end{align}

    with 

    \begin{align*}
        m_1(z) \equiv m_1(z;c\phi,\mu) := \frac{\int \frac{\lambda^2[1-c\phi-c\phi z m(z)]}{(\lambda[1-c\phi - c \phi z m(z)])^2} \dd \mu(\lambda)}{1 - c \phi \int \frac{z\lambda}{(\lambda[1-c\phi - c \phi z m(z)])^2}\dd \mu(\lambda)},
    \end{align*}

    and

    \begin{align}
    \label{eq:K1_K2_gen_ridge}
    \begin{split}
        \mathcal{K}(z) &\equiv \mathcal{K}(z;c\phi,\mu,\varpi_1^z, \varpi_2^z) := \|\omega_{is}\|^2 \Big( \int \lambda \dd \varpi_1^z(\lambda) + c \phi z^2 m_1(-z) \int \lambda \dd \varpi_2^z(\lambda) \Big).
    \end{split}
    \end{align}
    
\vspace{10pt}

\textbf{Vanishing regularization}:

    The proof consists in adjusting the limiting terms depending on $z$ to accommodate the scenario of vanishing regularization.



    
    The decomposition of the second moment of $\hat{\beta}$ remains the same, that is

    \begin{align*}
        \mathbb{E}\left[ \hat{\beta} \hat{\beta}'\Big| X  \right] &= \lim_{z \to 0^+} \Big\{\Tilde{\omega}_{is} \Tilde{\omega}_{is}' + z^2 R \Tilde{\omega}_{is} \Tilde{\omega}_{is}' R - z R\Tilde{\omega}_{is} \Tilde{\omega}_{is}' - z \Tilde{\omega}_{is} \Tilde{\omega}_{is}'R +\frac{1}{n} R\hat{\Sigma}_xR \Big\} + o_n(1)
    \end{align*}

    Hence, it suffices to determine the limit of the three terms $T_1$, $T_2$ and $T_3$ as in the proof of Proposition \ref{prop:volstratret} when we first consider a vanishing penalty level.  

    \underline{Analysis of $T_1$}:
    Recall that $T_1 = \big( 1+\|\omega_{oos}\|^2 \big)\tr(\Sigma_x^{1/2}\mathbb{E}\left[ \hat{\beta} \hat{\beta}'\Big| X  \right] \Sigma_x^{1/2})$. We focus on the non-constant factor,

    \begin{align*}
        \tr(\Sigma_x^{1/2}\mathbb{E}\left[ \hat{\beta} \hat{\beta}'\Big| X  \right] \Sigma_x^{1/2}) &= \|\omega_{is}\|^2 + \lim_{z \to 0^+} \Big\{ z^2 \omega_{is}' \Sigma_x^{-1/2}R\Sigma_xR\Sigma_x^{-1/2}\omega_{is} - 2z \omega_{is}'\Sigma_x^{1/2} R \Sigma_x^{-1/2} \omega_{is} \\
        & \quad + \frac{1}{n} R\hat{\Sigma}_xR \Big\} + o_n(1)
    \end{align*}

    Substituting $\beta_{oos}$ and $\beta_{is}$ by $\Sigma_x^{-1/2} \omega_{is}$ in the definition of $G_n(z)$, we obtain by Corollary \ref{cor:asycvridgeless},

    \begin{align*}
        - 2 \lim_{z \to 0^+} z \omega_{is}'\Sigma_x^{1/2} R \Sigma_x^{-1/2} \omega_{is} \xrightarrow[]{a.s.} -2\omega_{is}' ( I + c \phi s_0 \Sigma_x)^{-1} \omega_{is} 
    \end{align*}

    Similarly, replacing $\beta$ by $\Sigma_x^{-1/2} \omega_{is}$ in [\cite{hastie2020surprises}, Theorem 2 \& 3], we have 

    \begin{align*}
       & \lim_{z \to 0^+} z^2 \omega_{is}' \Sigma_x^{-1/2}R\Sigma_xR\Sigma_x^{-1/2}\omega_{is} \\
        &\quad \xrightarrow[]{a.s.} [1+c\phi s_0'] \omega_{is}'(I+c\phi s_0\Sigma_x)^{-2} \omega_{is}
    \end{align*}

    with 

    \begin{align*}
        s_0' \equiv s_0'(c\phi,\mu) := \frac{\int \frac{\lambda^2}{(1+c\phi s_0 \lambda)^2} \dd \mu(\lambda)}{\int \frac{\lambda}{(1+c\phi s_0 \lambda)^2}\dd \mu(\lambda)}.
    \end{align*}

    Further, remark that 

    \begin{align*}
        &I + [1+c\phi s_0'] (I+c\phi s_0\Sigma_x)^{-2} -2 ( I + c \phi s_0 \Sigma_x)^{-1} \\
        = & I + (I+c\phi s_0\Sigma_x)^{-2}\Big[c\phi ( s_0'I - 2 s_0 \Sigma_x) - I\Big]
    \end{align*}

    Finally, from the proof of [\cite{hastie2020surprises}, Theorem 2 \& 3], we have 

    \begin{align*}
        \lim_{z \to 0^+} \frac{1}{n}\tr(R\hat{\Sigma}_xR) \xrightarrow[]{a.s.} c\phi \frac{\int \frac{\lambda^2}{(1+s_0c\phi \lambda)^2} \dd \mu(\lambda)}{\int \frac{\lambda}{(1+s_0c\phi \lambda)^2} \dd \mu(\lambda)},
    \end{align*}

    as $n,p \to \infty$. Put together, we have 

    \begin{align}
    \label{eq:T1_lim_gen_ridgeless}
        T_1 \xrightarrow[]{a.s.} &\big( 1+\|\omega_{oos}\|^2 \big) \Bigg[ \|\omega_{is}\|^2 \int \Big(1+\frac{c\phi ( s_0'I - 2 s_0 \lambda) - 1}{(1+c\phi s_0\lambda)^{2}} \Big) \dd \iota_{is}(\lambda) + c\phi \frac{\int \frac{\lambda^2}{(1+s_0c\phi \lambda)^2} \dd \mu(\lambda)}{\int \frac{\lambda}{(1+s_0c\phi \lambda)^2} \dd \mu(\lambda)} \Bigg]
    \end{align}

    \underline{Analysis of $T_2$}:

    Recall that $T_2 = 2 \omega_{oos}' \Sigma_x^{1/2}\mathbb{E}\left[ \hat{\beta} \hat{\beta}'\Big| X  \right] \Sigma_x^{1/2} \omega_{oos}$. Using the decomposition of $\mathbb{E}\left[ \hat{\beta} \hat{\beta}'\Big| X  \right]$, we have 

    \begin{align*}
        T_2 &= 2 \big( \omega_{is}'\omega_{oos} - \lim_{z \to 0^+}z \omega_{oos}' \Sigma_x^{1/2} R \Sigma_x^{-1/2}\omega_{is}\big)^2 + \lim_{z \to 0^+}\frac{2}{n}\omega_{oos}' \Sigma_x^{1/2}R\hat{\Sigma}_xR\Sigma_x^{1/2} \omega_{oos} + o_n(1)
    \end{align*}

    Substituting $\beta_{oos}$ and $\beta_{is}$ by $\Sigma_x^{-1/2} \omega_{oos}$ and $\Sigma_x^{-1/2} \omega_{is}$, respectively, in the definition of $G_n(z)$, we obtain by Corollary \ref{cor:asycvridgeless},

    \begin{align*}
        - \lim_{z \to 0^+}z \omega_{oos}' \Sigma_x^{1/2} (\hat{\Sigma}_x+zI)^{-1} \Sigma_x^{-1/2}\omega_{is} \xrightarrow[]{a.s.} - \omega_{oos}'( I + c \phi s_0 \Sigma_x)^{-1} \omega_{is}.
    \end{align*}

    Regarding the second term, in the rhs of $T_2$, we have that for any $D,\varepsilon>0$, there exists $C := C(\varepsilon,D)$ such that 

    \begin{align*}
        &\Bigg| \lim_{z \to 0^+}\frac{1}{n} \omega_{oos}' \Sigma_x^{1/2} R\hat{\Sigma}_xR\Sigma_x^{1/2} \omega_{oos} - \frac{c \phi \|\omega_{oos}\|_2^2}{p} \cdot \frac{\int \frac{\lambda^2}{(1+\lambda c\phi s_0)^2} \dd \hat{\iota}_{oos} (\lambda)}{\int \frac{\lambda}{(1+\lambda c\phi s_0)^2} \dd \mu(\lambda)} \Bigg| \leq \frac{C \|\omega_{oos}\|_2^2}{n^{\frac{1}{7}}}
    \end{align*}

    holds with probability at least $1 - Cn^{-D}$. Notice that the denominator of the deterministic term has the following lower bound

    \begin{align*}
        \int \frac{\lambda}{(1+\lambda c\phi s_0)^2} \dd \mu (\lambda) &\geq \lambda_{\min}(\Sigma_x(I + c \phi s_0 \Sigma_x)^{-2}) \\
        & \geq \lambda_{\min}(\Sigma_x) \|(I + c \phi s_0 \Sigma_x)^{2}\|_{op} \\
        & > \tau,
    \end{align*}

where the last inequality follows from Assumption \ref{ass:mineigen} and $\|(I + c \phi s_0 \Sigma_x)^{2}\|_{op} \geq 1$. Besides, Assumption \ref{ass:opbound} implies that

\begin{align*}
    \int \frac{\lambda^2}{(1+\lambda c\phi s_0)^2} \dd \hat{\iota}_d (\lambda) \leq \|\Sigma_x^2\|_{op} \leq \tau^{-2} 
\end{align*}

All in all, we obtain

\begin{align*}
    \frac{c \phi \|\omega_{oos}\|_2^2}{p} \cdot \frac{\int \frac{\lambda^2}{(1+\lambda c\phi s_0)^2} \dd \mu (\lambda)}{\int \frac{\lambda}{(1+\lambda c\phi s_0)^2} \dd \hat{\iota}_d (\lambda)} \le \frac{c \phi \tau^{-3} \|\omega_{oos}\|_2^2}{p} \to 0.
\end{align*}

Therefore, proceeding as in [\cite{hastie2020surprises}, Section A.4], we obtain

    \begin{align*}
        \lim_{z \to 0^+}\frac{1}{n} \omega_{oos}' \Sigma_x^{1/2} R\hat{\Sigma}_xR\Sigma_x^{1/2} \omega_{oos} \xrightarrow[]{a.s.} 0
    \end{align*}

    as $n,p \to \infty$.

    After some rearrangements, we finally get

    \begin{align}
    \label{eq:T2_lim_gen_ridgeless}
        T_2 \xrightarrow[]{a.s.} 2 \|\omega_{is}\|^2 \|\omega_{oos}\|^2 \Big(\int \frac{c \phi s_0 \lambda}{1 + c \phi s_0 \lambda} \dd\iota_d\Big)^2 = \big( \mathcal{E}^{(d)}(z;c\phi,\iota_d) \big)^2,
    \end{align}

    where the equality follows from Proposition \ref{prop:asyret_general}.

    \underline{Analysis of $T_3$}:

    Recall that $T_3 = \omega_{oos}'\diag(\Sigma_x^{1/2}\mathbb{E}\left[ \hat{\beta} \hat{\beta}'\Big| X  \right] \Sigma_x^{1/2})\omega_{oos}$. Decomposing $\mathbb{E}\left[ \hat{\beta} \hat{\beta}'\Big| X  \right]$ once again, we have 

    \begin{align*}
        T_3 &= \lim_{z\to0^+} \Big\{\|\diag(\omega_{oos}) \omega_{is}\|^2 + z^2 \omega_{is}'\Sigma_x^{-1/2}R\Sigma_x^{1/2}\diag(\omega_{oos})^2 \Sigma_x^{1/2} R \Sigma_x^{-1/2}\omega_{is} \\
        &\quad -2z\omega_{is}'\diag(\omega_{oos})^2 \Sigma_x^{1/2} R \Sigma_x^{-1/2}\omega_{is} + \frac{1}{n} \tr( \diag(\omega_{oos})^2 \Sigma_x^{1/2}R\hat{\Sigma}_xR\Sigma_x^{1/2}) \Big\} + o_n(1)
    \end{align*}

    Slightly modifying the proof of [\cite{hastie2020surprises}, Theorem 2 \& 3] (i.e., replacing $\beta$ by $\Sigma_x^{-1/2} \omega_{is}$ and $\Sigma$ by $\Sigma_x^{1/2}\diag(\omega_{oos})^2 \Sigma_x^{1/2}$) yields that  

    \begin{align*}
       & \lim_{z\to0^+} z^2 \omega_{is}'\Sigma_x^{-1/2}(\hat{\Sigma}_x+zI)^{-1}\Sigma_x^{1/2}\diag(\omega_{oos})^2 \Sigma_x^{1/2} (\hat{\Sigma}_x+zI)^{-1} \Sigma_x^{-1/2}\omega_{is} \\
        &\quad \xrightarrow[]{a.s.} [1+c\phi s_0'] \omega_{is}'(I+c\phi s_0\Sigma_x)^{-1} \diag(\omega_{oos})^2 (I+c\phi s_0\Sigma_x)^{-1} \omega_{is}
    \end{align*}

    In addition, substituting $\beta_{oos}$ and $\beta_{is}$ by $\Sigma_x^{-1/2} \diag(\omega_{oos})^2 \omega_{is}$ and $\Sigma_x^{-1/2} \omega_{is}$, respectively, in the definition of $G_n(z)$, we obtain by Corollary \ref{cor:asycvridgeless},

    \begin{align*}
        &- \lim_{z\to0^+}z\omega_{is}'\diag(\omega_{oos})^2 \Sigma_x^{1/2} (\hat{\Sigma}_x+zI)^{-1} \Sigma_x^{-1/2}\omega_{is} \\
        &\quad \xrightarrow[]{a.s.} -\omega_{is}' \diag(\omega_{oos})^2 ( I + c \phi s_0 \Sigma_x)^{-1} \omega_{is}.
    \end{align*}

    Regarding the last term, in the rhs of $T_3$, we have that for any $D,\varepsilon>0$, there exists $C := C(\varepsilon,D)$ such that 

    \begin{align*}
        &\Bigg| \lim_{z \to 0^+}\frac{1}{n} \tr( \diag(\omega_{oos})^2 \Sigma_x^{1/2}R\hat{\Sigma}_xR\Sigma_x^{1/2}) - \frac{c \phi}{p} \cdot \frac{\tr(\diag(\omega_{oos})^2 (I+c\phi s_0 \Sigma_x)^{-2} \Sigma_x^2)}{\frac{1}{p} \tr( (I+c\phi s_0 \Sigma_x)^{-2} \Sigma_x)}  \Bigg| \leq \frac{C \|\omega_{oos}\|_2^2}{n^{\frac{1}{7}}}
    \end{align*}

    holds with probability at least $1 - Cn^{-D}$. Notice that the denominator of the deterministic term has the following lower bound

    \begin{align*}
        \frac{1}{p} \tr( (I+c\phi s_0 \Sigma_x)^{-2} \Sigma_x) &\geq \lambda_{\min}((I+c\phi s_0 \Sigma_x)^{-1})^2 \lambda_{\min} (\Sigma_x) > \tau,
    \end{align*}

where the last inequality follows from Assumption \ref{ass:mineigen} and $\lambda_{\min}((I+c\phi s_0 \Sigma_x)^{-1}) \geq 1$. Besides, Assumption \ref{ass:opbound} implies that

\begin{align*}
    \tr(\diag(\omega_{oos})^2 (I+c\phi s_0 \Sigma_x)^{-2} \Sigma_x^2) &\leq \|\omega_{oos}\|_2^2 \|(I+c\phi s_0 \Sigma_x)^{-1} \Sigma_x\|_{op}^2 \\
    & \leq \|\omega_{oos}\|_2^2 \|(I+c\phi s_0 \Sigma_x)^{-1} \|_{op} \|\Sigma_x\|_{op}^2 \\
    & \leq \Big( \frac{\|\omega_{oos}\|_2}{\tau}\Big)^2
\end{align*}

All in all, we obtain

\begin{align*}
    \frac{c \phi}{p} \cdot \frac{\tr(\diag(\omega_{oos})^2 (I+c\phi s_0 \Sigma_x)^{-2} \Sigma_x^2)}{\frac{1}{p} \tr( (I+c\phi s_0 \Sigma_x)^{-2} \Sigma_x)}  \leq \frac{c \phi}{p} \cdot \|\omega_{oos}\|_2^2 \tau^{-3} \to 0.
\end{align*}

Therefore, proceeding as in [\cite{hastie2020surprises}, Section A.4], we obtain

    \begin{align*}
        \lim_{z \to 0^+}\frac{1}{n} \tr( \diag(\omega_{oos})^2 \Sigma_x^{1/2}R\hat{\Sigma}_xR\Sigma_x^{1/2}) \xrightarrow[]{a.s.} 0
    \end{align*}

    as $n,p \to \infty$.

    Rearranging the limiting terms, we get 

    \begin{align}
    \label{eq:T3_lim_gen_ridgeless}
    \begin{split}
        T_3 &\xrightarrow[]{a.s.} \omega_{is}' \big(I - ( I + c \phi s_0 \Sigma_x)^{-1} \big) \diag(\omega_{oos})^2 \big(I - ( I + c \phi s_0 \Sigma_x)^{-1} \big)\omega_{is} \\
        & \quad + c\phi s_0' \omega_{is}'(I+c\phi s_0\Sigma_x)^{-1} \diag(\omega_{oos})^2 (I+c\phi s_0\Sigma_x)^{-1} \omega_{is}\\
        & = \|\omega_{is}\|^2 \Big( \int \lambda \dd \varpi_1^0(\lambda) + c\phi s_0' \int \lambda \dd \varpi_2^0(\lambda) \Big)
    \end{split}
    \end{align}

    \underline{Conclusion}:

    Recall that from Eq. \eqref{eq:secondmomentret_decomp},

    \begin{align*}
        \mathbb{E}\left[ \Big(r^{(s)}_{t+1}(z) \Big)^2 \Big| X  \right] = T_1+T_2+(m_4-3)T_3.
    \end{align*}

    The desired result follows from taking the limits of the three terms in the rhs, which are given in Eqs. \eqref{eq:T1_lim_gen_ridgeless}, \eqref{eq:T2_lim_gen_ridgeless} and \eqref{eq:T3_lim_gen_ridgeless}, respectively. That is, in the ridgeless scenario, we have

    \begin{align*}
        \mathbb{E}\left[ \Big(r^{(s)}_{t+1}(z) \Big)^2 \Big| X  \right] &\xrightarrow[\substack{z\to 0^+ \\\mathstrut n,p,q \rightarrow \infty \\ p/n \rightarrow c\phi }]{\mathds{P}} \quad (1+\|\omega_{oos}\|^2) \mathcal{L} + 2 \big( \mathcal{E}^{(d)}(z;c\phi,\iota_d) \big)^2 + (m_4-3)\mathcal{K}
    \end{align*}

    where $\mathcal{L}$ is the (limiting) leverage of the strategy, that is 

    \begin{align}
    \label{eq:leverage_ridgeless}
        \mathcal{L} \equiv \mathcal{L}(c\phi, \mu, \iota_{is}) := \|\omega_{is}\|^2 \int \Big(1+\frac{c\phi ( s_0'I - 2 s_0 \lambda) - 1}{(1+c\phi s_0\lambda)^{2}} \Big) \dd \iota_{is}(\lambda)+ c\phi s_0',
    \end{align}

    with 

    \begin{align*}
        s_0' \equiv s_0'(c\phi,\mu) := \frac{\int \frac{\lambda^2}{(1+c\phi s_0 \lambda)^2} \dd \mu(\lambda)}{\int \frac{\lambda}{(1+c\phi s_0 \lambda)^2}\dd \mu(\lambda)},
    \end{align*}

    and

    \begin{align}
    \label{eq:K1_K2_ridgeless}
    \begin{split}
        \mathcal{K} &\equiv \mathcal{K}(c\phi,\mu,\varpi_1^0, \varpi_2^0) := \|\omega_{is}\|^2 \Big( \int \lambda \dd \varpi_1^0(\lambda) + c\phi s_0' \int \lambda \dd \varpi_2^0(\lambda) \Big).
    \end{split}
    \end{align} 

\end{proof}

\vspace{10pt}

\textbf{Proof of Proposition \ref{prop:volstratret}}: 

The proof simply follows from decomposing the variance with the Koenig-Huygens formula and applying the limiting results provided in Propositions \ref{prop:asyret_general} and \ref{prop:secondmomstratret}.

\qed

\subsection{Misspecified model (Isotropic features)}
\label{subsub:misspecifiedisotropic}

\noindent \textbf{Proof of Proposition \ref{prop:expretdriftisomis}}:

$\Sigma = I$, then $\Sigma_{xw} = 0$. Moreover, as $X$ and $W$ are independent and $\E[W] = 0$, we have 

\begin{equation*}
    \E \Big[\hat{\Sigma}_{xw}\Big| X \Big] = \frac{1}{n} X' \E \big[W\big| X \big] = \frac{1}{n} X' \E \big[W \big] = 0
\end{equation*}

Consequently,

\begin{align*}
    \mathcal{M}(z) & = \big( \beta_{oos}' \Sigma_x + \theta_{oos}' \Sigma_{xw}' \big) (zI + \hat{\Sigma}_x)^{-1} \E \Big[\hat{\Sigma}_{xw}\Big| X \Big] \theta_{is} \\
    & \quad + \theta_{oos}' \Sigma_{xw}' (zI + \hat{\Sigma}_x)^{-1} \big( \hat{\Sigma}_x \beta_{is} +  \frac{1}{n} X' \xi \big) \\
    & = 0.
\end{align*}

Finally, we end up in the same setting as Proposition \ref{prop:isotropic} or \ref{prop:isotropicridgeless} where the complexity ratio is now the misspecified one, i.e. $c\phi$. \qed

\vspace{10mm}

\begin{proposition}
\label{prop:secondmomstratret_iid}
    Let $z>0$ and Assumptions (\ref{ass:iddistcov})-(\ref{ass:complexity}) and (\ref{ass:opbound})-(\ref{ass:esvd_diag_omega_oos}) hold. Denote $m_4 = \E[z_i^4]$. Then, 

    \begin{align*}
        \mathbb{E}\left[ \big( r^{(s)}_{t+1}(z)\big)^2 \Big| X  \right] &\xrightarrow[\substack{\mathstrut n,p,q \rightarrow \infty \\ p/n \rightarrow c\phi }]{\mathds{P}}  (1+S_{oos}) \mathcal{L} + 2( \hspace{0.4em} \underbrace{f(z;c\phi) \langle \beta_{oos}, \beta_{is} \rangle}_{=\mathcal{E}^{(d)}(z;c\phi)} \hspace{0.4em} )^2 + (m_4-3) f(z;c\phi) \|\beta_{is} \circ \beta_{oos}\|^2 
    \end{align*}

    where $S_{oos} = \|\beta_{oos}\|^2+\|\theta_{oos}\|_2^2$ is the amount of signal contained in the out-of-sample return and $\mathcal{L}$ is the (limiting) leverage of the strategy, that is, $\E[|\hat{\pi}_t(z)|^2 |X] \to \mathcal{L}$ in probability when $n,p,q \to \infty$ such that $p/n \to c\phi$. The latter quantity is defined in the Appendix in Eq. \eqref{eq:leverage_ridge_iid} for ridge regularization and in Eq. \eqref{eq:leverage_ridgeless_iid} when the penalty level is vanishing. Recall that $f$ is defined in Eq. \eqref{eq:f_isotropic} for the ridge scenario and in Eq. \eqref{eq:isotropicridgeless} for the ridgeless case.
\end{proposition}

\begin{proof}

\textbf{Ridge regularization}:
    For notational simplicity, we drop the "is" subscript for the vector of estimated coefficients, i.e. $\hat{\beta} \equiv \hat{\beta}_{is}$. The second moment of the strategy's return is given by

    \begin{align*}
        \mathbb{E}\left[ \Big(r^{(s)}_{t+1}(z) \Big)^2 \Big| X  \right] &= \mathbb{E}\left[ \Big(\hat{\beta}'x_tr_{t+1} \Big)^2 \Big| X  \right] \\
        &= \mathbb{E}\left[ \hat{\beta}'x_tr_{t+1} r_{t+1} x_t' \hat{\beta} \Big| X  \right] \\
        &= \mathbb{E}\left[ \hat{\beta}'x_t (x_t'\beta_{oos}+w_t'\theta_{oos}+e_t) (\beta_{oos}'x_t+\theta_{oos}'w_t+e_t) x_t' \hat{\beta} \Big| X  \right] \\
        &=\mathbb{E}\left[ \hat{\beta}'x_t \big[(x_t'\beta_{oos}+w_t'\theta_{oos})(\beta_{oos}'x_t+\theta_{oos}'w_t)+1 \big] x_t' \hat{\beta} \Big| X  \right] \\
        &= \mathbb{E}\left[ \hat{\beta}'x_t x_t' \beta_{oos} \beta_{oos}' x_t x_t' \hat{\beta} \Big| X  \right] + \mathbb{E}\left[ \hat{\beta}'x_t w_t' \theta_{oos} \theta_{oos}' w_t x_t' \hat{\beta} \Big| X  \right] + \mathbb{E}\left[ \hat{\beta}' \hat{\beta} \Big| X  \right] \\
        &= \beta_{oos}' \mathbb{E}\left[ x_t x_t' \hat{\beta} \hat{\beta}' x_t x_t' \Big| X  \right] \beta_{oos} + \theta_{oos}' \mathbb{E}\left[ w_t x_t' \hat{\beta} \hat{\beta}' x_t w_t' \Big| X  \right] \theta_{oos} + \mathbb{E}\left[ \hat{\beta}' \hat{\beta} \Big| X  \right]
    \end{align*}

    Direct calculation yields that, for any p.s.d. matrix $A \in \R^{p \times p}$ and i.i.d. random vectors $z_1 \in \R^q$, $z_2 \in \R^p$, with mean $0$ and unit variance, we have 

    \begin{align*}
        \mathbb{E}\left[ z_1 z_2' A z_2 z_1'  \right] = \tr(A)I.
    \end{align*}

    Choosing $A = \hat{\beta} \hat{\beta}'$, $z_1 = w_t$ and $z_2 = x_t$ gives 

    \begin{align*}
        \theta_{oos}' \mathbb{E}\left[ w_t x_t' \hat{\beta} \hat{\beta}' x_t w_t' \Big| X  \right] \theta_{oos} = \|\theta_{oos}\|_2^2 \tr(\mathbb{E}[ \hat{\beta} \hat{\beta}' | X  ])   
    \end{align*}

    Moreover, using the decomposition in Eq. \eqref{eq:decomp_zzAzz}, we get 

    \begin{align}
    \label{eq:secondmomentret_decomp_iid}
    \begin{split}
        \mathbb{E}\left[ \Big(r^{(s)}_{t+1}(z) \Big)^2 \Big| X  \right] &= (1+\|\beta_{oos}\|^2+\|\theta_{oos}\|_2^2) \tr(\mathbb{E}[ \hat{\beta} \hat{\beta}' | X  ]) + 2 \beta_{oos}' \mathbb{E}\left[ \hat{\beta} \hat{\beta}'\Big| X  \right] \beta_{oos} \\
        & \quad + (m_4-3)\beta_{oos}'\diag(\mathbb{E}\left[ \hat{\beta} \hat{\beta}'\Big| X  \right] )\beta_{oos} \\
        & = T_4 + T_5 + (m_4-3) T_6
    \end{split}
    \end{align}

    The key term to study is the second moment of the regression coefficients vector conditional on the observed features, that is 
    
    \begin{align}
    \label{eq:condexp_betabeta}
        \mathbb{E}\left[ \hat{\beta} \hat{\beta}'\Big| X  \right] &= \mathbb{E}\left[ \frac{1}{n^2} (\hat{\Sigma}_x+zI)^{-1}X'yy'X(\hat{\Sigma}_x+zI)^{-1}\Big| X  \right] 
    \end{align}

    Over the unobserved training covariates, $W$, the second moment of the in-sample labels is given by

    \begin{align*}
        \E_{W}[yy'] &= X \beta_{is} \beta_{is}' X' + \E[W \theta_{is} \theta_{is}' W'] + \xi \xi' + 2 X \beta_{is} \xi'.
    \end{align*}

    Observe that $\E[(W \theta_{is} \theta_{is}' W')_{ij}] = \E[\theta_{is}' w_i w_j' \theta_{is}]$. Hence, $\E[W \theta_{is} \theta_{is} W'] = \|\theta_{is}\|_2^2 I$ and 

    \begin{align*}
        \E_{W}[yy'] &= X \beta_{is} \beta_{is}' X' + \|\theta_{is}\|_2^2I + \xi \xi' + 2 X \beta_{is} \xi'.
    \end{align*}

    Injecting this in Eq. \eqref{eq:condexp_betabeta} yields

    \begin{align*}
        \mathbb{E}\left[ \hat{\beta} \hat{\beta}'\Big| X  \right] &=  R\hat{\Sigma}_x \beta_{is} \beta_{is}' \hat{\Sigma}_xR + \frac{\|\theta_{is}\|_2^2}{n} R\hat{\Sigma}_xR + \frac{1}{n^2}RX'(\xi\xi'+2X\beta_{is}\xi')XR,
    \end{align*}

    where $R := (\hat{\Sigma}_x+zI)^{-1}$. Applying the law of large numbers, we have $\frac{2}{n^2}RX'X\beta_{is}\xi'XR \to 0$ in $L^2$. Further, noticing that $n^{-2}\E[X'\xi \xi'X] = n^{-1} \hat{\Sigma}_x$, we have 

    \begin{align*}
        \mathbb{E}\left[ \hat{\beta} \hat{\beta}'\Big| X  \right] &\to  R\hat{\Sigma}_x \beta_{is} \beta_{is}' \hat{\Sigma}_xR + \frac{1+\|\theta_{is}\|_2^2}{n} R\hat{\Sigma}_xR,
    \end{align*}

    



    in probability. Rewriting $\hat{\Sigma}_x = \hat{\Sigma}_x + zI - zI$, we can decompose the first term in the rhs as

    \begin{align*}
        R\hat{\Sigma}_x \beta_{is}\beta_{is}' \hat{\Sigma}_xR &= \beta_{is}\beta_{is}' + z^2 R \beta_{is}\beta_{is}' R - z R\beta_{is}\beta_{is}' - z \beta_{is}\beta_{is}'R
    \end{align*}

    Thus, we can come back to \eqref{eq:secondmomentret_decomp_iid} and plug our decomposition of $\mathbb{E}\left[ \hat{\beta} \hat{\beta}'\Big| X  \right]$ into the different terms in the rhs. 
    

    \underline{Analysis of $T_4$}:
    Recall that $T_4 = (1+\|\beta_{oos}\|^2+\|\theta_{oos}\|_2^2) \tr(\mathbb{E}[ \hat{\beta} \hat{\beta}' | X  ])$. We focus on the non-constant factor,

    \begin{align*}
        \tr(\mathbb{E}\left[ \hat{\beta} \hat{\beta}'\Big| X  \right]) &= \tr(\beta_{is}\beta_{is}') + z^2 \tr( R \beta_{is} \beta_{is}' R) - 2z \tr( R \beta_{is} \beta_{is}') + \frac{1+\|\theta_{is}\|_2^2}{n} \tr(R\hat{\Sigma}_xR) + o_n(1) \\
        &= \|\beta_{is}\|^2 + z^2 \beta_{is}' R^2\beta_{is} - 2z \beta_{is}'R \beta_{is} + \frac{1+\|\theta_{is}\|_2^2}{n} \tr(R\hat{\Sigma}_xR) + o_n(1)
    \end{align*}

    Substituting $\beta_{oos}$ by $\beta_{is}$ in the definition of $F_n(z)$, we obtain by Lemma \ref{lem:asympcv},

    \begin{align*}
        - 2z \beta_{is}' R \beta_{is} \xrightarrow[]{a.s.} - \frac{2z\|\beta_{is}\|_2^2}{1-c\phi+z [1+c\phi m(-z)]}
    \end{align*}

    Similarly, replacing $\beta$ by $\beta_{is}$ in [\cite{hastie2020surprises}, Theorem 5 \& 6], we have 

    \begin{align*}
       & z^2 \beta_{is}'R^2\beta_{is} \xrightarrow[]{a.s.} \frac{z^2 [1+c\phi m_1(-z)]}{\big(1-c\phi+z [1+c\phi m(-z)]\big)^2}\|\beta_{is}\|_2^2
    \end{align*}

    with

    \begin{align*}
        m_1(z) \equiv m_1(z;c\phi) := \frac{1-c\phi-c\phi z m(z)}{(1-c\phi - c \phi z m(z)-z)^2 - c\phi z}.
    \end{align*}

    Finally, note that

    \begin{align*}
        \frac{1}{n}\tr(R\hat{\Sigma}_xR) = c \phi \frac{\partial}{\partial z} \Big\{ \frac{z}{p} \tr(R) \Big\}.
    \end{align*}

    Again, from the proof of [\cite{hastie2020surprises}, Theorem 5 \& 6], we have 

    \begin{align*}
        \frac{1}{n}\tr(R\hat{\Sigma}_xR) \xrightarrow[]{a.s.} c\phi \frac{1 - c\phi + c\phi z^2 m'(-z)}{\big(1-c\phi+z [1+c\phi m(-z)]\big)^2}.
    \end{align*}

    Put together, we have 

    \begin{align}
    \label{eq:T4_lim}
    \begin{split}
        T_4 \xrightarrow[]{a.s.} &(1+\|\beta_{oos}\|^2+\|\theta_{oos}\|_2^2) \Bigg[\frac{z^2 c\phi m_1(-z) + [1-c\phi+czm(-z)]^2}{\big(1-c\phi+z [1+c\phi m(-z)]\big)^2} \|\beta_{is}\|_2^2 \\
        & \quad + (1+\|\theta_{is}\|_2^2)c\phi\frac{1 - c\phi + c\phi z^2 m'(-z)}{\big(1-c\phi+z [1+c\phi m(-z)]\big)^2} \Bigg] \\
        & = (1+\|\beta_{oos}\|^2+\|\theta_{oos}\|_2^2) \Bigg[f(z;c\phi) \|\beta_{is}\|_2^2 + (1+\|\theta_{is}\|_2^2)c\phi\frac{1 - c\phi + c\phi z^2 m'(-z)}{\big(1-c\phi+z [1+c\phi m(-z)]\big)^2} \Bigg]
    \end{split}
    \end{align}

    where the equality follows from simple algebra and $f$ is defined in Eq. \eqref{eq:f_isotropic}. 

    \underline{Analysis of $T_5$}:

    Recall that $T_5 = 2 \beta_{oos}' \mathbb{E}\left[ \hat{\beta} \hat{\beta}'\Big| X  \right] \beta_{oos}$. Using the decomposition of $\mathbb{E}\left[ \hat{\beta} \hat{\beta}'\Big| X  \right]$, we have 

    \begin{align*}
        T_5 = 2 \big( \beta_{is}'\beta_{oos} - z \beta_{oos}' R\beta_{is}\big)^2 +\frac{1+\|\theta_{is}\|_2^2}{n} \beta_{oos}'R\hat{\Sigma}_xR\beta_{oos} + o_n(1)
    \end{align*}

    By Lemma \ref{lem:asympcv}, we obtain

    \begin{align*}
        - z \beta_{oos}' R \beta_{is} \xrightarrow[]{a.s.} -\frac{z}{1-c\phi+z[1+c\phi m(-z)]} \langle \beta_{oos}, \beta_{is} \rangle.
    \end{align*}

    For the second term in the rhs of $T_5$, we can replace the trace in [\cite{hastie2020surprises}, Section A.3] by the quadratic form in $\beta_{oos}$, i.e., $\beta \mapsto \beta'R\hat{\Sigma}_xR\beta$, and proceed as therein. Indeed, [\cite{knowles2017anisotropic}, Theorem 3.16 (i)] still holds in this setting. Thus, for any $D,\varepsilon>0$, there exists $C := C(\varepsilon,D)$ such that 

    \begin{align*}
        \Bigg| \frac{1}{n} \beta_{oos}'R\hat{\Sigma}_xR\beta_{oos} - \frac{c\phi}{p} \cdot \frac{1 - c\phi + c\phi z^2 m'(-z)}{\big(1-c\phi+z [1+c\phi m(-z)]\big)^2}\|\beta_{oos}\|_2^2 \Bigg| \leq \frac{C \| \beta_{oos}\|_2^2}{z^2 n^{\frac{1-\varepsilon}{2}}}
    \end{align*}

    holds with probability at least $1 - Cn^{-D}$. Notice that $c\phi p^{-1} \big[1 - c\phi + c\phi z^2 m'(-z)\big]/\big[1-c\phi+z [1+c\phi m(-z)]\big]^2\|\beta_{oos}\|_2^2 \to 0$ as $n,p \to \infty$. Finally, proceeding as in [\cite{hastie2020surprises}, Section A.4], we obtain

    \begin{align*}
        \frac{1}{n} \beta_{oos}'R\hat{\Sigma}_xR\beta_{oos} \xrightarrow[]{a.s.} 0.
    \end{align*}

    Therefore, after some rearrangements, we obtain

    \begin{align}
    \label{eq:T5_lim}
        T_5 \xrightarrow[]{a.s.} 2 \Big(\frac{1-c\phi+c\phi z m(-z)}{1-c\phi+z[1+c\phi m(-z)]} \Big)^2 \langle \beta_{oos}, \beta_{is} \rangle^2 = 2 (f(z;c\phi)\langle \beta_{oos}, \beta_{is} \rangle)^2
    \end{align}

    \underline{Analysis of $T_6$}:

    Recall that $T_6 = \beta_{oos}'\diag(\mathbb{E}\left[ \hat{\beta} \hat{\beta}'\Big| X  \right] )\beta_{oos}$. Observe that, 

    \begin{align*}
        \beta_{oos}'\diag(\mathbb{E}\left[ \hat{\beta} \hat{\beta}'\Big| X  \right] )\beta_{oos} = \tr(\diag(\beta_{oos})^2\mathbb{E}\left[ \hat{\beta} \hat{\beta}'\Big| X  \right]),
    \end{align*}

    where $\diag(\beta_{oos})$ is a diagonal matrix whose diagonal is the vector $\beta_{oos}$. 

    Thus, decomposing $\mathbb{E}\left[ \hat{\beta} \hat{\beta}'\Big| X  \right]$ once again, we have 

    \begin{align*}
        T_6 &= \tr(\diag(\beta_{oos})^2 \beta_{is}\beta_{is}') + z^2 \tr(\diag(\beta_{oos})^2 R \beta_{is}\beta_{is}'R) \\
        &\quad -2z\tr(\diag(\beta_{oos})^2 R \beta_{is}\beta_{is}') + \frac{1+\|\theta_{is}\|_2^2}{n} \tr(\diag(\beta_{oos})^2R\hat{\Sigma}_xR) + o_n(1)\\
        &= \|\beta_{is} \circ \beta_{oos}\|^2 + z^2 \beta_{is}'R\diag(\beta_{oos})^2 R \beta_{is} -2z\beta_{is}'\diag(\beta_{oos})^2 R \beta_{is} \\
        &\quad + \frac{1+\|\theta_{is}\|_2^2}{n} \tr(\diag(\beta_{oos})^2R\hat{\Sigma}_xR) + o_n(1)
    \end{align*}

    Slightly modifying the proof of [\cite{hastie2020surprises}, Theorem 5 \& 6] (i.e., replacing $\beta$ by $\beta_{is}$ and $\Sigma$ by $\diag(\beta_{oos})^2 $) yields that  

    \begin{align*}
       & z^2 \beta_{is}'R\diag(\beta_{oos})^2 R \beta_{is} \\
        &\quad \xrightarrow[]{a.s.} \frac{z^2 [1+c\phi m_1(-z)]}{(1-c\phi+z [1+c\phi m(-z)])^2}\|\beta_{is} \circ \beta_{oos}\|^2
    \end{align*}

    In addition, substituting $\beta_{oos}$ by $ \diag(\beta_{oos})^2 \beta_{is}$ in the definition of $F_n(z)$, we obtain by Lemma \ref{lem:asympcv},

    \begin{align*}
        &- z\beta_{is}'\diag(\beta_{oos})^2 R \beta_{is} \\
        &\quad \xrightarrow[]{a.s.}  -\frac{z}{1-c\phi+z [1+c\phi m(-z)]}\|\beta_{is} \circ \beta_{oos}\|^2.
    \end{align*}

    For the last term in the rhs of $T_6$, we can proceed as in [\cite{hastie2020surprises}, Section A.3] simply by replacing $\Sigma$ in Eq. (86) by $\diag(\beta_{oos})^2$. Here again, [\cite{knowles2017anisotropic}, Theorem 3.16 (i)] still holds in this setting. Thus, for any $D,\varepsilon>0$, there exists $C := C(\varepsilon,D)$ such that 

    \begin{align*}
        \Bigg| \frac{1}{n} \tr(\diag(\beta_{oos})^2R\hat{\Sigma}_xR) - \frac{c\phi}{p} \cdot \frac{1 - c\phi + c\phi z^2 m'(-z)}{\big(1-c\phi+z [1+c\phi m(-z)]\big)^2}\|\beta_{oos}\|_2^2 \Bigg| \leq \frac{C\|\beta_{oos}\|_2^2}{z^2 n^{\frac{1-\varepsilon}{2}}}
    \end{align*}

    holds with probability at least $1 - Cn^{-D}$. Notice that $c\phi p^{-1} \big[1 - c\phi + c\phi z^2 m'(-z)\big]/\big[1-c\phi+z [1+c\phi m(-z)]\big]^2\|\beta_{oos}\|_2^2 \to 0$ as $n,p \to \infty$. Finally, proceeding as in [\cite{hastie2020surprises}, Section A.4], we obtain

    \begin{align*}
        \frac{1}{n} \tr(\diag(\beta_{oos})^2R\hat{\Sigma}_xR) \xrightarrow[]{a.s.} 0.
    \end{align*}

    Rearranging the limiting terms, we get 

    \begin{align}
    \label{eq:T6_lim}
    \begin{split}
        T_6 &\xrightarrow[]{a.s.} \Bigg[\Big(  \frac{1-c\phi+czm(-z)}{1-c\phi+z [1+c\phi m(-z)]} \Big)^2 + \frac{c \phi z^2 m_1(-z)}{(1-c\phi+z [1+c\phi m(-z)])^2} \Bigg]\|\beta_{is} \circ \beta_{oos}\|^2\\
        &\quad = f(z;c\phi)\|\beta_{is} \circ \beta_{oos}\|^2
    \end{split}
    \end{align}

    where the equality follows by simple algebra and $f$ is defined in Eq. \eqref{eq:f_isotropic}.

    \underline{Conclusion}:

    Recall that from Eq. \eqref{eq:secondmomentret_decomp_iid},

    \begin{align*}
        \mathbb{E}\left[ \Big(r^{(s)}_{t+1}(z) \Big)^2 \Big| X  \right] = T_4+T_5+(m_4-3)T_6.
    \end{align*}

    The desired result follows from taking the limits of the three terms in the rhs, which are given in Eqs. \eqref{eq:T4_lim}, \eqref{eq:T5_lim} and \eqref{eq:T6_lim}, respectively. That is, 
    
    \begin{align*}
        \mathbb{E}\left[ \Big(r^{(s)}_{t+1}(z) \Big)^2 \Big| X  \right] &\xrightarrow[\substack{\mathstrut n,p,q \rightarrow \infty \\ p/n \rightarrow c\phi }]{\mathds{P}} \quad (1+S_{oos}) \mathcal{L} + 2 (f(z;c\phi) \langle \beta_{oos}, \beta_{is} \rangle)^2 + (m_4-3) f(z;c\phi) \|\beta_{is} \circ \beta_{oos}\|^2 
    \end{align*}

    where $S_{oos} = \|\beta_{oos}\|^2+\|\theta_{oos}\|_2^2$ is the amount of signal contained in the out-of-sample return and $\mathcal{L}$ is the (limiting) leverage of the strategy, that is $\E[|\hat{\pi}_t(z)|^2 |X] \to \mathcal{L}$ in probability when $n,p,q \to \infty$ such that $p/n \to c\phi$, and is equal to

    \begin{align}
    \label{eq:leverage_ridge_iid}
        \mathcal{L} \equiv \mathcal{L}(z,c\phi) := f(z;c\phi)\|\beta_{is}\|^2 + (1+\|\theta_{is}\|_2^2)\mathcal{H} .
    \end{align}

    with 

    \begin{align}
    \label{eq:K1_K2_ridge_iid}
    \begin{split}
        \mathcal{H} &\equiv \mathcal{H}(z;c\phi) := c\phi\frac{1 - c\phi + c\phi z^2 m'(-z)}{\big(1-c\phi+z [1+c\phi m(-z)]\big)^2}
    \end{split}
    \end{align}
    
    \qed

\vspace{10mm}

\textbf{Vanishing regularization:}

The proof consists in adjusting the limiting terms depending on $z$ to accommodate the scenario of vanishing regularization.

The decomposition of the second moment of $\hat{\beta}$ remains the same, that is

    \begin{align*}
        \mathbb{E}\left[ \hat{\beta} \hat{\beta}'\Big| X  \right] &= \lim_{z \to 0^+} \Big\{\beta_{is}\beta_{is}' + z^2 R \beta_{is}\beta_{is}' R  - z R\beta_{is}\beta_{is}' - z \beta_{is}\beta_{is}'R + \frac{1+\|\theta_{is}\|_2^2}{n} R\hat{\Sigma}_xR\Big\} + o_n(1)
    \end{align*}

    Hence, it suffices to determine the limit of the three terms $T_4$, $T_5$ and $T_6$ as in the proof of Proposition \ref{prop:volstratret_iid} when we first consider a vanishing penalty level. 

    \underline{Analysis of $T_4$}:
    Recall that $T_4 = (1+\|\beta_{oos}\|^2+\|\theta_{oos}\|_2^2) \tr(\mathbb{E}[ \hat{\beta} \hat{\beta}' | X  ])$. We focus on the non-constant factor,

    \begin{align*}
        \lim_{z \to 0^+} \tr(\mathbb{E}\left[ \hat{\beta} \hat{\beta}'\Big| X  \right]) = \lim_{z \to 0^+} \Big\{\|\beta_{is}\|^2 + z^2 \beta_{is}' R^2\beta_{is} - 2z \beta_{is}'R \beta_{is} + \frac{1+\|\theta_{is}\|_2^2}{n} \tr(R\hat{\Sigma}_xR) \Big\} + o_n(1)
    \end{align*}

    Substituting $\beta_{oos}$ by $\beta_{is}$ in the definition of $G_n(z)$, we obtain by Corollary \ref{cor:asycvridgeless},

    \begin{align*}
        - \lim_{z \to 0^+} 2z \beta_{is}' R \beta_{is} \xrightarrow[]{a.s.} - 2 \Big(1 - \frac{1}{c \phi}\Big) \|\beta_{is}\|_2^2 \mathds{1}_{\{c\phi>1\}}
    \end{align*}

    Similarly, replacing $\beta$ by $\beta_{is}$ in [\cite{hastie2020surprises}, Theorem 3], we have 

    \begin{align*}
       & \lim_{z \to 0^+} z^2 \beta_{is}'R^2\beta_{is} \xrightarrow[]{a.s.} \Big(1 - \frac{1}{c \phi}\Big)\|\beta_{is}\|_2^2 \mathds{1}_{\{c\phi>1\}}
    \end{align*}

    Finally, from the proof of [\cite{hastie2020surprises}, Theorem 1], we have 

    \begin{align*}
        \lim_{z \to 0^+}\frac{1}{n}\tr(R\hat{\Sigma}_xR) \xrightarrow[]{a.s.} \frac{c \phi}{1 - c\phi}\mathds{1}_{\{c\phi<1\}} + \frac{1}{c\phi-1}\mathds{1}_{\{c\phi>1\}}.
    \end{align*}

    Put together, we have 

    \begin{align}
    \label{eq:T4_lim_iid_ridgeless}
    \begin{split}
        \lim_{z \to 0^+} T_4 \xrightarrow[]{a.s.} &(1+\|\beta_{oos}\|^2+\|\theta_{oos}\|_2^2) \Bigg[ \Big(1-\frac{c \phi - 1}{c \phi}\mathds{1}_{\{c\phi>1\}} \Big)\|\beta_{is}\|_2^2 \\
        & \quad + (1+\|\theta_{is}\|_2^2) \Big( \frac{c \phi}{1 - c\phi}\mathds{1}_{\{c\phi<1\}} + \frac{1}{c\phi-1}\mathds{1}_{\{c\phi>1\}} \Big) \Bigg] \\
        & = (1+\|\beta_{oos}\|^2+\|\theta_{oos}\|_2^2) \Bigg[ f(c\phi)\|\beta_{is}\|_2^2 \\
        & \quad + (1+\|\theta_{is}\|_2^2) \Big( \frac{c \phi}{1 - c\phi}\mathds{1}_{\{c\phi<1\}} + \frac{1}{c\phi-1}\mathds{1}_{\{c\phi>1\}} \Big) \Bigg]
    \end{split}
    \end{align}

    where $f$ is defined in Eq. \eqref{eq:isotropicridgeless}.

    \underline{Analysis of $T_5$}:

    Recall that $T_5 = 2 \beta_{oos}' \mathbb{E}\left[ \hat{\beta} \hat{\beta}'\Big| X  \right] \beta_{oos}$. Using the decomposition of $\mathbb{E}\left[ \hat{\beta} \hat{\beta}'\Big| X  \right]$, we have 

    \begin{align*}
        \lim_{z \to 0^+} T_5 = 2 \big( \beta_{is}'\beta_{oos} - \lim_{z \to 0^+} z \beta_{oos}' R\beta_{is}\big)^2 + \lim_{z \to 0^+} \frac{1+\|\theta_{is}\|_2^2}{n}\beta_{oos}'R\hat{\Sigma}_xR\beta_{oos} + o_n(1)
    \end{align*}

    By Corollary \ref{cor:asycvridgeless}, we obtain

    \begin{align*}
        - \lim_{z \to 0^+} z \beta_{is}' R \beta_{oos} \xrightarrow[]{a.s.} -\Big(1 -  \frac{1}{c \phi}\Big) \langle\beta_{is}, \beta_{oos} \rangle \mathds{1}_{\{c\phi>1\}}
    \end{align*}

    For the second term in the rhs of $T_5$, we can replace the trace in [\cite{hastie2020surprises}, Section A.3] by the quadratic form in $\beta_{oos}$, i.e., $\beta \mapsto \beta'R\hat{\Sigma}_xR\beta$, and proceed as therein. Indeed, [\cite{knowles2017anisotropic}, Theorem 3.16 (i)] still holds in this setting. Define, $s(-z) = r(-z)/(c\phi)$ where we recall that $r$ is the companion Stieltjes transform of the ESD defined in \eqref{eq:esd}. Thus, using the same approximation argument as in [\cite{hastie2020surprises}, Section B.3] for vanishing regularization, we get that for any $D,\varepsilon>0$, there exists $C := C(\varepsilon,D)$ such that 

    \begin{align*}
        \Bigg| \lim_{z \to 0^+} \frac{1}{n} \beta_{oos}'R\hat{\Sigma}_xR\beta_{oos} - \lim_{z \to 0^+}\frac{c\phi}{p} \cdot \frac{c\phi s'(-z)}{1+c\phi s(-z)}\|\beta_{oos}\|_2^2 \Bigg| \leq \frac{C \| \beta_{oos}\|_2^2}{n^{1/7}}
    \end{align*}

    holds with probability at least $1 - Cn^{-D}$. From [\cite{hastie2020surprises}, Section B.3], we have that $s'$ is bounded by constants depending uniquely on $\tau$ for any $z \in [-\tau^{-1},0]$. Therefore, we have  

    \begin{align*}
        \lim_{z \to 0^+}\frac{c\phi}{p} \cdot \frac{c\phi s'(-z)}{1+c\phi s(-z)}\|\beta_{oos}\|_2^2 \to 0,
    \end{align*}

     as $n,p \to \infty$. Finally, proceeding as in [\cite{hastie2020surprises}, Section A.4], we obtain

    \begin{align*}
        \lim_{z \to 0^+} \frac{1}{n} \beta_{oos}'R\hat{\Sigma}_xR\beta_{oos} \xrightarrow[]{a.s.} 0.
    \end{align*}

    Therefore, after some rearrangements, we obtain

    \begin{align}
    \label{eq:T5_lim_iid_ridgeless}
        \lim_{z \to 0^+} T_5 \xrightarrow[]{a.s.} 2 \Big(1-\frac{c\phi - 1}{c\phi} \mathds{1}_{\{c\phi>1\}} \Big)^2 \langle \beta_{oos}, \beta_{is} \rangle^2 = 2 (f(c\phi) \langle \beta_{oos}, \beta_{is} \rangle)^2.
    \end{align}

    where $f$ is defined in Eq. \eqref{eq:isotropicridgeless}.

    \underline{Analysis of $T_6$}:

    Recall that $T_6 = \beta_{oos}'\diag(\mathbb{E}\left[ \hat{\beta} \hat{\beta}'\Big| X  \right] )\beta_{oos}$. Decomposing $\mathbb{E}\left[ \hat{\beta} \hat{\beta}'\Big| X  \right]$ once again, we have 

    \begin{align*}
        \lim_{z \to 0^+}T_6 &= \tr(\diag(\beta_{oos})^2 \beta_{is}\beta_{is}') + \lim_{z \to 0^+} z^2 \tr(\diag(\beta_{oos})^2 R \beta_{is}\beta_{is}'R) \\
        &\quad -2 \lim_{z \to 0^+}z\tr(\diag(\beta_{oos})^2 R \beta_{is}\beta_{is}') + \lim_{z \to 0^+}\frac{1+\|\theta_{is}\|_2^2}{n} \tr(\diag(\beta_{oos})^2R\hat{\Sigma}_xR) + o_n(1)\\
        &= \|\beta_{is} \circ \beta_{oos}\|^2 + \lim_{z \to 0^+}z^2 \beta_{is}'R\diag(\beta_{oos})^2 R \beta_{is} -2\lim_{z \to 0^+}z\beta_{is}'\diag(\beta_{oos})^2 R \beta_{is} \\
        &\quad + \lim_{z \to 0^+}\frac{1+\|\theta_{is}\|_2^2}{n} \tr(\diag(\beta_{oos})^2R\hat{\Sigma}_xR) + o_n(1)
    \end{align*}

    Slightly modifying the proof of [\cite{hastie2020surprises}, Theorem 2, 5 \& 6] (i.e., replacing $\beta$ by $\beta_{is}$ and $\Sigma$ by $\diag(\beta_{oos})^2 $) yields that 


    \begin{align*}
       \lim_{z \to 0^+}z^2 \beta_{is}'R\diag(\beta_{oos})^2 R \beta_{is} \xrightarrow[]{a.s.} \frac{c\phi-1}{c\phi} \|\beta_{is} \circ \beta_{oos}\|^2 \mathds{1}_{\{c\phi>1\}}
    \end{align*}

    In addition, substituting $\beta_{oos}$ by $ \diag(\beta_{oos})^2 \beta_{is}$ in the definition of $F_n(z)$, we obtain by Corollary \ref{cor:asycvridgeless},

    \begin{align*}
        &- z\beta_{is}'\diag(\beta_{oos})^2 R \beta_{is} \xrightarrow[]{a.s.}  -\frac{c\phi-1}{c\phi}\|\beta_{is} \circ \beta_{oos}\|^2\mathds{1}_{\{c\phi>1\}}.
    \end{align*}

    For the last term in the rhs of $T_6$, we can proceed as in [\cite{hastie2020surprises}, Section A.3] simply by replacing $\Sigma$ in Eq. (86) by $\diag(\beta_{oos})^2$. Here again, [\cite{knowles2017anisotropic}, Theorem 3.16 (i)] still holds in this setting. Thus, using the same approximation argument as in [\cite{hastie2020surprises}, Section B.3] for vanishing regularization, for any $D,\varepsilon>0$, there exists $C := C(\varepsilon,D)$ such that 

    \begin{align*}
        \Bigg| \lim_{z \to 0^+} \frac{1}{n} \tr(\diag(\beta_{oos})^2R\hat{\Sigma}_xR) - \lim_{z \to 0^+}\frac{c\phi}{p} \cdot \frac{c\phi s'(-z)}{1+c\phi s(-z)}\|\beta_{oos}\|_2^2 \Bigg| \leq \frac{C\|\beta_{oos}\|_2^2}{ n^{1/7}}
    \end{align*}

    holds with probability at least $1 - Cn^{-D}$. In the analysis of $T_5$, we have showed that the deterministic term vanishes. That is,

    \begin{align*}
        \lim_{z \to 0^+}\frac{c\phi}{p} \cdot \frac{c\phi s'(-z)}{1+c\phi s(-z)}\|\beta_{oos}\|_2^2 \to 0,
    \end{align*}

     as $n,p \to \infty$. Finally, proceeding as in [\cite{hastie2020surprises}, Section A.4], we obtain

    \begin{align*}
        \lim_{z \to 0^+} \frac{1}{n} \tr(\diag(\beta_{oos})^2R\hat{\Sigma}_xR) \xrightarrow[]{a.s.} 0.
    \end{align*}

    Rearranging the limiting terms, we get 

    \begin{align}
    \label{eq:T6_lim_iid_ridgeless}
    \begin{split}
        \lim_{z \to 0^+} T_6 &\xrightarrow[]{a.s.} \Big( 1 - \frac{c\phi-1}{c\phi}\mathds{1}_{\{c\phi>1\}} \Big) \|\beta_{is} \circ \beta_{oos}\|^2 = f(c\phi)\|\beta_{is} \circ \beta_{oos}\|^2,
    \end{split}
    \end{align}

    with $f$ defined in Eq. \eqref{eq:isotropicridgeless}.

    \underline{Conclusion}:

    Recall that from Eq. \eqref{eq:secondmomentret_decomp_iid},

    \begin{align*}
        \mathbb{E}\left[ \Big(r^{(s)}_{t+1}(z) \Big)^2 \Big| X  \right] = T_4+T_5+(m_4-3)T_6.
    \end{align*}

    The desired result follows from taking the limits of the three terms in the rhs, which are given in Eqs. \eqref{eq:T4_lim_iid_ridgeless}, \eqref{eq:T5_lim_iid_ridgeless} and \eqref{eq:T6_lim_iid_ridgeless}, respectively. That is, 
    
    \begin{align*}
        \mathbb{E}\left[ \Big(r^{(s)}_{t+1}(z) \Big)^2 \Big| X  \right] &\xrightarrow[\substack{\mathstrut n,p,q \rightarrow \infty \\ p/n \rightarrow c\phi }]{\mathds{P}} \quad (1+S_{oos}) \mathcal{L} + 2 (f(c\phi) \langle \beta_{oos}, \beta_{is} \rangle)^2 + (m_4-3) f(c\phi) \|\beta_{is} \circ \beta_{oos}\|^2 
    \end{align*}

    where $S_{oos} = \|\beta_{oos}\|^2+\|\theta_{oos}\|_2^2$ is the amount of signal contained in the out-of-sample return and $\mathcal{L}$ is the (limiting) leverage of the strategy, that is $\E[|\hat{\pi}_t(z)|^2 |X] \to \mathcal{L}$ in probability when $n,p,q \to \infty$ such that $p/n \to c\phi$, and is equal to

    \begin{align}
    \label{eq:leverage_ridgeless_iid}
        \mathcal{L} \equiv \mathcal{L}(z,c\phi) := f(c\phi)\|\beta_{is}\|^2 + (1+\|\theta_{is}\|_2^2)\mathcal{H} .
    \end{align}

    with 
    
    \begin{align}
    \label{eq:K1_K2_ridgeless_iid}
    \begin{split}
        \mathcal{H} &\equiv \lim_{z \to 0^+} \mathcal{H}(z;c\phi) := \frac{c \phi}{1 - c\phi}\mathds{1}_{\{c\phi<1\}} + \frac{1}{c\phi-1}\mathds{1}_{\{c\phi>1\}}.
    \end{split}
    \end{align}
    
\end{proof}

\textbf{Proof of Proposition \ref{prop:volstratret_iid}}: 

The proof simply follows from decomposing the variance with the Koenig-Huygens formula and applying the limiting results provided in Propositions \ref{prop:expretdriftisomis} and \ref{prop:secondmomstratret_iid}.

\qed

\subsection{Misspecified model (Latent space model)}
\label{subsub:misspecifiedlatent}

We can also consider a latent space model as presented in \cite{hastie2020surprises} and \cite{misiakiewicz2023six}. This model will allow us to highlight a situation where the expected returns of the strategy depend on the alignment between the eigenvectors of the population covariance matrix and the true vector of regression coefficients. Remark that this symmetry plays no role when the features are isotropic (see Appendix \ref{subsub:misspecifiedisotropic}, e.g.).  

Namely, we consider a model in which the dependent variable is linear in a latent covariate vector, $z_i \in \R^d$. Though, we are only able to observe a vector of $p (\geq d)$ covariates, $x_i$, whose components are linear in the latent features. Formally, 

\begin{align}
\label{eq:dgplatent}
\begin{split}
    (z_i,b_i,u_i) &\sim \mathcal{N}(0,I_d) \times  \mathcal{N}(0,\sigma_b^2) \times \mathcal{N}(0,1),  \quad i=1,\dots,n  \\
    y_i &= z'_i\theta + b_i,  \quad i=1,\dots,n, \\
    x_i &= W z_i + u_i, \quad i=1,\dots,n,
\end{split}
\end{align}
where the $n$ random draws are independent. $(b_i)_{1 \leq i \leq n}$ and $(u_{ij})_{1 \leq i \leq n, 1 \leq j \leq p}$ are independent. Moreover, the latent covariates, $z_i$, are independent from both noise variables. $W \in \R^{p \times d}$ is the linear transformation relating the observed features to the latent covariates.

DGP \eqref{eq:dgplatent} is a special case of DGP \eqref{eq:dgpmis} where $w_i = 0$, $P_x \sim \mathcal{N}(0,\Sigma)$ and $P_e \sim \mathcal{N}(0,\sigma^2)$. Formally, this new DGP can be written 

\begin{align*}
\label{eq:dgplatentspe}
\begin{split}
    (x_i,e_i) &\sim \mathcal{N}(0,\Sigma) \times  \mathcal{N}(0,\sigma^2),  \quad i=1,\dots,n \\
    y_i &= x'_i\beta+ e_i,  \quad i=1,\dots,n. 
\end{split}
\end{align*}

It is straightforward that $\Sigma = \E[x_i x_i'] = \E[(W z_i + u_i) (W z_i + u_i)'] = I + WW'$. Further, taking $\beta = W (I+W'W)^{-1} \theta$ and $e_i = b_i + \theta ' (I+W'W)^{-1} (z_i - W'u_i)$ yields $y_i = z'_i\theta + b_i$. Hence, $\sigma^2 = \E[e_i^2] = \sigma_v^2 + \theta' (I+W'W)^{-1} \theta $.

Let $\Upsilon = d/p \in (0,1)$. For the sake of simplicity, we also consider $W$ to be proportional to an orthogonal matrix, i.e. $W'W = \Upsilon^{-1}I_d$. 

\begin{proposition}
\label{prop:latentspace}
    Let $z>0$. Assume the data is generated as per \eqref{eq:dgplatent}. Then,

    \begin{align*}
        \mathbb{E}\left[ r^{(s)}_{t+1}(z) \big| X  \right]  \xrightarrow[\substack{\mathstrut n,p,d \rightarrow \infty \\ p/n \rightarrow c \\ d/p \rightarrow \Upsilon}]{\mathds{P}} g(z;c,\Upsilon) \langle \theta_{is}, \theta_{oos} \rangle 
    \end{align*}

    with 

    \begin{align*}
        g(z;c,\Upsilon) = \frac{1}{1+\Upsilon} \Big( 1 - \frac{z}{(1+\Upsilon^{-1})(czm + 1 - c) +z} \Big)
    \end{align*}

    and $m := m(-z;c,\Upsilon)$ the unique solution in $\C_+$ to

    \begin{align*}
        m = \frac{1-\Upsilon}{1-c+czm+z} + \frac{\Upsilon}{(1+\Upsilon^{-1})(1-c+czm)+z}
    \end{align*}

    Moreover, in the ridgeless limit,

    \begin{align*}
        \mathbb{E}\left[ r^{(s)}_{t+1}(z) \big| X  \right]  \xrightarrow[\substack{\mathstrut n,p,d \rightarrow \infty \\ p/n \rightarrow c \\ d/p \rightarrow \Upsilon \\ z \to 0}]{\mathds{P}} g(c,\Upsilon) \langle \theta_{is}, \theta_{oos} \rangle 
    \end{align*}

    with

    \begin{align*}
        g(c,\Upsilon) = \lim_{z \to 0} g(z;c) = \frac{1}{1+\Upsilon} \Big( 1 - \frac{1}{1 + (1+\Upsilon^{-1})c s_0} \mathds{1}_{c>1} \Big)
    \end{align*}

    and $s_0 := s_0(c,\Upsilon)$ the non-negative solution to 

    \begin{equation*}
        1 - \frac{1}{c} = \frac{1 - \Upsilon}{1+cs_0} + \frac{\Upsilon}{1+c(1+\Upsilon^{-1})s_0}
    \end{equation*}

\end{proposition}

\begin{proof} 

    First, we need to determine the measures $\hat{\mu}$ and $\hat{\zeta}_d$. Note that $\Sigma$ admits the following eigenvalue decomposition

    \begin{equation}
    \label{eq:sigmaeigendecomp}
        \Sigma = PDP' + I_p = \sum_{i=1}^d (1+\Upsilon^{-1}) v_i v_i' + \sum_{i=d+1}^p v_i v_i',
    \end{equation}

    with $(v_i)_{1 \leq i \leq d}$ the columns of $P \in \R^{p \times d}$ and, for $d+1 \leq i \leq p$, $v_i$ is a canonical basis vector with $1$ in position $i$. 

    $\Sigma$ has $d$ eigenvalues equal to $1+\Upsilon^{-1}$ and $p-d$ equal to $1$, hence

    \begin{equation*}
        \hat{\mu} = \Upsilon \delta_{1+\Upsilon^{-1}} + (1-\Upsilon) \delta_1
    \end{equation*}

    Next, let $W$ admit the following singular value decomposition

    \begin{align*}
        W = PDQ'
    \end{align*}

    with $P \in \R^{p \times d}$ and $Q \in \R^{d \times p}$ orthogonal and $D \in \R^{d \times p}$ such that

    \begin{align*}
        D_{ij} = \begin{cases}
            \Upsilon^{-1/2}, \quad &i=j \\
            0, \quad &i \neq j.
        \end{cases}
    \end{align*}

    Then,

    \begin{align*}
        \langle \beta, v_i \rangle^2 &= \beta' v_i v_i' \beta \\
        & = (1+\Upsilon^{-1})^{-2} \theta' W' v_i v_i W \theta \\
        & = (1+\Upsilon^{-1})^{-2} \theta' QD'P' v_i v_i' PDQ' \theta
    \end{align*}

    $P'v_i = v_i' P = (v_1' v_i, v_2' v_i, \dots, v_d'v_i)'$ which is equal to the canonical basis vector with $1$ in position $i$ if $i\leq d$ or to the null vector if $d+1 \leq i \leq p$. 

    Then, for $u \in \{is, oos \}$, we have 

    \begin{align*}
        \hat{\zeta}_u &= \frac{1}{\| \beta_u \|^2} \sum_{i=1}^p |\langle \beta_u, v_i \rangle |^2 \delta_{\lambda_i}\\
        &= \frac{1}{\| \beta_u \|^2} \Big( \sum_{i=1}^d |\langle \beta_u, v_i \rangle |^2 \delta_{1+\Upsilon^{-1}} + \sum_{i=d+1}^p |\langle \beta_u, v_i \rangle |^2 \delta_{1} \Big) \\
        & = \frac{1}{\| \beta_u \|^2} \beta_u' \big( \sum_{i=1}^d v_i v_i' \big) \beta_u  \delta_{1+\Upsilon^{-1}} \\
        & = \frac{1}{\| \beta_u \|^2} \beta_u ' \beta_u \delta_{1+\Upsilon^{-1}} \\
        & = \delta_{1+\Upsilon^{-1}}. 
    \end{align*}

    Similarly, $\hat{\zeta}_r = \delta_{1+\Upsilon^{-1}}$. Note that $\| \beta_u \|^2 = \Upsilon (1+\Upsilon)^{-2}\| \theta_u \|^2$. Hence,

    \begin{align*}
        \hat{\zeta}_d &= \frac{1}{2\|\beta_{is}\| \|\beta_{oos}\|} \big( \|\beta_{is}\|^2 \hat{\zeta}_{is} + \|\beta_{oos}\|^2 \hat{\zeta}_{oos} - \|\beta_{oos} - \beta_{is}\|^2 \hat{\zeta}_r \big) \\
        & =\frac{1}{2\|\theta_{is}\| \|\theta_{oos}\|} \big( \|\theta_{is}\|^2 + \|\theta_{oos}\|^2 - \|\theta_{oos} - \theta_{is}\|^2 \big) \delta_{1+\Upsilon^{-1}} \\
        & = \frac{\langle \theta_{is}, \theta_{oos} \rangle}{\|\theta_{is}\| \|\theta_{oos}\|} \delta_{1+\Upsilon^{-1}}
    \end{align*}

    From there, the result in the ridge regularization scenario follows from Proposition \ref{prop:cvprobdrift} and Definition \ref{def:mn} while the result for the ridgeless limit is a direct application of Proposition \ref{prop:cvprobdriftridgeless} and Definition \ref{def:s0}. 

    The proof of Proposition \ref{prop:latentspace} is complete. 
    
\end{proof}

\end{document}